\documentclass[a4paper,10pt]{article}
\usepackage{a4wide}
\usepackage{pgf, tikz}
\usepackage{graphicx}
\usepackage{subfigure}
\usepackage{url}

\usepackage{amsmath}
\usepackage[english]{babel}
\usepackage{microtype}
\usepackage{mathrsfs}
\usepackage{amsmath}
\usepackage{amssymb}
\usepackage{xspace}
\usepackage{color}
\usepackage{proof}

\usepackage{xy}
\xyoption{frame}
\xyoption{arrow}
\xyoption{curve}
\xyoption{arc}
\xyoption{line}
\xyoption{matrix}
\xyoption{tile}
\xyoption{ps}

\newcommand{\rn}{\mapsto_{\mathsf{n}}}

\newcommand{\ps}[2]{{#1}\oplus{#2}}
\newcommand{\val}{\mathsf{V}\LOP}

\newcommand{\ssemn}[2]{#1\Rightarrow_\mathsf{IN} #2}

\newcommand{\bsemn}[2]{#1\Downarrow_\mathsf{CN} #2}
\newcommand{\ibsemn}[2]{#1\Downarrow_\mathsf{IN} #2}

\newcommand{\sumd}[1]{\sum{#1}}

\newcommand{\emdist}{\emptyset}

\newcommand{\emcon}{\emptyset}

\newcommand{\RRp}[2]{\RRN_{[#1,#2]}}

\newcommand{\evlabel}{\tau}

\newcommand{\bnf}{\; ::=\;}

\newcommand{\btn}{\mathsf{bt}}
\newcommand{\bvn}{\mathsf{bv}}
\newcommand{\ban}{\mathsf{ba}}
\newcommand{\bsn}{\mathsf{bs}}

\newcommand{\midd}{\; \; \mbox{\Large{$\mid$}}\;\;}

\newcommand{\NN}{\mathbb{N}}

\newcommand{\RRN}{\mathbb{R}}

\newcommand{\PCF}{\ensuremath{\mathsf{PCF}}}

\newcommand{\LOP}{\ensuremath{\Lambda_{\oplus}}}
\newcommand{\LOPp}[1]{\ensuremath{\Lambda_{\oplus}}(#1)}

\newcommand{\pfone}{\pi}
\newcommand{\pftwo}{\rho}

\newcommand{\indsetone}{I}
\newcommand{\indsettwo}{J}
\newcommand{\indsetthree}{K}
\newcommand{\realone}{r}
\newcommand{\realtwo}{s}

\newcommand{\termone}{M}
\newcommand{\termtwo}{N}
\newcommand{\termthree}{L}
\newcommand{\termfour}{P}
\newcommand{\termfive}{Q}
\newcommand{\termsix}{R}
\newcommand{\termseven}{S}

\newcommand{\termnine}{U}

\newcommand{\varone}{x}
\newcommand{\vartwo}{y}
\newcommand{\varthree}{z}

\newcommand{\vecvarone}{\overline{\varone}}
\newcommand{\vecvartwo}{\overline{\vartwo}}
\newcommand{\valone}{V}
\newcommand{\valtwo}{W}
\newcommand{\valthree}{X}

\newcommand{\vectermthree}{\overline{\termthree}}

\newcommand{\distone}{\mathscr{D}}
\newcommand{\disttwo}{\mathscr{E}}
\newcommand{\distthree}{\mathscr{F}}
\newcommand{\distfour}{\mathscr{G}}
\newcommand{\distfive}{\mathscr{H}}

\newcommand{\setone}{X}
\newcommand{\settwo}{Y}

\newcommand{\probone}{p}
\newcommand{\probtwo}{q}
\newcommand{\probthree}{r}
\newcommand{\probfour}{s}

\newcommand{\elone}{a}
\newcommand{\eltwo}{b}
\newcommand{\natone}{n}

\newcommand{\fsone}{\mathbf{S}}
\newcommand{\fstwo}{\mathbf{T}}
\newcommand{\fsthree}{\mathbf{U}}
\newcommand{\fsfour}{\mathbf{V}}
\newcommand{\nil}{\mathtt{nil}}
\newcommand{\las}[2]{[\cdot]{#1}::{#2}}

\newcommand{\stk}[1]{\mathcal{FS}(#1)}
\newcommand{\stktm}[2]{E_{#1}(#2)}

\newcommand{\cbnfsred}{\leadsto_{\mathsf{n}}}

\newcommand{\cbnfscp}[3]{(#1,#2)\downarrow_{\mathsf{n}}^{#3}}
\newcommand{\cbnfssup}[2]{\mathbb{C}(#1,#2)}
\newcommand{\cbnciuleq}{\preceq^{\mathsf{CIU}}}
\newcommand{\cbnconleq}{\mathord{\leq_\oplus}}
\newcommand{\cbnconequiv}{\mathord{\simeq_\oplus}} 
\newcommand{\cbnsimleq}{\pas}
\newcommand{\cbnciuequiv}{\cong^{\mathsf{CIU}}}

\newcommand{\howe}[1]{{#1}^{H}}

\newcommand{\subst}[3]{#1 \sub #3 #2}
\newcommand{\substn}[5]{#1 \sub{#3}{#2} \ldots \sub{#5}{#4}}

\newcommand{\setvar}{\mathsf{X}}

\newcommand{\abstr}[2]{\lambda #1.#2}
\newcommand{\clabstr}[2]{\mathord{\nu}#1.#2}
\newcommand{\app}[2]{#1#2}
\newcommand{\pair}[2]{\langle #1,#2\rangle}

\newcommand{\supp}[1]{\mathsf{S}(#1)}

\newenvironment{varitemize}
{
\begin{list}{\labelitemi}
{\setlength{\itemsep}{0pt}
 \setlength{\topsep}{0pt}
 \setlength{\parsep}{0pt}
 \setlength{\partopsep}{0pt}
 \setlength{\leftmargin}{15pt}
 \setlength{\rightmargin}{0pt}
 \setlength{\itemindent}{0pt}
 \setlength{\labelsep}{5pt}
 \setlength{\labelwidth}{10pt}
}}
{
 \end{list} 
}

\newcounter{numberone}

\newenvironment{varenumerate}
{
\begin{list}{\arabic{numberone}.}
{
  \usecounter{numberone}
  \setlength{\itemsep}{0pt}
  \setlength{\topsep}{0pt}
  \setlength{\parsep}{0pt}
  \setlength{\partopsep}{0pt}
  \setlength{\leftmargin}{15pt}
  \setlength{\rightmargin}{0pt}
  \setlength{\itemindent}{0pt}
  \setlength{\labelsep}{5pt}
  \setlength{\labelwidth}{15pt}
}}
{
\end{list} 
}

\newcounter{numbertwo}

\newcommand{\ev}{\Downarrow}

\newcommand{\FV}[1]{\mathsf{FV}(#1)}

\newcommand{\id}{\mathsf{ID}}

\newcommand{\pdists}{\mathcal{P}}

\newcommand{\sem}[1]{\mathcal{S}(#1)}
\newcommand{\sumsem}[1]{\sum\sem{#1}}

\newcommand{\pow}[1]{\mathcal{P}(#1)}
\newcommand{\powfin}[1]{\mathcal{P}_{\mathsf{FIN}}(#1)}
\newcommand{\card}[1]{\left|{#1}\right|}

\newcommand{\statesone}{\mathcal{S}}
\newcommand{\labelsone}{\mathcal{L}}
\newcommand{\labelone}{\ell}
\newcommand{\stateone}{s}
\newcommand{\statetwo}{t}
\newcommand{\statethree}{v}
\newcommand{\transone}{\mathcal{P}}
\newcommand{\translop}{\mathcal{P}_{\oplus}}

\newcommand{\relone}{\mathcal{R}}
\newcommand{\reltwo}{\mathcal{T}}

\newcommand{\quot}[2]{#1/\mathord{#2}}

\newcommand{\econe}{\mathtt{E}}
\newcommand{\ectwo}{\mathtt{F}}
\newcommand{\xsubset}{X}

\newenvironment{proof}{\begin{trivlist}
       \item[\hskip \labelsep {\bfseries Proof.}]}{\hfill $\Box$ \end{trivlist}}

\newcommand{\st}{\;|\;}

\newcommand{\rel}[4]{#1\vdash #2\;#3\;#4}
\newcommand{\relu}[3]{#1\;#2\;#3}

\newcommand{\pab}{\sim}
\newcommand{\pabn}{\cbn{\textnormal{PAB}}}

\newcommand{\pas}{\lesssim}
\newcommand{\pasn}{\cbn{\textnormal{PAS}}}

\newcommand{\tcrel}[1]{#1^+}

\newcommand{\cbn}[1]{#1}
\newcommand{\cbnpab}{\cbn{\pab}}
\newcommand{\cbnpas}{\cbn{\pas}}

\newcommand{\cbv}[1]{{#1}_{\mathsf{v}}}
\newcommand{\cbvpab}{\cbv{\pab}}
\newcommand{\cbvpas}{\cbv{\pas}}

\newcommand{\Comone}{\ensuremath{(\mathsf{Com1})}}
\newcommand{\Comtwo}{\ensuremath{(\mathsf{Com2})}}
\newcommand{\Comthree}{\ensuremath{(\mathsf{Com3})}}
\newcommand{\ComthreeL}{\ensuremath{(\mathsf{Com3L})}}
\newcommand{\ComthreeR}{\ensuremath{(\mathsf{Com3R})}}
\newcommand{\Comfour}{\ensuremath{(\mathsf{Com4})}}
\newcommand{\ComfourL}{\ensuremath{(\mathsf{Com4L})}}
\newcommand{\ComfourR}{\ensuremath{(\mathsf{Com4R})}}

\newcommand{\Howeone}{\ensuremath{(\mathsf{How1})}}
\newcommand{\Howetwo}{\ensuremath{(\mathsf{How2})}}
\newcommand{\Howethree}{\ensuremath{(\mathsf{How3})}}
\newcommand{\Howefour}{\ensuremath{(\mathsf{How4})}}

\newcommand{\Howestkone}{\ensuremath{(\mathsf{Howstk1})}}
\newcommand{\Howestktwo}{\ensuremath{(\mathsf{Howstk2})}}

\newcommand{\Howeredone}{\ensuremath{(\mathsf{empty})}}
\newcommand{\Howeredtwo}{\ensuremath{(\mathsf{value})}}
\newcommand{\Howeredthree}{\ensuremath{(\mathsf{term})}}

\newcommand{\Ctxone}{\ensuremath{(\mathsf{Ctx1})}}
\newcommand{\Ctxtwo}{\ensuremath{(\mathsf{Ctx2})}}
\newcommand{\Ctxthree}{\ensuremath{(\mathsf{Ctx3})}}
\newcommand{\Ctxfour}{\ensuremath{(\mathsf{Ctx4})}}
\newcommand{\Ctxfive}{\ensuremath{(\mathsf{Ctx5})}}
\newcommand{\Ctxsix}{\ensuremath{(\mathsf{Ctx6})}}

\newcommand{\TCone}{\ensuremath{(\mathsf{TC1})}}
\newcommand{\TCtwo}{\ensuremath{(\mathsf{TC2})}}

\newcommand{\ctxset}{\mathsf{C}\LOP}
\newcommand{\ctxsetp}[2]{\mathsf{C}\LOP(#1\,;\,#2)}

\newcommand{\ctxone}{C}
\newcommand{\ctxtwo}{D}
\newcommand{\ctxthree}{E}

\newcommand{\ctxhole}[1]{[#1]}

\newcommand{\evp}[1]{\mathord{\ev_{#1}}}
\newcommand{\caset}{\mathbb{CA}}
\newcommand{\cbncfleq}{\mathsf{\leq_{\oplus}^{ca}}}

\newcommand{\passone}{\mathbf{P}}

\newcommand{\nt}[1]{\mathcal{N}_{#1}}
\newcommand{\vt}[1]{\mathcal{V}_{#1}}
\newcommand{\ed}[1]{\mathcal{E}_{#1}}
\newcommand{\fnet}[1]{\nt{#1}\defi(\vt{#1},\ed{#1})}
\newcommand{\ce}[1]{c(#1)}

\newcommand{\flone}{f}

\newcommand{\edone}{e}

\newcommand{\ratioone}{q}

\newcommand{\cutsone}{S}
\newcommand{\cutstwo}{T}
\newcommand{\cuttone}{A}
\newcommand{\cutttwo}{B}

\newcommand{\cutone}{C}
\newcommand{\cuttwo}{D}
\newcommand{\cutthree}{E}

\newcommand{\ccut}[1]{C_#1}

\newcommand{\arr}[1]{\mathrel{\stackrel{{\;#1\;}}{\mbox{\rightarrowfill}}}}

\newcommand{\DSLongrightarrow}{ 
\mathrel{\;\mbox{\raisebox{.7ex}{$\!\leadsto$}{$\!\!\!\!\!\!\leadsto$}}\;}}

\newcommand{\DwaDO}[1]{\mathord{\downdownarrows_{#1}}}  

\newcommand{\LLT}{\mathrel{=_{{\rm{LL}}}}}

\newcommand{\logbis}{\approx}

\newcommand{\conteqDO}{\mathrel{\simeq_\oplus^{{\rm{FS}}}}}

\newcommand{\conteqDOeval}{\mathrel{\approxeq_\oplus^{{\rm{FS}}}}}

\newcommand{\DSlongrightarrow}{\leadsto}
\newcommand{\ULdwa}{\Longrightarrow}

\newcommand{\coupledrel}{coupled relation}

\newcommand{\myoplus}{\oplus}

\newcommand{\LaDO}{\Lambda_\oplus^{{\rm{FS}}}}

\newcommand{\LaPRO}{\LOP}

\newcommand{\dist}[1]{\mathcal{D}(#1)}
\newcommand{\distQ}{\mathcal{D}}

\newcommand{\overl}[1]{#1}

\newcommand{\prob}[1]{{\boldsymbol\Sigma}({#1})}

\newcommand{\starred}[1]{\mathrel{#1^{\rm{C}}}}
\newcommand{\starredD}[1]{\mathrel{#1^{\rm{C FS}} }}

\newcommand{\Prod}[1]{\Sigma_{#1}}
\newcommand{\Sum}{\sum}

 \def\E{{\cal E}}

\def\V{{\cal V}}
\def\VV{\mathrel{\cal V}}
\def\W{{\cal W}}

\newcommand{\til}[1]{\overline{#1}}
\newcommand{\Om}{\Omega}

\newcommand{\myfrac}[2]{\frac{#1}{#2}} 

\def\st{\; {\mbox{s.t.}} \;}

\def\SS{\mathrel{\cal S}}

\def\R{{\cal R}}
\def\RR{\mathrel{\cal R}}

\def\defi{\mathrel{\stackrel{\mbox{\scriptsize {\rm{def}}}}{=}}} 

\def\sub#1#2{\{\raisebox{.5ex}{\small$#1$}\! / \mbox{\small$#2$}\}}

\def\midd{\; \; \mbox{\Large{$\mid$}}\;\;}

\newcommand{\longrightarrowP}[1]{\longrightarrow_{#1}}     

\newcommand{\LongrightarrowP}[1]{\Longrightarrow_{#1}}     

\def\trans#1{{\mathsf{#1}}}   

\newcommand{\Dwa}{\Downarrow}           
\newcommand{\Up}{\Uparrow}           

\newcommand{\DwaP}[1]{\Dwa_{#1}}     

\newcommand{\brac}[1]{[#1] }   

\newcommand{\contexthole}{ [ \cdot  ] }      

\newcommand{\ct}[1]{ C \brac{#1} }   
\newcommand{\qct}{ C  }              

\newcommand{\simO}{\mathrel{\sim^{\rm O}}}

\newcommand{\simOn}[1]{\mathrel{\sim^{\rm O}_{#1}}}

\renewcommand{\sem}[1]{[\![ #1]\!]}
\renewcommand{\bsemn}[2]{#1\Downarrow #2}
\renewcommand{\ssemn}[2]{#1\Rightarrow #2} 
\renewcommand{\ibsemn}[2]{\bsemn{#1}{#2}}
\renewcommand{\PCF}{\ensuremath{\mathsf{PCF}}}

\newtheorem{definition}{Definition}[section] 
 
\newtheorem{example}[definition]{Example} 
\newtheorem{lemma}[definition]{Lemma} 
\newtheorem{theorem}[definition]{Theorem} 
\newtheorem{corollary}[definition]{Corollary} 
\newtheorem{proposition}[definition]{Proposition}  
\newtheorem{remark}[definition]{Remark}

\begin{document}

\title{On Coinductive
  Equivalences\\ for Higher-Order Probabilistic Functional Programs} 
\author{Ugo Dal Lago \and Davide Sangiorgi \and Michele Alberti}
\date{}
\maketitle

\newcommand{\hsk}[1]{\emph{\texttt{#1}}}

\begin{abstract}
We study bisimulation and context equivalence in a probabilistic
$\lambda$-calculus. The contributions of this paper are
threefold.  Firstly we show a technique for proving congruence of
\emph{probabilistic applicative bisimilarity}. While the technique
follows Howe's method, some of the technicalities are quite different,
relying on non-trivial ``disentangling'' properties for sets of real
numbers. Secondly we show that, while bisimilarity is in general 
strictly finer than context equivalence, coincidence between the 
two relations is attained on pure $\lambda$-terms. The resulting 
equality is that induced by \emph{Levy-Longo trees}, generally accepted as 
the finest extensional equivalence on pure $\lambda$-terms under 
a lazy regime.  Finally, we derive a coinductive characterisation 
of context equivalence on the whole probabilistic language, via 
an extension in which terms akin to distributions may appear in redex position.  
Another motivation for the extension is that its operational semantics allows
us to experiment with a different congruence technique, namely that of
\emph{logical bisimilarity}.
\end{abstract}

\section{Introduction}
Probabilistic models are more and more pervasive. Not only are they a formidable tool
when dealing with uncertainty and incomplete information, but they sometimes are 
a \emph{necessity} rather than an option, like in computational cryptography (where, e.g.,
secure public key encryption schemes need to be
probabilistic~\cite{GoldwasserMicali}). 
A nice way to deal computationally with probabilistic models  is to allow probabilistic choice as 
a primitive when designing algorithms, this way switching from usual, deterministic computation 
to a new paradigm, called probabilistic computation. Examples of application areas in which 
probabilistic computation has proved to be useful include natural language processing~\cite{manning1999foundations}, 
robotics~\cite{thrun2002robotic}, computer vision~\cite{comaniciu2003kernel}, and machine learning~\cite{pearl1988probabilistic}.

This new form of computation, of course, needs to be available to
programmers to be accessible.  And indeed, various probabilistic programming languages
have been introduced in the last years, spanning from abstract ones~\cite{Plotkin,RamseyPfeffer,ParkPfenningThrun}
to more concrete ones~\cite{Pfeffer01,Goodman}, being
inspired by various programming paradigms like imperative, functional or even object oriented. 
A quite common scheme consists in endowing any deterministic language with
one or more primitives for probabilistic choice, like binary probabilistic
choice or primitives for distributions.

One class of languages that copes well with probabilistic computation are functional languages. Indeed,
viewing algorithms as functions allows a smooth integration of distributions into the playground, itself
nicely reflected at the level of types through monads~\cite{GordonABCGNRR13,RamseyPfeffer}. As a matter of fact, many existing probabilistic
programming languages~\cite{Pfeffer01,Goodman} are designed around the $\lambda$-calculus or one of its incarnations, like
\textsf{Scheme}. All these allows to write higher-order functions (i.e., programs can
take functions as inputs and produce them as outputs). 

The focus of this paper are operational techniques for understanding and reasoning about program equality in
higher-order probabilistic languages. Checking computer programs for equivalence is a crucial, but challenging,
problem. Equivalence between two programs generally means that the programs should  behave ``in the same manner'' 
under any context~\cite{Morris-68}. Specifically, two $\lambda$-terms are \emph{context equivalent} if they 
have the same convergence behavior (i.e., they do or do not terminate) in any possible context. Finding 
effective methods for context equivalence proofs is particularly challenging in higher-order languages.

\emph{Bisimulation} has emerged as a very powerful operational method for proving equivalence of programs 
in various kinds of languages, due to the associated coinductive proof method. To be useful, the behavioral 
relation resulting from bisimulation --- \emph{bisimilarity} --- should be a \emph{congruence}, and should
also be sound with respect to context equivalence. Bisimulation has been transplanted onto higher-order languages  
by Abramsky~\cite{Abramsky-90}. This version of bisimulation, called \emph{applicative bisimulation} has received considerable
attention~\cite{Gordon-92,PittsSurvey,Sands98,Las98a,Pit97,LengletSS09}. In short, two functions $M$ and $N$ are applicative bisimilar when
their applications $MP$ and $NP$ are applicative bisimilar for any argument $P$.

Often, checking a given notion of bisimulation to be a congruence in higher-order languages is nontrivial. 
In the case of applicative bisimilarity, congruence proofs usually rely on Howe's method~\cite{Howe-96}. 
Other forms of bisimulation have been proposed, such as environmental bisimulation and logical 
bisimulation \cite{SangiorgiKS07fsen,SangiorgiKS11,KoutavasLS11}, with the goal of
relieving the burden of the proof of congruence, and of accommodating
language extensions.

In this work, we consider the pure $\lambda$-calculus extended with a probabilistic choice operator.
Context equivalence of two terms means that they have the same \emph{probability of convergence} in all contexts. 
The objective of the paper is to understand context equivalence and bisimulation in this 
paradigmatic  probabilistic higher-order language, called $\LOP$. 

The paper contains three main technical contributions.  The first  is a
proof of congruence for probabilistic applicative bisimilarity along the
lines of Howe's method. This technique consists in defining, for every
relation on terms $\relone$, its Howe's lifting $\howe{\relone}$. The
construction, essentially by definition, ensures that the relation obtained
by lifting bisimilarity is a congruence; the latter is then proved to be
itself a bisimulation, therefore coinciding with applicative
bisimilarity. Definitionally, probabilistic applicative bisimulation is
obtained by setting up a labelled Markov chain on top of $\lambda$-terms,
then adapting to it the coinductive scheme introduced by Larsen and Skou in
a first-order setting \cite{LarsenSkou}. In the proof of congruence, the
construction $\howe{(\cdot)}$ closely reflects analogous constructions for
nondeterministic extensions of the $\lambda$-calculus.  The novelties are
in the technical details for proving that the resulting relation is a
bisimulation: in particular our proof of the so-called Key Lemma --- an essential
ingredient in Howe's method --- relies on non-trivial ``disentangling''
properties for sets of real numbers, these properties themselves proved by
modeling the problem as a flow network and then apply the Max-flow Min-cut
Theorem.  The congruence of applicative bisimilarity yields soundness
with respect to context equivalence as
an easy corollary.  Completeness, however, fails: applicative bisimilarity
is proved to be finer.  A subtle aspect is also the late vs. early
formulation of  bisimilarity; with a choice operator the two
versions are semantically different; our construction crucially relies on
the late style.

In our second main technical contribution we show that the presence of
higher-order functions and probabilistic choice in contexts gives context
equivalence and applicative bisimilarity maximal discriminating power
on pure $\lambda$-terms.  We do so by proving that, on pure
$\lambda$-terms, both context equivalence and applicative bisimilarity
coincide with the \emph{Levy-Longo tree equality}, which equates terms
with the same Levy-Longo tree (briefly 
LLT). The LLT equality is generally accepted as the finest extensional equivalence on pure
$\lambda$-terms under a \emph{lazy} regime.
The result is in sharp contrast with what happens under a
nondeterministic interpretation of choice (or in the absence of choice), 
where context equivalence is coarser than LLT equality. 

Our third main contribution is a coinductive characterisation of
probabilistic context equivalence on the whole language $\LaPRO$
(as opposed to the subset of pure $\lambda$-terms).  We obtain this
result by setting a bisimulation game on an extension of $\LaPRO$ in which
weighted formal sums~--- 
terms akin to distributions~---
 may appear in redex position.  Thinking of
distributions as sets of terms, the construction reminds us of the
reduction of nondeterministic to deterministic automata. The
technical details are however quite different, because we are in a
higher-order language and therefore --- once more --- we are faced
with the congruence problem for bisimulation, and because
formal sums may contain an \emph{infinite} number of terms.
For the proof of congruence of bisimulation  in this extended language, 
we have experimented the technique of
logical bisimulation. In this method (and in the related method of
environmental bisimulation), the clauses of applicative bisimulation
are modified so to allow the standard congruence argument for
bisimulations in first-order languages, where the bisimulation method
itself is exploited to establish that the closure of the bisimilarity
under contexts is again a bisimulation.  Logical bisimilarities have
two key elements. First, bisimilar functions may be tested with
\emph{bisimilar} (rather than identical) arguments (more precisely,
the arguments should be in the context closure of the bisimulation;
the use of contexts is necessary for soundness).  Secondly, the
transition system should be small-step, deterministic (or at least
confluent), and  the bisimulation game should also be played on
internal moves. In our probabilistic setting, the ordinary logical bisimulation game has to be
modified substantially. Formal sums represent possible evolutions
of running terms, hence they should appear in redex
position only (allowing them anywhere would complicate matters
considerably), also making the resulting bisimulation proof technique
more cumbersome). The obligation of redex position for certain terms
is in contrast with the basic schema of logical bisimulation, in which
related terms can be used as arguments to bisimilar functions and can
therefore end up in arbitrary positions.  We solve this problem by
moving to \emph{coupled logical bisimulations}, where a
bisimulation is formed by a pair of relations, one on
$\LOP$-terms, the other on terms extended with
formal sums.  The bisimulation game is played on both relations, but
only the first relation is used to assemble input arguments for
functions.

Another delicate point is the meaning of internal transitions for
formal sums.  In logical bisimilarity the transition system should
be small-step; and formal sums should evolve into values in a finite
number of steps, even if the number of terms composing the
formal sum is infinite. We satisfy these requirements by defining
the transition system for extended terms on top of that of
$\LaPRO$-terms.
The proof of congruence of coupled logical  bisimilarity also exploits
an ``up-to distribution'' bisimulation proof technique. 

In the paper we adopt call-by-name evaluation. The results on applicative bisimilarity
can be transported onto  call-by-value; in contrast, transporting  the other
results is less clear, and we leave it for future work. See Section~\ref{sect:beyond}
for more details. An extended version of this paper with more details is available~\cite{EV}.
\subsection{Further Related Work}
Research on (higher-order) probabilistic functional languages have, so far, mainly focused
on either new programming constructs, or denotational semantics, or applications. The underlying 
operational theory, which in the ordinary $\lambda$-calculus is known to be very rich, has  remained
so far largely unexplored. In this section, we give some pointers to the relevant literature on
probabilistic $\lambda$-calculi, without any hope of being exhaustive. 

Various probabilistic $\lambda$-calculi have been proposed, starting from
the pioneering work by Saheb-Djahromi~\cite{Djahromi78}, followed by more
advanced studies by Jones and Plotkin~\cite{Plotkin}. Both these works are
mainly focused on the problem of giving a denotational semantics to
higher-order probabilistic computation, rather than on studying it from an
operational point view.  More recently,
there has been a revamp on this line of work, with the introduction of
adequate (and sometimes also fully-abstract) denotational models for
probabilistic variations of
\PCF~\cite{DanosHarmer,EhrhardPaganiTasson}. There is also another thread
of research in which various languages derived from the $\lambda$-calculus
are given types in monadic style, allowing this way to nicely model
concrete problems like Bayesian inference and probability models arising in
robotics~\cite{RamseyPfeffer,ParkPfenningThrun,GordonABCGNRR13}; these
works however, do not attack the problem of giving an operationally based
theory of program equivalence.

Nondeterministic extensions of the $\lambda$-calculus have been analysed
in typed calculi~\cite{AstesianoCosta84,Sieber93,Las98a} as well as
in untyped calculi~\cite{JagadeesanPanangaden90,Boudol94,Ong93,deLiguoroPiperno95}. 
The emphasis in all these works is
mainly domain-theoretic.  Apart from~\cite{Ong93}, all cited authors
closely follow the testing theory~\cite{NicolaHennessy84}, in its
modalities \emph{may} or \emph{must}, separately or together. 
Ong's approach~\cite{Ong93} 
inherits both testing and bisimulation elements.

Our definition of applicative bisimulation follows
Larsen and Skou's scheme~\cite{LarsenSkou} for fully-probabilistic
systems. Many other forms of probabilistic bisimulation  have been introduced
in the literature, but their greater complexity is usually due to the presence 
of \emph{both} nondeterministic and probabilistic
behaviors, or to  continuous probability distributions.
See surveys such as \cite{BernardoNL13,Panangaden09,Hennessy12}.

Contextual characterisations of LLT equality include \cite{BoLa94},
in a $\lambda$-calculus with multiplicities in which deadlock is
observable, and \cite{DezTU99}, in  a
$\lambda$-calculus with choice, parallel composition, and
both call-by-name and call-by-value applications. 
The characterisation in \cite{San94sce} in a $\lambda$-calculus with non-deterministic
operators, in contrast, is not contextual, as derived from a
bisimulation that includes a clause on internal  moves so to observe
branching  structures in behaviours. See \cite{DezG01}
for a survey on  observational characterisations of
$\lambda$-calculus trees. 

\section{Preliminaries}\label{sect:p}

\subsection{A Pure, Untyped, Probabilistic Lambda Calculus}
Let $\setvar=\{\varone,\vartwo,\ldots\}$ be a denumerable set of variables.
The set $\LOP$ of \emph{term expressions}, or \emph{terms} is defined as follows:
$$
\termone,\termtwo,\termthree\bnf\varone\midd\abstr{\varone}{\termone}\midd\app{\termone}{\termtwo}\midd\ps{\termone}{\termtwo},
$$
where $\varone\in\setvar$. The only non-standard operator in $\LOP$ is probabilistic
choice: $\ps{\termone}{\termtwo}$ is a term which is meant to behave as either $\termone$
or $\termtwo$, each with probability $\frac{1}{2}$. A more general construct $\termone\oplus_p\termtwo$ where
$p$ is any (computable) real number from $[0,1]$, is derivable, given the universality of the $\lambda$-calculus (see, e.g., \cite{DalLagoZorzi}).
The set of free variables of a term $\termone$ is indicated as $\FV{\termone}$ and is
defined as usual. Given a finite set of variables $\vecvarone\subseteq\setvar$, $\LOP(\vecvarone)$ denotes the set of
terms whose free variables are among the ones in $\vecvarone$. A term $\termone$ is \emph{closed} if 
$\FV{\termone}=\emptyset$ or, equivalently, if
$\termone\in\LOP(\emptyset)$. 
The (capture-avoiding) substitution of $\termtwo$ for the free occurrences
of $\varone$ in $\termone$ is denoted
$\subst{\termone}{\varone}{\termtwo}$.  We  sometimes use the identity
term $I\defi \lambda x.x$, the projector $K\defi\lambda x . \lambda y. x$,
and the purely divergent term $\Omega\defi (\lambda x. x x)(\lambda x. x x)$.

Terms  are now given a call-by-name semantics following~\cite{DalLagoZorzi}.
A term is a \emph{value} if it is a closed $\lambda$-abstraction. We  call $\val$ 
the set of all values. Values are ranged over by metavariables like $\valone,\valtwo,\valthree$.
Closed terms evaluates not to a single value, but to a \emph{(partial) value distribution}, that is,  a function 
$\distone:\val\rightarrow\RRp{0}{1}$ such that $\sum_{\valone\in\val} \distone(\valone)\leq 1$.
The set of all value distributions is $\cbv{\pdists}$.
Distributions do not necessarily sum to $1$, so to 
model the possibility of (probabilistic) divergence.  Given a value distribution $\distone$,
its \emph{support} $\supp{\distone}$ is the subset of $\val$ whose elements
are values to which $\distone$ attributes positive probability. Value
distributions ordered pointwise form both a lower semilattice and an
$\omega\mathbf{CPO}$: limits of $\omega$-chains always exist.  Given a
value distribution $\distone$, its \emph{sum} $\sumd{\distone}$ is 
$\sum_{\valone\in\val} \distone(\valone)$.

The call-by-name semantics of a closed term $\termone$ is a value distribution 
$\sem{\termone}$ defined in one of the ways explained in~\cite{DalLagoZorzi}.
We  recall this now, though only briefly for lack of space.
The first step consists in defining a formal system deriving
finite \emph{lower approximations} to the semantics of $\termone$.
Big-step approximation semantics, as an example, derives
judgments in the form $\bsemn{\termone}{\distone}$, where $\termone$ is a term
and $\distone$ is a value distribution of finite support (see Figure~\ref{fig:bscbnsem}). 
\begin{figure*}
$$
\infer
    [\btn]
    {\bsemn{\termone}{\emdist}}
    {}
\qquad
\infer
    [\bvn]
    {\bsemn{\valone}{\{\valone\}}}
    {}
\qquad
\infer[\ban] 
      {\bsemn{\app{\termone}{\termtwo}}{\sum_{\abstr{\varone}{\termfour}\in\supp{\distone}}\distone(\abstr{\varone}{\termfour})\cdot\disttwo_{\termfour,\termtwo}}}
      {\bsemn{\termone}{\distone} & \{\bsemn{\subst{\termfour}{\varone}{\termtwo}}{\disttwo_{\termfour,\termtwo}}\}_{\abstr{\varone}{\termfour}\in\supp{\distone}}}
\qquad
\infer[\bsn]
      {\bsemn{\ps{\termone}{\termtwo}}{\frac{1}{2}\cdot\distone+\frac{1}{2}\cdot\disttwo}}
      {\bsemn{\termone}{\distone} & \bsemn{\termtwo}{\disttwo}}
$$
\caption{Big-step call-by-name approximation semantics for $\LOP$.}\label{fig:bscbnsem}
\end{figure*}
Small-step approximation semantics can be defined similarly, and derives
judgments in the form $\ssemn{\termone}{\distone}$. Noticeably, big-step and
small-step can simulate each other, i.e. if $\bsemn{\termone}{\distone}$, then
$\ssemn{\termone}{\disttwo}$ where $\disttwo\geq\distone$, and
\emph{vice versa}~\cite{DalLagoZorzi}. In the second
step, $\sem{\termone}$, called the \emph{semantics} of $\termone$, is
set as the least upper 
bound of distributions obtained in either of the two ways: 
$$
\sem{\termone}\defi\sup_{\bsemn{\termone}{\distone}}\distone=\sup_{\ssemn{\termone}{\distone}}\distone.
$$
Notice that the above is well-defined because for every $\termone$, the set of all distributions
$\distone$ such that $\bsemn{\termone}{\distone}$ is directed, and thus its least upper bound is
a value distribution because of $\omega$-completeness.

\begin{example}
  Consider the term $\termone\defi\ps{I}{(\ps{K}{\Omega})}$. We have
 $\bsemn{\termone}{\distone}$, where $\distone(I)=\frac{1}{2}$
  and $\distone(\valone)$ is $0$ elsewhere, as well as
  $\bsemn{\termone}{\emdist}$, where $\emdist$ is the empty
  distribution. The distribution $\sem{\termone}$ assigns
  $\frac{1}{2}$ to $I$ and $\frac{1}{4}$ to $K$.
\end{example}
The semantics of terms satisfies some useful equations, such as:
\begin{lemma}\label{lemma:sembetaCBN}
  $\sem{\app{(\abstr{\varone}{\termone})}{\termtwo}}=\sem{\subst{\termone}{\varone}{\termtwo}}$.
\end{lemma}
\begin{lemma}\label{lemma:semsumCBN}
  $\sem{\ps{\termone}{\termtwo}}=\frac{1}{2}\sem{\termone}+\frac{1}{2}\sem{\termtwo}$.
\end{lemma}
\begin{proof}
  See~\cite{DalLagoZorzi} for detailed proofs.
\end{proof}
We are interested in context equivalence in this probabilistic
setting. Typically, in a qualitative scenario as the (non)deterministic
one, terms are considered context equivalent if they both converge or
diverge. Here, we need to take into account quantitative information.
\begin{definition}[Context Preorder and Equivalence]
The expression $\termone\evp{\probone}$ stands for $\sumsem{\termone}=\probone$, i.e.,
the term $\termone$ converges with probability $\probone$. 
The \emph{context preorder}  $\cbnconleq$ stipulates 
$\termone \cbnconleq \termtwo$ if 
$\ctxone\ctxhole{\termone}\evp{\probone}$ implies
$\ctxone\ctxhole{\termtwo}\evp{\probtwo}$ with
$\probone\leq\probtwo$, for every closing context $\ctxone$. 
The equivalence induced by   $\cbnconleq$ is
 \emph{probabilistic context equivalence}, denoted as
$\cbnconequiv $.
\end{definition}

\begin{remark}[Types, Open Terms]
\label{r:to} 
The results in this paper are stated for an untyped language.  Adapting
them to a simply-typed language is straightforward; we  use integers,
booleans and recursion in examples.  Moreover, while the results are often stated for closed
terms only, they can be generalized to open terms in the expected manner.
In the paper, context equivalences and preorders are  defined on
open terms; (bi)similarities are defined on closed terms and it is then
intended that they are extended to open terms by requiring the usual
closure under substitutions.
\end{remark}

\begin{example}\label{ex:exp}
We give some basic examples of higher-order probabilistic programs, which we will
analyse using the coinductive techniques we introduce later in this paper. Consider the 
functions \hsk{expone}, \hsk{exptwo}, and \hsk{expthree} from Figure~\ref{fig:examples}.
They are written in a \emph{\textsf{Haskell}}-like language extended with probabilistic choice, 
but can also be seen as terms in a (typed) probabilistic $\lambda$-calculus with
integers and recursion akin to $\LOP$.
\begin{figure}
{\footnotesize
\begin{verbatim}
expone f n = (f n) (+) (expone f n+1)
exptwo f = (\x -> f x) (+) (exptwo (\x -> f (x+1)))
expthree k f n = foldp k n f (expthree (expone id k) f)

foldp 0 n f g = g n
foldp m n f g = (f n) (+) (foldp (m-1) (n+1) f g)
\end{verbatim}}
\caption{Three Higher-Order Functions}\label{fig:examples}
\end{figure}
Term \hsk{expone}  takes a function \hsk{f} and a natural number \hsk{n}
in input, then it proceeds by tossing a fair coin (captured here by the binary infix operator
\hsk{(+)}) and, depending on the outcome of the toss, either calls \hsk{f} on \hsk{n}, 
or recursively calls itself on \hsk{f} and \hsk{n+1}. When fed with, e.g., the identity
and the natural number \hsk{1}, the program \hsk{expone} evaluates to the geometric distribution
assigning probability $\frac{1}{2^{n}}$ to any positive natural number $n$. A similar effect
can be obtained by \hsk{exptwo}, which only takes \hsk{f} in input, then ``modifying'' it along
the evaluation. The function \hsk{expthree} is more complicated, at least apparently. To understand
its behavior, one should first look at the auxiliary function \hsk{foldp}. If
\hsk{m} and \hsk{n} are two natural numbers and \hsk{f} and \hsk{g} are two functions,
\hsk{foldp m n f g} call-by-name reduces to the following expression:

{\footnotesize
\vspace{-8pt}
$$
\hsk{(f n) (+) ((f n+1) (+) \ldots\ ((f n+m-1) (+) (g n+m)))}.
$$}
The term \hsk{expthree} works by forwarding its three arguments to
\hsk{foldp}. The fourth argument is a recursive call to \hsk{expthree} where, however,
\hsk{k} is replaced by any number greater or equal to it, chosen according to a 
geometric distribution. The functions above can all be expressed
in
$\LOP$, using fixed-point combinators. As we will see soon,
\hsk{expone}, \hsk{exptwo}, and \hsk{expthree k} are context
equivalent whenever \hsk{k} is a natural number.
\end{example}

\subsection{Probabilistic Bisimulation}
In this section we recall the definition and a few basic notions of
bisimulation for labelled Markov chains, following Larsen and
Skou~\cite{LarsenSkou}. In Section~\ref{sect:pasht} we will then adapt this
form of bisimilarity to the probabilistic $\lambda$-calculus $\LOP$ by combining
it with Abramsky's applicative bisimilarity.
\begin{definition}
A \emph{labelled Markov chain} is a triple $(\statesone,\labelsone,\transone)$
such that:
\begin{varitemize}
\item
  $\statesone$ is a countable set of \emph{states};
\item
  $\labelsone$ is set of \emph{labels};
\item
  $\transone$ is a \emph{transition probability matrix}, i.e. a function
  $$
  \transone:\statesone\times\labelsone\times\statesone\rightarrow\RRp{0}{1}
  $$
  such that the following normalization condition holds:
  $$
  \forall\labelone\in\labelsone.\forall\stateone\in\statesone.\,\transone(\stateone,\labelone,\statesone)\leq 1
  $$
  where, as usual $\transone(\stateone,\labelone,\setone)$
  stands for $\sum_{\statetwo\in\setone}\transone(\stateone,\labelone,\statetwo)$ whenever $\setone\subseteq\statesone$.
\end{varitemize}
\end{definition}
If $\relone$ is an equivalence relation on $\statesone$, $\quot{\statesone}{\relone}$ denotes
the \emph{quotient} of $\statesone$ modulo $\relone$, i.e., the set of all equivalence classes
of $\statesone$ modulo $\relone$. Given any binary relation $\relone$, its reflexive and transitive
closure is denoted as $\relone^*$.
\begin{definition}\label{d:probabis}
Given a labelled Markov chain $(\statesone,\labelsone,\transone)$, a
\emph{probabilistic bisimulation} is an equivalence relation $\relone$ on
$\statesone$ such that $(\stateone,\statetwo)\in\relone$ implies that for
every $\labelone\in\labelsone$ and for every
$\econe\in\quot{\statesone}{\relone}$,
$\transone(\stateone,\labelone,\econe)=\transone(\statetwo,\labelone,\econe)$.
\end{definition}
Note that a probabilistic bisimulation has to be, by definition, an
\emph{equivalence relation}. This means that, in principle, we are not allowed
to define probabilistic bisimilarity simply as the union of all probabilistic
bisimulations. As a matter of fact, given $\relone, \reltwo$ two
equivalence relations, $\relone \cup \reltwo$ is not necessarily an
equivalence relation. The following is a standard way to overcome the
problem:
\begin{lemma}\label{lemma:kleenestar}
  If $\{\relone_i\}_{i\in I}$, is a collection of probabilistic
  bisimulations, then also their reflexive and transitive closure
  $(\bigcup_{i \in I} \relone_i)^*$ is a probabilistic bisimulation.
\end{lemma}
\begin{proof}
  Let us fix $\reltwo\defi(\bigcup_{i \in I} \relone_i)^*$. The fact that
  $\reltwo$ is an equivalence relation can be proved as follows:
  \begin{varitemize}
  \item Reflexivity is easy: $\reltwo$ is reflexive by definition.
  \item Symmetry is a consequence of symmetry of each of the relations in
    $\{\relone_i\}_{i\in I}$: if $\stateone\;\reltwo\;\statetwo$, then
    there are $n\geq 0$ states $\statethree_0,\ldots,\statethree_n$ such
    that $\statethree_0=\stateone$, $\statethree_n=\statetwo$ and for every
    $1\leq i\leq n$ there is $j$ such that
    $\statethree_{i-1}\;\relone_j\;\statethree_i$. By the symmetry of each
    of the $\relone_j$, we easily get that
    $\statethree_{i}\;\relone_j\;\statethree_{i-1}$. As a consequence,
    $\statetwo\;\reltwo\;\stateone$.
  \item Transitivity is itself very easy: $\reltwo$ is transitive by
    definition.
  \end{varitemize}
  Now, please notice that for any $i\in I$, $\relone_i\subseteq \bigcup_{j
    \in I} \relone_j\subseteq\reltwo$.  This means that any equivalence
  class with respect to $\reltwo$ is the union of equivalence classes with
  respect to $\relone_i$. Suppose that
  $\stateone\;\reltwo\;\statetwo$. Then there are $n\geq 0$ states
  $\statethree_0,\ldots,\statethree_n$ such that $\statethree_0=\stateone$,
  $\statethree_n=\statetwo$ and for every $1\leq i\leq n$ there is $j$ such
  that $\statethree_{i-1}\;\relone_j\;\statethree_i$. Now, if
  $\labelone\in\labelsone$ and $\econe\in\statesone/\reltwo$, we obtain
  $$
  \transone(\stateone,\labelone,\econe)=\transone(\statethree_0,\labelone,\econe)=
  \ldots=\transone(\statethree_{n},\labelone,\econe)=\transone(\statetwo,\labelone,\econe).
  $$
  This concludes the proof.
\end{proof}
\noindent{}Lemma~\ref{lemma:kleenestar} allows us to define the largest probabilistic
bisimulation,  called \textit{probabilistic bisimilarity}. It is
 $\sim\defi\bigcup\{ \relone \
|\ \relone$ is a probabilistic bisimulation$\}$. Indeed, by Lemma~\ref{lemma:kleenestar}, $(\sim)^*$ is a
probabilistic bisimulation too;  we now claim that $\sim\;\mathrel{=}(\sim)^*$. The
inclusion $\sim\;\subseteq(\sim)^*$ is obvious. The other way around,
$\sim\supseteq(\sim)^*$, follows by  $(\sim)^*$ being a
probabilistic bisimulation and hence included in
the union of them all, that is $\sim$.

In the notion of a probabilistic simulation, preorders play the
role of equivalence relations: given a labelled Markov chain 
$(\statesone,\labelsone,\transone)$, a \emph{probabilistic simulation} 
is a preorder relation $\relone$ on $\statesone$ such that 
$(\stateone,\statetwo)\in\relone$ implies that for
every $\labelone\in\labelsone$ and for every
$\setone\subseteq\statesone$, $\transone(\stateone,\labelone,\setone)
\leq\transone(\statetwo,\labelone,\relone(\xsubset))$, where 
as usual $\relone(\setone)$ stands for the $\relone$\textit{-closure} of $\setone$, namely
the set $\{\vartwo\in\statesone\ |\ \exists\,\varone\in\setone.\
\varone\;\relone\;\vartwo\}$. 
Lemma~\ref{lemma:kleenestar} can be adapted to probabilistic simulations:
\begin{proposition}\label{prop:kleenestarsim}
  If $\{\relone_i\}_{i\in I}$, is a collection of probabilistic
  simulations, then also their reflexive and transitive closure
  $(\bigcup_{i \in I} \relone_i)^*$ is a probabilistic simulation.
\end{proposition}
\begin{proof}
  The fact that $\relone\defi(\bigcup_{i \in I} \relone_i)^*$ is a preorder
  follows by construction. Then, for being a probabilistic simulation
  $\relone$ must satisfy the following property:
  $(\stateone,\statetwo)\in\relone$ implies that for every
  $\labelone\in\labelsone$ and for every $\xsubset\subseteq\statesone$,
  $\transone(\stateone,\labelone,\xsubset)\leq\transone(\statetwo,\labelone,\relone(\xsubset))$.
  Let $(\stateone,\statetwo)\in\relone$.
  There are $n\geq 0$ states $\statethree_1,\ldots,\statethree_n$ and for
  every $2\leq i\leq n$ there is $j_i$ such that
  $$
  \stateone=\statethree_1\;\relone_{j_2}\;\statethree_2\ldots\statethree_{n-1}\relone_{j_{n}}\statethree_n=\statetwo.    
  $$
  As a consequence, for every $\labelone\in\labelsone$ and for every
  $\xsubset\subseteq\statesone$, it holds that
  $$
  \transone(\statethree_1,\labelone,\xsubset)\leq\transone(\statethree_2,\labelone,\relone_{j_2}(\xsubset))
  \leq\transone(\statethree_3,\labelone,\relone_{j_3}(\relone_{j_2}(\xsubset)))\leq\dots
  \leq\transone(\statethree_n,\labelone,\relone_{j_n}(\dots(\relone_{j_2}(\relone_{j_1}(\xsubset)))))
  $$
  Since, by definition, 
  $$
  \relone_{j_n}(\dots(\relone_{j_2}(\relone_{j_1}(\xsubset))))\subseteq \relone(\xsubset),
  $$
  it follows that $\transone(\stateone,\labelone,\xsubset)\leq\transone(\statetwo,\labelone,\relone(\xsubset))$.
  This concludes the proof.
\end{proof}
As a consequence, we define \emph{similarity} simply as
$\lesssim\,\defi\bigcup\{ \relone \ |\ \relone$ is a probabilistic
simulation$\}$.

Any symmetric probabilistic simulation is a
probabilistic bisimulation.
\begin{lemma}
  If $\relone$ is a symmetric probabilistic simulation, then $\relone$ is
  a probabilistic bisimulation.
\end{lemma}
\begin{proof}
  If $\relone$ is a symmetric probabilistic simulation, by definition, it
  is also a preorder: that is, it is a reflexive and transitive
  relation. Therefore, $\relone$ is an equivalence relation.  But for being
  a probabilistic bisimulation $\relone$ must also satisfy the property
  that $\stateone\relone\statetwo$ implies, for every
  $\labelone\in\labelsone$ and for every
  $\econe\in\quot{\statesone}{\relone}$,
  $\transone(\stateone,\labelone,\econe)=\transone(\statetwo,\labelone,\econe)$.
  From the fact that $\relone$ is a simulation, it follows that if
  $\stateone\relone\statetwo$, for every $\labelone\in\labelsone$ and for
  every $\econe\in\quot{\statesone}{\relone}$,
  $\transone(\stateone,\labelone,\econe)\leq\transone(\statetwo,\labelone,\relone(\econe))$. Since
  $\econe\in\quot{\statesone}{\relone}$ is an $\relone$-equivalence class,
  it holds $\relone(\econe) = \econe$. Then, from the latter follows
  $\transone(\stateone,\labelone,\econe)\leq\transone(\statetwo,\labelone,\econe)$.
  We get the other way around by symmetric property of $\relone$, which
  implies that, for every label $\labelone$ and for every
  $\econe\in\quot{\statesone}{\relone}$,
  $\transone(\statetwo,\labelone,\econe)\leq\transone(\stateone,\labelone,\econe)$. Hence,
  $\transone(\stateone,\labelone,\econe)=\transone(\statetwo,\labelone,\econe)$
  which completes the proof.
\end{proof}
Moreover, every
probabilistic bisimulation, and its inverse, is a probabilistic simulation.
\begin{lemma}\label{lemma:bisim=simcosim}
  If $\relone$ is a probabilistic bisimulation, then $\relone$ and
  $\relone^{\mathit{op}}$ are probabilistic simulation.
\end{lemma}
\begin{proof}
  Let us prove $\relone$ probabilistic simulation first. 
  Consider the set $\{\setone_i\}_{i\in I}$ of equivalence subclasses
  module $\relone$ contained in $\setone$. Formally, $\setone =
  \biguplus_{i\in I}\setone_i$ such that, for all $i\in I$,
  $\setone_i\subseteq\econe_i$ with $\econe_i$ equivalence class modulo
  $\relone$. Please observe that, as a consequence,
  $\relone(\setone)=\biguplus_{i\in I}\econe_i$.
  Thus, the result easily follows, for every $\labelone\in\labelsone$ and
  every $\setone\subseteq\statesone$,
  \begin{align*}
    \transone(\stateone,\labelone,\setone)&=\sum_{i\in
      I}\transone(\stateone,\labelone,\setone_i)\\
    &\leq \sum_{i\in I}\transone(\stateone,\labelone,\econe_i)\\
    &=\sum_{i\in
      I}\transone(\statetwo,\labelone,\econe_i)=\transone(\statetwo,\labelone,\relone(\setone)).
  \end{align*}
  Finally, $\relone^{\mathit{op}}$ is also a probabilistic simulation as a
  consequence of symmetric property of $\relone$ and the fact, just proved,
  that $\relone$ is a probabilistic simulation.
\end{proof}
Contrary to the nondeterministic case, however,
simulation equivalence coincides with bisimulation:
\begin{proposition}\label{prop:pab=pascopas}
  $\sim$ coincides with $\lesssim\cap\lesssim^{\mathit{op}}$.
\end{proposition}
\begin{proof}
  The fact that $\sim$ is a subset of $\lesssim\cap\lesssim^{\mathit{op}}$
  is a straightforward consequence of symmetry property of $\sim$ and the
  fact that, by Lemma~\ref{lemma:bisim=simcosim}, every probabilistic
  bisimulation is also a probabilistic simulation.
  Let us now prove that $\lesssim\cap\lesssim^{\mathit{op}}$ is a subset of
  $\sim$, i.e., the former of being a probabilistic bisimulation. Of
  course, $\lesssim\cap\lesssim^{\mathit{op}}$ is an equivalence relation
  because $\lesssim$ is a preorder. Now, consider any equivalence class
  $\econe$ modulo $\lesssim\cap\lesssim^{\mathit{op}}$. Define the
  following two sets of states $\setone\defi\lesssim(\econe)$ and
  $\settwo\defi\setone-\econe$. Observe that $\settwo$ and $\econe$ are
  disjoint set of states whose union is precisely $\setone$. Moreover,
  notice that both $\setone$ and $\settwo$ are closed with respect to
  $\lesssim$:
  \begin{varitemize}
  \item On the one hand, if $\stateone\in\lesssim(\setone)$, then
    $\stateone\in\lesssim(\lesssim(\econe))=\lesssim(\econe)=\setone$;
  \item On the other hand, if $\stateone\in\lesssim(\settwo)$, then there
    is $\statetwo\in\setone$ which is not in $\econe$ such that
    $\statetwo\lesssim\stateone$. But then $\stateone$ is itself in
    $\setone$ (see the previous point), but cannot be $\econe$, because
    otherwise we would have $\stateone\lesssim\statetwo$, meaning that
    $\stateone$ and $\statetwo$ are in the same equivalence class modulo
    $\lesssim\cap\lesssim^{\mathit{op}}$, and thus $\statetwo\in\econe$, a
    contradiction.
  \end{varitemize}
  As a consequence, given any
  $(\stateone,\statetwo)\in\lesssim\cap\lesssim^{\mathit{op}}$ and any
  $\labelone\in\labelsone$,
  \begin{align*}
    \transone(\stateone,\labelone,\setone)\leq\transone(\statetwo,\labelone,\lesssim(\setone))=\transone(\statetwo,\labelone,\setone),\\
    \transone(\statetwo,\labelone,\setone)\leq\transone(\stateone,\labelone,\lesssim(\setone))=\transone(\stateone,\labelone,\setone).
  \end{align*}
  It follows
  $\transone(\stateone,\labelone,\setone)=\transone(\statetwo,\labelone,\setone)$
  and, similarly,
  $\transone(\stateone,\labelone,\settwo)=\transone(\statetwo,\labelone,\settwo)$. But
  then,
  \begin{align*}
    \transone(\stateone,\labelone,\econe)&=\transone(\stateone,\labelone,\setone)-\transone(\stateone,\labelone,\settwo)\\
    &=\transone(\statetwo,\labelone,\setone)-\transone(\statetwo,\labelone,\settwo)=\transone(\statetwo,\labelone,\econe)
  \end{align*}
  which is the thesis.
\end{proof}
For technical reasons that will become apparent soon, it is convenient to consider Markov chains in which the
state space is partitioned into disjoint sets, in such a way that comparing states coming from different components
is not possible. Remember that the disjoint union $\biguplus_{i\in I}\setone_i$ of a family of sets $\{\setone_i\}_{i\in\indsetone}$ is defined
as $\{(\elone,i)\mid i\in\indsetone\wedge\elone\in\setone_i\}$.
If the set of states $\statesone$ of a labelled Markov chain is a disjoint union $\biguplus_{i\in\indsetone}\setone_i$, one wants 
that (bi)simulation relations only compare elements coming from the same $\setone_i$, i.e. $(\elone,i)\relone(\eltwo,j)$ implies
$i=j$. In this case, we say that the underlying labelled Markov chain is \emph{multisorted}.

\section{Probabilistic Applicative Bisimulation and Howe's technique}\label{sect:pasht}

In this section, notions of similarity and bisimilarity for $\LOP$ are introduced, in
the spirit of Abramsky's work on applicative bisimulation~\cite{Abramsky-90}. Definitionally, this
consists in seeing $\LOP$'s operational semantics as a labelled Markov chain, then giving
the Larsen and Skou's notion of (bi)simulation for it. States will be terms, while labels
will be of two kinds: one can either \emph{evaluate} a term, obtaining (a distribution of) values, 
or \emph{apply} a term to a value.

The resulting bisimulation (probabilistic applicative bisimulation) 
will be shown to be a congruence, thus included in probabilistic context equivalence.
This will be done by a non-trivial generalization of Howe's technique~\cite{Howe-96}, which is
a well-known methodology to get congruence results in presence of higher-order functions, but
which has not been applied to probabilistic calculi so far. 

Formalizing probabilistic applicative bisimulation requires some care.
As usual, two values $\abstr{\varone}{\termone}$ and $\abstr{\varone}{\termtwo}$ are defined
to be bisimilar if for every $\termthree$, $\subst{\termone}{\varone}{\termthree}$
and $\subst{\termtwo}{\varone}{\termthree}$ are themselves bisimilar. But how if
we rather want to compare two arbitrary closed terms $\termone$ and $\termtwo$?
The simplest solution consists in following Larsen and Skou and stipulate that
every equivalence class of $\val$ modulo bisimulation is attributed the same
measure by both $\sem{\termone}$ and $\sem{\termtwo}$. Values
are thus treated in two different ways (they are both terms and values), and this is the reason why each of
them corresponds to \emph{two} states in the underlying Markov chain.

\begin{definition}
\label{d:multisort}
$\LOP$ can be seen as a multisorted labelled Markov chain
$(\LOPp{\emptyset}\uplus\val,\LOPp{\emptyset}\uplus\{\evlabel\},\translop)$
that we denote with $\cbn{\LOP}$. Labels are either closed terms, which model
parameter passing, or $\evlabel$, that models evaluation. 
Please observe that the states of the labelled Markov chain we have 
just defined are elements of the disjoint union $\LOPp{\emptyset}\uplus\val$. 
Two distinct states correspond to the
same value $\valone$, and to avoid ambiguities, we call the second one
(i.e. the one coming from $\val$) a \emph{distinguished value}.  When we
want to insist on the fact that a value $\abstr{\varone}{\termone}$ is
distinguished, we indicate it with $\clabstr{\varone}{\termone}$. We
define the transition probability matrix $\translop$ as follows:
\begin{varitemize}
\item 
  For every term $\termone$ and for every 
  distinguished value $\clabstr{\varone}{\termtwo}$,
  $$
  \translop(\termone, \evlabel,\clabstr{\varone}{\termtwo})\defi\sem{\termone}(\clabstr{\varone}{\termtwo});
  $$
\item 
  For every term $\termone$ and for every 
  distinguished value $\clabstr{\varone}{\termtwo}$,
  $$
  \translop(\clabstr{\varone}{\termtwo}, \termone,
  \subst{\termtwo}{\varone}{\termone})\defi 1;
  $$
\item
  In all other cases, $\translop$ returns $0$.
\end{varitemize}
\end{definition}
Terms seen as states  only interact with the
environment by performing $\evlabel$, while distinguished values 
only take other closed terms as parameters. 

Simulation and bisimulation relations can be defined for $\cbn{\LOP}$ as for any 
labelled Markov chain. Even if, strictly speaking, these are binary relations on 
$\LOPp{\emptyset}\uplus\val$, we often see them just as their restrictions to 
$\LOPp{\emptyset}$. 
Formally, a \emph{probabilistic applicative bisimulation} (a
$\pabn$) is simply a probabilistic bisimulation on $\cbn{\LOP}$. This way one can 
define \emph{probabilistic applicative bisimilarity}, which is denoted $\cbnpab$.
Similarly for \emph{probabilistic applicative simulation} ($\pasn$) and \emph{probabilistic
applicative similarity}, denoted $\cbnpas$.

\begin{remark}[Early vs. Late]\label{r:el} 
  Technically, the distinction between terms and values in
  Definition~\ref{d:multisort} means that our
  bisimulation is in \emph{late} style. In bisimulations for
  value-passing concurrent languages, ``late'' indicates the explicit
  manipulation of functions in the clause for input actions: functions are
  chosen first, and only later, the input value received is taken into
  account~\cite{SaWabook}. Late-style is used in contraposition to 
  \emph{early} style, where the order of quantifiers is exchanged, so that the
  choice of functions may depend on the specific input value received.
  In our setting, adopting an early style would mean having transitions such as
  $\lambda x. M \arr{N} M \sub N x$, and then setting up a probabilistic
  bisimulation on top of the resulting transition system. 
  We leave for future work a study of the comparison between the two
  styles. In this paper, we stick to the late style because easier to deal with,
  especially under Howe's technique.
  Previous works on applicative
  bisimulation for nondeterministic functions also focus on the late
  approach~\cite{Ong93,PittsSurvey}.
\end{remark} 

\begin{remark}
\label{r:pabfirenze} 
Defining applicative bisimulation in terms of multisorted labelled Markov
chains has the advantage of recasting the definition in a familiar
framework; most importantly, this formulation will be useful when dealing
with Howe's method. To spell out the explicit operational details of the
definition, a probabilistic applicative bisimulation can be seen as an equivalence
relation $\relone\subseteq \LOPp{\emptyset}\times\LOPp{\emptyset}$ such
that whenever $\termone\mathrel{\relone}\termtwo$:
  \begin{varenumerate}
  \item $\sem{\termone}(\econe \cap
    \val)=\sem{\termtwo}(\econe \cap \val)$, for any
    equivalence class $\econe$ of $\relone$ 
(that
    is, the probability of reaching a value in $\econe$ is the same for the
    two terms);
  \item
    \label{cl:ct} if $\termone$ and $\termtwo$ are values, say $\lambda
    x. P $ and $ \lambda x. Q$, then $P \sub L x \RR Q \sub L x$, for
    all $L \in \LOPp{\emptyset}$.
  \end{varenumerate}
  The special treatment of values, in Clause~\ref{cl:ct}., motivates the
  use of \emph{multisorted} labelled Markov chains in Definition~\ref{d:multisort}.
\end{remark}

As usual, one way to show that any two terms are bisimilar is to prove that one relation
containing the pair in question is a $\pabn$. 
Terms with the same semantics are indistinguishable:
\begin{lemma}\label{lemma:samesem}
  The binary relation $\relone=\{(\termone,\termtwo)\in
  \LOPp{\emptyset}\times\LOPp{\emptyset}\st
  \sem{\termone}=\sem{\termtwo}\}\,\biguplus\,\{(\valone,\valone)\in
  \val\times\val\}$ is a $\pabn$.
\end{lemma}
\begin{proof}
  The fact $\relone$ is an equivalence easily follows from reflexivity,
  symmetry and transitivity of set-theoretic equality. $\relone$ must
  satisfy the following property for closed terms: if
  $\termone\relone\termtwo$, then for every
  $\econe\in\quot{\val}{\relone}$,
  $\translop(\termone,\evlabel,\econe)=\translop(\termtwo,\evlabel,\econe)$. Notice
  that if $\sem{\termone}=\sem{\termtwo}$, then clearly
  $\translop(\termone,\evlabel,\valone)=\translop(\termtwo,\evlabel,\valone)$,
  for every $\valone\in\val$. With the same hypothesis,
  \begin{align*}
    \translop(\termone,\evlabel,\econe)&=\sum_{\valone\in\econe}\translop(\termone,\evlabel,\valone)\\
    &=
    \sum_{\valone\in\econe}\translop(\termtwo,\evlabel,\valone)=\translop(\termtwo,\evlabel,\econe).
  \end{align*}
  Moreover, $\relone$ must satisfy the following property for cloned
  values: if
  $\clabstr{\varone}{\termone}\relone\clabstr{\varone}{\termtwo}$, then for
  every close term $\termthree$ and for every
  $\econe\in\quot{\LOPp{\emptyset}}{\relone}$,
  $\translop(\clabstr{\varone}{\termone},\termthree,\econe)=\translop(\clabstr{\varone}{\termtwo},\termthree,\econe)$. Now,
  the hypothesis
  $\sem{\clabstr{\varone}{\termone}}=\sem{\clabstr{\varone}{\termtwo}}$
  implies $\termone = \termtwo$. Then clearly
  $\translop(\clabstr{\varone}{\termone},\termthree,\termfour)=\translop(\clabstr{\varone}{\termtwo},\termthree,\termfour)$
  for every $\termfour\in\LOPp{\emptyset}$. With the same hypothesis,
  \begin{align*}
    \translop(\clabstr{\varone}{\termone},\termthree,\econe)&=\sum_{\termfour\in\econe}\translop(\clabstr{\varone}{\termone},\termthree,\termfour)\\
    &=
    \sum_{\termfour\in\econe}\translop(\clabstr{\varone}{\termtwo},\termthree,\termfour)=\translop(\clabstr{\varone}{\termtwo},\termthree,\econe).
  \end{align*}
  This concludes the proof.
\end{proof}
Please notice that the previous result yield a nice consequence: for every
$\termone,\,\termtwo\in\LOPp{\emptyset}$,
$(\abstr{\varone}{\termone})\termtwo\cbnpab\subst{\termone}{\varone}{\termtwo}$. Indeed,
Lemma~\ref{lemma:semsumCBN} tells us that the latter terms have the same
semantics.

Conversely, knowing that two terms $\termone$ and
$\termtwo$ are (bi)similar means knowing quite a lot about their
convergence probability:
\begin{lemma}[Adequacy of Bisimulation]\label{lemma:sumsempabCBN}
  If $\termone\cbnpab\termtwo$, then $\sumsem{\termone}=\sumsem{\termtwo}$. Moreover,
  if $\termone\cbnpas\termtwo$, then $\sumsem{\termone}\leq\sumsem{\termtwo}$.
\end{lemma}
\begin{proof}
  \begin{align*}
    \sumsem{\termone}&=\sum_{\econe\in\quot{\val}{\cbnpab}}\translop(\termone,\evlabel,\econe)\\
    &=
    \sum_{\econe\in\quot{\val}{\cbnpab}}\translop(\termtwo,\evlabel,\econe)=\sumsem{\termtwo}.
  \end{align*} 
  And, 
  \begin{align*}
    \sumsem{\termone}&=\translop(\termone,\evlabel,\val)\\
    &\leq\translop(\termtwo,\evlabel,\cbnpas(\val))\\
    &=\translop(\termtwo,\evlabel,\val)=\sumsem{\termtwo}.
  \end{align*}
  This concludes the proof.
\end{proof}
\begin{example}
  Bisimilar terms do not necessarily have the same semantics. After all, this is one
  reason for using bisimulation, and its proof method, as basis to prove
  fine-grained equalities among functions. Let us consider the
  following terms:
  \begin{align*}
        \termone\defi&\;((\lambda x. (x \oplus x))\oplus(\lambda x.x))\oplus\Omega;\\
        \termtwo\defi&\;\Omega\oplus(\lambda x. I x) ;
  \end{align*}                         
  Their semantics                      
  differ, as for every value $\valone$, we have:
  \[
   \begin{array}{rcl}
  \sem{\termone}(\valone) &=& \left\{
    \begin{array}{ll}
      \frac{1}{4} &\mbox{if }\valone \mbox{ is } \lambda x. (x \oplus x) \mbox{ or } \lambda x.x;\\
      0 & \mbox{otherwise};
    \end{array}
    \right.
  \\[10pt]
  \sem{\termtwo}(\valone)& =&  \left\{
    \begin{array}{ll}
      \frac{1}{2} &\mbox{if }\valone \mbox{ is } \lambda x. I x;\\
      0 & \mbox{otherwise}.
    \end{array}
    \right.
    \end{array}
   \]
   Nonetheless,   we can prove
   $\termone\cbnpab\termtwo$. Indeed, $\nu x. (x \oplus x)\cbnpab\nu
   x.x\cbnpab\nu x.I x$ because, for every $L\in\LOPp{\emptyset}$, the three
   terms $L$, $\ps{L}{L}$ and $\app{I}{L}$ all have the same semantics,
   i.e., $\sem{L}$. Now, consider any equivalence class $\econe$ of
   distinguished values modulo $\cbnpab$. If $\econe$ includes the three
   distinguished values above, then
   \[\translop(\termone,
   \evlabel,\econe)=\sum_{\valone\in\econe}\sem{\termone}(\valone) =
   \frac{1}{2} = \sum_{\valone\in\econe}\sem{\termtwo}(\valone) =
   \translop(\termtwo,\evlabel,\econe).
   \]
   Otherwise, $\translop(\termone,
   \evlabel,\econe)= 0 = \translop(\termtwo,\evlabel,\econe)$.
\end{example}
Let us prove the following technical result that, moreover, stipulate that
bisimilar distinguished values are bisimilar values.
\begin{lemma}\label{lemma:lambdaredCBN}
  $\abstr{\varone}{\termone}\cbnpab\abstr{\varone}{\termtwo}$ iff
  $\clabstr{\varone}{\termone}\cbnpab\clabstr{\varone}{\termtwo}$ iff
  $\subst{\termone}{\varone}{\termthree}\cbnpab
  \subst{\termtwo}{\varone}{\termthree}$, for all
  $\termthree\in\LOPp{\emptyset}$.
\end{lemma}
\begin{proof}
  The first double implication is obvious. For that matter, distinguished values are value
  terms. Let us now detail the second double implication.
  $(\Rightarrow)$ The fact that $\cbnpab$ is a $\pabn$ implies, by its
  definition, that for every $\termthree\in\LOPp{\emptyset}$ and every
  $\econe\in\quot{\LOPp{\emptyset}}{\cbnpab}$,
  $\translop(\clabstr{\varone}{\termone},\termthree,\econe) =
  \translop(\clabstr{\varone}{\termtwo},\termthree,\econe)$. Suppose then,
  by contradiction, that
  $\subst{\termone}{\varone}{\termthree}\cbn{\not\pab}
  \subst{\termtwo}{\varone}{\termthree}$, for some
  $\termthree\in\LOPp{\emptyset}$. The latter means that, there exists
  $\ectwo\in\quot{\LOPp{\emptyset}}{\cbnpab}$ such that
  $\subst{\termone}{\varone}{\termthree}\in\ectwo$ and
  $\subst{\termtwo}{\varone}{\termthree}\not\in\ectwo$. According to its
  definition, for all $\termfour\in\LOPp{\emptyset}$,
  $\translop(\clabstr{\varone}{\termone},\termthree,\termfour) = 1$ iff
  $\termfour\equiv\subst{\termone}{\varone}{\termthree}$, and
  $\translop(\clabstr{\varone}{\termone},\termthree,\termfour) = 0$
  otherwise. Then, since $\subst{\termone}{\varone}{\termthree}\in\ectwo$,
  we derive $\translop(\clabstr{\varone}{\termone},\termthree,\ectwo) =
  \sum_{\termfour\in\ectwo}\translop(\abstr{\varone}{\termone},\termthree,\termfour)
  \geq
  \translop(\clabstr{\varone}{\termone},\termthree,\subst{\termone}{\varone}{\termthree})
  = 1$, which implies
  $\sum_{\termfour\in\ectwo}\translop(\clabstr{\varone}{\termone},\termthree,\termfour)
  = \translop(\clabstr{\varone}{\termone},\termthree,\ectwo) = 1$. Although
  $\clabstr{\varone}{\termtwo}$ is a distinguished value and the starting
  reasoning we have just made above still holds,
  $\translop(\clabstr{\varone}{\termtwo},\termthree,\ectwo) =
  \sum_{\termfour\in\ectwo}\translop(\clabstr{\varone}{\termtwo},\termthree,\termfour)
  = 0$. We get the latter because there is no $\termfour\in\ectwo$ of the
  form $\subst{\termtwo}{\varone}{\termthree}$ due to the hypothesis that
  $\subst{\termtwo}{\varone}{\termthree}\not\in\ectwo$.

  From the hypothesis on the equivalence class $\ectwo$, i.e.
  $\translop(\clabstr{\varone}{\termone},\termthree,\ectwo) =
  \translop(\clabstr{\varone}{\termtwo},\termthree,\ectwo)$, we derive the
  absurd:
  \begin{align*}
    1 = \translop(\clabstr{\varone}{\termone},\termthree,\ectwo) =
    \translop(\clabstr{\varone}{\termtwo},\termthree,\ectwo) = 0.
  \end{align*}
  $(\Leftarrow)$ We need to prove that, for every
  $\termthree\in\LOPp{\emptyset}$ and every
  $\econe\in\quot{\LOPp{\emptyset}}{\cbnpab}$,
  $\translop(\clabstr{\varone}{\termone},\termthree,\econe) =
  \translop(\clabstr{\varone}{\termtwo},\termthree,\econe)$ supposing that
  $\subst{\termone}{\varone}{\termthree}\cbnpab
  \subst{\termtwo}{\varone}{\termthree}$ holds. First of all, let us
  rewrite $\translop(\clabstr{\varone}{\termone},\termthree,\econe)$ and
  $\translop(\clabstr{\varone}{\termtwo},\termthree,\econe)$ as
  $\sum_{\termfour\in\econe}\translop(\clabstr{\varone}{\termone},\termthree,\termfour)$
  and
  $\sum_{\termfour\in\econe}\translop(\clabstr{\varone}{\termtwo},\termthree,\termfour)$
  respectively. Then, from the hypothesis and the same reasoning we have
  made for ($\Rightarrow$), for every $\econe\in\quot{\LOPp{\emptyset}}{\cbnpab}$:
  $$
  \sum_{\termfour\in\econe}\translop(\clabstr{\varone}{\termone},\termthree,\termfour)
  = \left\{
    \begin{array}{ll}
      1 &\mbox{if }\subst{\termone}{\varone}{\termthree}\in\econe;\\
      0 & \mbox{otherwise}
    \end{array}
  \right.  = \left\{
    \begin{array}{ll}
      1 &\mbox{if }\subst{\termtwo}{\varone}{\termthree}\in\econe;\\
      0 & \mbox{otherwise}
    \end{array}
  \right.  =
  \sum_{\termfour\in\econe}\translop(\clabstr{\varone}{\termtwo},\termthree,\termfour)
  $$
  which proves the thesis.
\end{proof}
The same result holds for $\cbnpas$.
\subsection{Probabilistic Applicative Bisimulation is a Congruence}
In this section, we prove that probabilistic applicative bisimulation is indeed
a congruence, and that its non-symmetric sibling is a precongruence. 
The overall structure of the proof is similar to the one by
Howe~\cite{Howe-96}. The main idea consists in defining a way to turn
an arbitrary relation $\relone$ on (possibly open) terms to another one, $\howe{\relone}$,
in such a way that, if $\relone$ satisfies a few simple conditions, then $\howe{\relone}$ is 
a (pre)congruence including $\relone$. The key step, then, is to prove
that $\howe{\relone}$ is indeed a (bi)simulation. In view of Proposition~\ref{prop:pab=pascopas}, 
considering similarity suffices here.

It is here convenient to work with generalizations of relations called
\emph{$\LOP$-relations}, i.e.  sets of triples in the form
$(\vecvarone,\termone,\termtwo)$, where
$\termone,\termtwo\in\LOP(\vecvarone)$.  Thus if a relation has the pair
$(M,N)$ with $\termone,\termtwo\in\LOP(\vecvarone)$, then the corresponding
$\LOP$-relation will include $(\vecvarone,\termone,\termtwo)$.  (Recall
that applicative (bi)similarity is extended to open terms by considering
all closing substitutions.) 
Given any
$\LOP$-relation $\relone$,  we write $\rel{\vecvarone}{\termone}{\relone}{\termtwo}$
if $(\vecvarone,\termone,\termtwo)\in\relone$. A
$\LOP$-relation $\relone$ is said to be \emph{compatible} iff the four
conditions below hold:
\begin{varitemize}
\item[\Comone] $\forall\vecvarone\in\powfin{\setvar}$, $\varone\in\vecvarone$:
  $\rel{\vecvarone}{\varone}{\relone}{\varone}$,
\item[\Comtwo]
  $\forall\vecvarone\in\powfin{\setvar}$,$\forall\varone\in\setvar-\vecvarone$,$\forall\termone,
  \termtwo\in\LOP(\vecvarone\cup\{\varone\})$:
  $\rel{\vecvarone\cup\{\varone\}}{\termone}{\relone}{\termtwo}\Rightarrow
  \rel{\vecvarone}{\abstr{\varone}{\termone}}{\relone}{\abstr{\varone}{\termtwo}}$,
\item[\Comthree]
  $\forall\vecvarone\in\powfin{\setvar}$,$\forall\termone,\termtwo,\termthree,\termfour\in\LOP(\vecvarone)$:
  $\rel{\vecvarone}{\termone}{\relone}{\termtwo}\wedge\rel{\vecvarone}{\termthree}{\relone}{\termfour}\Rightarrow
  \rel{\vecvarone}{\app{\termone}{\termthree}}{\relone}{\app{\termtwo}{\termfour}}$,
\item[\Comfour]
  $\forall\vecvarone\in\powfin{\setvar}$,$\forall\termone,\termtwo,\termthree,\termfour\in\LOP(\vecvarone)$:
  $\rel{\vecvarone}{\termone}{\relone}{\termtwo}\wedge\rel{\vecvarone}{\termthree}{\relone}{\termfour}\Rightarrow
  \rel{\vecvarone}{\ps{\termone}{\termthree}}{\relone}{\ps{\termtwo}{\termfour}}$.
\end{varitemize}
We will often use the following technical results to establish
$\Comthree$ and $\Comfour$ under particular hypothesis.
\begin{lemma}\label{lemma:com3LR}
  Let us consider the properties
  \begin{varitemize}
  \item[\ComthreeL]
    $\forall\vecvarone\in\powfin{\setvar}$,$\forall\termone,\termtwo,\termthree\in\LOP(\vecvarone)$:
    $\rel{\vecvarone}{\termone}{\relone}{\termtwo}\Rightarrow
    \rel{\vecvarone}{\app{\termone}{\termthree}}{\relone}{\app{\termtwo}{\termthree}}$,
  \item[\ComthreeR]
    $\forall\vecvarone\in\powfin{\setvar}$,$\forall\termone,\termtwo,\termthree\in\LOP(\vecvarone)$:
    $\rel{\vecvarone}{\termone}{\relone}{\termtwo}\Rightarrow
    \rel{\vecvarone}{\app{\termthree}{\termone}}{\relone}{\app{\termthree}{\termtwo}}$.
  \end{varitemize}
  If $\relone$ is transitive, then $\ComthreeL$ and $\ComthreeR$ together imply $\Comthree$.
\end{lemma}
\begin{proof}
  Proving $\Comthree$ means to show that the hypothesis
  $\rel{\vecvarone}{\termone}{\relone}{\termtwo}$ and
  $\rel{\vecvarone}{\termthree}{\relone}{\termfour}$ imply
  $\rel{\vecvarone}{\app{\termone}{\termthree}}{\relone}{\app{\termtwo}{\termfour}}$. Using
  $\ComthreeL$ on the first one, with $\termthree$ as steady term, it
  follows
  $\rel{\vecvarone}{\app{\termone}{\termthree}}{\relone}{\app{\termtwo}{\termthree}}$. Similarly,
  using $\ComthreeR$ on the second one, with $\termtwo$ as steady term, it
  follows
  $\rel{\vecvarone}{\app{\termtwo}{\termthree}}{\relone}{\app{\termtwo}{\termfour}}$. Then,
  we conclude by transitivity property of $\relone$.
\end{proof}

\begin{lemma}\label{lemma:com4LR}
  Let us consider the properties
  \begin{varitemize}
  \item[\ComfourL]
    $\forall\vecvarone\in\powfin{\setvar}$,$\forall\termone,\termtwo,\termthree\in\LOP(\vecvarone)$:
    $\rel{\vecvarone}{\termone}{\relone}{\termtwo}\Rightarrow
    \rel{\vecvarone}{\ps{\termone}{\termthree}}{\relone}{\ps{\termtwo}{\termthree}}$,
  \item[\ComfourR]
    $\forall\vecvarone\in\powfin{\setvar}$,$\forall\termone,\termtwo,\termthree\in\LOP(\vecvarone)$:
    $\rel{\vecvarone}{\termone}{\relone}{\termtwo}\Rightarrow
    \rel{\vecvarone}{\ps{\termthree}{\termone}}{\relone}{\ps{\termthree}{\termtwo}}$.
  \end{varitemize}
  If $\relone$ is transitive, then $\ComfourL$ and $\ComfourR$ together imply $\Comfour$.
\end{lemma}
\begin{proof}
  Proving $\Comfour$ means to show that the hypothesis
  $\rel{\vecvarone}{\termone}{\relone}{\termtwo}$ and
  $\rel{\vecvarone}{\termthree}{\relone}{\termfour}$ imply
  $\rel{\vecvarone}{\ps{\termone}{\termthree}}{\relone}{\ps{\termtwo}{\termfour}}$. Using
  $\ComfourL$ on the first one, with $\termthree$ as steady term, it
  follows
  $\rel{\vecvarone}{\ps{\termone}{\termthree}}{\relone}{\ps{\termtwo}{\termthree}}$. Similarly,
  using $\ComfourR$ on the second one, with $\termtwo$ as steady term, it
  follows
  $\rel{\vecvarone}{\ps{\termtwo}{\termthree}}{\relone}{\ps{\termtwo}{\termfour}}$. Then,
  we conclude by transitivity property of $\relone$.
\end{proof}
The notions of an equivalence relation and of a preorder can be straightforwardly generalized to $\LOP$-relations, and
any compatible $\LOP$-relation that is an equivalence relation (respectively, a preorder) is said to be a \emph{congruence} 
(respectively, a \emph{precongruence}).

If bisimilarity is a congruence, then $\ctxone[\termone]$ is bisimilar to $\ctxone[\termtwo]$ whenever
$\termone\cbnpab\termtwo$ and $\ctxone$ is a context. In other words, terms can be replaced by equivalent ones in any context.
This is a crucial sanity-check any notion of equivalence  is expected to pass.

It is well-known that proving bisimulation to be a congruence may be
nontrivial when the underlying language contains higher-order
functions. This is also the case here. Proving $\Comone,\Comtwo$ and
$\Comfour$ just by inspecting the operational semantics of the involved
terms is indeed possible, but the method fails for $\Comthree$, when the
involved contexts contain applications.  In particular, proving $\Comthree$ requires
probabilistic applicative bisimilarity of being stable with respect to
substitution of bisimilar terms, hence not necessarily the
same. In general, a $\LOP$-relation $\relone$ is called (term) \emph{substitutive} if
for all $\vecvarone\in\powfin{\setvar}$, $\varone\in\setvar-\vecvarone$,
$\termone,\termtwo\in\LOPp{\vecvarone\cup\{\varone\}}$ and
$\termthree,\termfour\in\LOPp{\vecvarone}$
\begin{equation}
  \label{eq:SubCBN}
  \rel{\vecvarone\cup\{\varone\}}{\termone}{\relone}{\termtwo}\wedge
  \rel{\vecvarone}{\termthree}{\relone}{\termfour}\Rightarrow
  \rel{\vecvarone}{\subst{\termone}{\varone}{\termthree}}{\relone}{\subst{\termtwo}{\varone}{\termfour}}.
\end{equation}
Note that if $\relone$ is also reflexive, then this implies
\begin{equation}
  \label{eq:CusCBN}
  \rel{\vecvarone\cup\{\varone\}}{\termone}{\relone}{\termtwo}\wedge\termthree\in\LOPp{\vecvarone}\Rightarrow
  \rel{\vecvarone}{\subst{\termone}{\varone}{\termthree}}{\relone}{\subst{\termtwo}{\varone}{\termthree}}.
\end{equation}
We say that $\relone$ is \textit{closed under term-substitution} if it
satisfies (\ref{eq:CusCBN}). Because of the way the open extension of
$\cbnpab$ and $\cbnpas$ are defined, they are closed under
term-substitution.

Unfortunately, directly prove $\cbnpas$ to enjoy such \emph{substitutivity}
property is hard. We will thus proceed indirectly by defining, starting
from $\cbnpas$, a new relation $\howe{\cbnpas}$, called the \emph{Howe's
  lifting} of $\cbnpas$, that has such property by construction and that
can be proved equal to $\cbnpas$.  

Actually, the Howe's lifting of any $\LOP$-relation $\relone$ is the relation $\howe{\relone}$ defined by the
rules in Figure~\ref{fig:howelifting}.
\begin{figure*}
  $$
  \infer[\Howeone]
  {\rel{\vecvarone}{\varone}{\howe{\relone}}{\termone}}
  {\rel{\vecvarone}{\varone}{\relone}{\termone}}
  \qquad
  \infer[\Howetwo]
  {\rel{\vecvarone}{\abstr{\varone}{\termone}}{\howe{\relone}}{\termtwo}}
  {
    \rel{\vecvarone\cup\{\varone\}}{\termone}{\howe{\relone}}{\termthree}
    &&
    \rel{\vecvarone}{\abstr{\varone}{\termthree}}{\relone}{\termtwo}
    &&
    \varone\notin\vecvarone
  }
  $$
  $$
  \infer[\Howethree]
  {\rel{\vecvarone}{\app{\termone}{\termtwo}}{\howe{\relone}}{\termthree}}
  {
    \rel{\vecvarone}{\termone}{\howe{\relone}}{\termfour}
    &&
    \rel{\vecvarone}{\termtwo}{\howe{\relone}}{\termfive}
    &&
    \rel{\vecvarone}{\app{\termfour}{\termfive}}{\relone}{\termthree}
  }
  $$
  $$
  \infer[\Howefour]
  {\rel{\vecvarone}{\ps{\termone}{\termtwo}}{\howe{\relone}}{\termthree}}
  {
    \rel{\vecvarone}{\termone}{\howe{\relone}}{\termfour}
    &&
    \rel{\vecvarone}{\termtwo}{\howe{\relone}}{\termfive}
    &&
    \rel{\vecvarone}{\ps{\termfour}{\termfive}}{\relone}{\termthree}
  }
  $$
\caption{Howe's Lifting for $\LOP$.}\label{fig:howelifting}
\end{figure*}
The reader familiar with Howe's method should have a sense of
\emph{d\'ej\`a vu} here: indeed, this is \emph{precisely} 
the same definition one finds in the realm of \emph{nondeterministic} $\lambda$-calculi. The language of terms, after all,
is the same. This facilitates the first part of the proof. Indeed, one
already knows that if $\relone$ is a preorder,
then $\howe{\relone}$ is compatible
and includes $\relone$, since all these properties are already known (see, e.g.~\cite{PittsSurvey}) and
only depend on the shape of terms and not on their operational semantics. 
\begin{lemma}\label{lemma:howeprop1}
    If $\relone$ is reflexive, then $\howe{\relone}$ is compatible.
  \end{lemma}
  \begin{proof}
    We need to prove that \Comone, \Comtwo,
    \Comthree, and \Comfour\ hold for $\howe{\relone}$:
    \begin{varitemize}
    \item
      Proving \Comone{} means to show:
      $$
      \forall\vecvarone\in\powfin{\setvar},
      \varone\in\vecvarone\Rightarrow\rel{\vecvarone}{\varone}{\howe{\relone}}{\varone}.
      $$
      Since $\relone$ is reflexive, $\forall\vecvarone\in\powfin{\setvar}$,
      $\varone\in\vecvarone\Rightarrow\rel{\vecvarone}{\varone}{\relone}{\varone}$. Thus,
      by \Howeone, we conclude
      $\rel{\vecvarone}{\varone}{\howe{\relone}}{\varone}$. Formally,
      $$
      \infer[\Howeone]
      {\rel{\vecvarone}{\varone}{\howe{\relone}}{\varone}}
      {\infer{\rel{\vecvarone}{\varone}{\relone}{\varone}} {}}
      $$
    \item 
      Proving \Comtwo{} means to show:
      $\forall\vecvarone\in\powfin{\setvar}$,
      $\forall\varone\in\setvar-\vecvarone$, $\forall\termone,
      \termtwo\in\LOP(\vecvarone\cup\{\varone\})$,
      $$
      \rel{\vecvarone\cup\{\varone\}}{\termone}{\howe{\relone}}{\termtwo}\Rightarrow
      \rel{\vecvarone}{\abstr{\varone}{\termone}}{\howe{\relone}}{\abstr{\varone}{\termtwo}}.
      $$
      Since $\relone$ is reflexive, we get
      $\rel{\vecvarone}{\abstr{\varone}{\termtwo}}{\relone}{\abstr{\varone}{\termtwo}}$. Moreover,
      we have
      $\rel{\vecvarone\cup\{\varone\}}{\termone}{\howe{\relone}}{\termtwo}$
      by hypothesis. Thus, by \Howetwo, we conclude
      $\rel{\vecvarone}{\abstr{\varone}{\termone}}{\howe{\relone}}{\abstr{\varone}{\termtwo}}$
      holds. Formally,
      $$
      \infer[\Howetwo]
      {\rel{\vecvarone}{\abstr{\varone}{\termone}}{\howe{\relone}}{\abstr{\varone}{\termtwo}}}
      { \rel{\vecvarone\cup\{\varone\}}{\termone}{\howe{\relone}}{\termtwo}
        &&
        \infer{\rel{\vecvarone}{\abstr{\varone}{\termtwo}}{\relone}{\abstr{\varone}{\termtwo}}}{}
        && \varone\notin\vecvarone }
      $$
   \item
      Proving \Comthree\ means to show:
      $\forall\vecvarone\in\powfin{\setvar}$, $\forall\termone,\termtwo,\termthree,\termfour\in\LOP(\vecvarone)$,
      $$
      \rel{\vecvarone}{\termone}{\howe{\relone}}{\termtwo}\wedge\rel{\vecvarone}{\termthree}{\howe{\relone}}{\termfour}\Rightarrow
      \rel{\vecvarone}{\app{\termone}{\termthree}}{\howe{\relone}}{\app{\termtwo}{\termfour}}.
      $$
      Since $\relone$ is reflexive, we get
      $\rel{\vecvarone}{\app{\termtwo}{\termfour}}{\relone}{\app{\termtwo}{\termfour}}$. Moreover,
      we have $\rel{\vecvarone}{\termone}{\howe{\relone}}{\termtwo}$ and
      $\rel{\vecvarone}{\termthree}{\howe{\relone}}{\termfour}$ by
      hypothesis. Thus, by \Howethree, we conclude
      $\rel{\vecvarone}{\app{\termone}{\termthree}}{\howe{\relone}}{\app{\termtwo}{\termfour}}$
      holds. Formally,
      $$
      \infer[\Howethree]
      {\rel{\vecvarone}{\app{\termone}{\termthree}}{\howe{\relone}}{\app{\termtwo}{\termfour}}}
      { \rel{\vecvarone}{\termone}{\howe{\relone}}{\termtwo} &&
        \rel{\vecvarone}{\termthree}{\howe{\relone}}{\termfour} &&
        \infer{\rel{\vecvarone}{\app{\termtwo}{\termfour}}{\relone}{\app{\termtwo}{\termfour}}}{}}
      $$
    \item
      Proving \Comfour\ means to show:
      $\forall\vecvarone\in\powfin{\setvar}$, $\forall\termone,\termtwo,\termthree,\termfour\in\LOP(\vecvarone)$,
      $$
      \rel{\vecvarone}{\termone}{\howe{\relone}}{\termtwo}\wedge\rel{\vecvarone}{\termthree}{\howe{\relone}}{\termfour}\Rightarrow
      \rel{\vecvarone}{\ps{\termone}{\termthree}}{\howe{\relone}}{\ps{\termtwo}{\termfour}}.$$
      Since $\relone$ is reflexive, we get
      $\rel{\vecvarone}{\ps{\termtwo}{\termfour}}{\relone}{\ps{\termtwo}{\termfour}}$. Moreover,
      we have $\rel{\vecvarone}{\termone}{\howe{\relone}}{\termtwo}$ and
      $\rel{\vecvarone}{\termthree}{\howe{\relone}}{\termfour}$ by
      hypothesis. Thus, by \Howefour, we conclude
      $\rel{\vecvarone}{\ps{\termone}{\termthree}}{\howe{\relone}}{\ps{\termtwo}{\termfour}}$
      holds. Formally,
      $$
      \infer[\Howefour]
      {\rel{\vecvarone}{\ps{\termone}{\termthree}}{\howe{\relone}}{\ps{\termtwo}{\termfour}}}
      { \rel{\vecvarone}{\termone}{\howe{\relone}}{\termtwo} &&
        \rel{\vecvarone}{\termthree}{\howe{\relone}}{\termfour} &&
        \infer{\rel{\vecvarone}{\ps{\termtwo}{\termfour}}{\relone}{\ps{\termtwo}{\termfour}}}{}
      }
      $$
    \end{varitemize}
    This concludes the proof.
\end{proof}
  \begin{lemma}\label{lemma:howeprop2}
    If $\relone$ is transitive, then
    $\rel{\vecvarone}{\termone}{\howe{\relone}}{\termtwo}$ and
    $\rel{\vecvarone}{\termtwo}{\relone}{\termthree}$ imply
    $\rel{\vecvarone}{\termone}{\howe{\relone}}{\termthree}$.
  \end{lemma}
\begin{proof}
  We prove the statement by inspection on the last 
  rule used in the derivation of
  $\rel{\vecvarone}{\termone}{\howe{\relone}}{\termtwo}$, thus on the
  structure of $\termone$.
  \begin{varitemize}
  \item 
    If $\termone$ is a variable, say $\varone \in \vecvarone$, then
    $\rel{\vecvarone}{\varone}{\howe{\relone}}{\termtwo}$ holds by
    hypothesis. The last rule used has to be \Howeone.  Thus, we get
    $\rel{\vecvarone}{\varone}{\relone}{\termtwo}$ as additional
    hypothesis. By transitivity of $\relone$, from
    $\rel{\vecvarone}{\varone}{\relone}{\termtwo}$ and
    $\rel{\vecvarone}{\termtwo}{\relone}{\termthree}$ we deduce
    $\rel{\vecvarone}{\varone}{\relone}{\termthree}$. We conclude by
    \Howeone\ on the latter, obtaining
    $\rel{\vecvarone}{\varone}{\howe{\relone}}{\termthree}$,
    i.e. $\rel{\vecvarone}{\termone}{\howe{\relone}}{\termthree}$. Formally,
    $$
    \infer[\Howeone]
    {\rel{\vecvarone}{\varone}{\howe{\relone}}{\termthree}}
    {\infer{\rel{\vecvarone}{\varone}{\relone}{\termthree}} {
        \rel{\vecvarone}{\varone}{\relone}{\termtwo} &&
        \rel{\vecvarone}{\termtwo}{\relone}{\termthree} } }
    $$
  \item If $\termone$ is a $\lambda$-abstraction, say
    $\abstr{\varone}{\termfive}$, then 
    $\rel{\vecvarone}{\abstr{\varone}{\termfive}}{\howe{\relone}}{\termtwo}$
    holds by hypothesis. The last rule used has to be \Howetwo. Thus, we
    get
    $\rel{\vecvarone\cup\{\varone\}}{\termfive}{\howe{\relone}}{\termfour}$
    and $\rel{\vecvarone}{\abstr{\varone}{\termfour}}{\relone}{\termtwo}$
    as additional hypothesis. By transitivity of $\relone$, from
    $\rel{\vecvarone}{\abstr{\varone}{\termfour}}{\relone}{\termtwo}$ and
    $\rel{\vecvarone}{\termtwo}{\relone}{\termthree}$ we deduce
    $\rel{\vecvarone}{\abstr{\varone}{\termfour}}{\relone}{\termthree}$. We
    conclude by \Howetwo\ on
    $\rel{\vecvarone\cup\{\varone\}}{\termfive}{\howe{\relone}}{\termfour}$
    and the latter, obtaining
    $\rel{\vecvarone}{\abstr{\varone}{\termfive}}{\howe{\relone}}{\termthree}$,
    i.e. $\rel{\vecvarone}{\termone}{\howe{\relone}}{\termthree}$.
    Formally, 
    $$
    \infer[\Howetwo]
    {\rel{\vecvarone}{\abstr{\varone}{\termfive}}{\howe{\relone}}{\termthree}}
    {
      \rel{\vecvarone\cup\{\varone\}}{\termfive}{\howe{\relone}}{\termfour}
      &&
      \infer{\rel{\vecvarone}{\abstr{\varone}{\termfour}}{\relone}{\termthree}}{
        \rel{\vecvarone}{\abstr{\varone}{\termfour}}{\relone}{\termtwo}
        && \rel{\vecvarone}{\termtwo}{\relone}{\termthree}}}
    $$
  \item If $\termone$ is an application, say $\app{\termsix}{\termseven}$,
    then 
    $\rel{\vecvarone}{\app{\termsix}{\termseven}}{\howe{\relone}}{\termtwo}$
    holds by hypothesis. The last rule used has to be \Howethree. Thus, we
    get $\rel{\vecvarone}{\termsix}{\howe{\relone}}{\termfour}$,
    $\rel{\vecvarone}{\termseven}{\howe{\relone}}{\termfive}$ and
    $\rel{\vecvarone}{\app{\termfour}{\termfive}}{\relone}{\termtwo}$ as
    additional hypothesis. By transitivity of $\relone$, from
    $\rel{\vecvarone}{\app{\termfour}{\termfive}}{\relone}{\termtwo}$ and
    $\rel{\vecvarone}{\termtwo}{\relone}{\termthree}$ we deduce
    $\rel{\vecvarone}{\app{\termfour}{\termfive}}{\relone}{\termthree}$. We
    conclude by \Howethree{} on
    $\rel{\vecvarone}{\termsix}{\howe{\relone}}{\termfour}$,
    $\rel{\vecvarone}{\termseven}{\howe{\relone}}{\termfive}$ and the
    latter, obtaining
    $\rel{\vecvarone}{\app{\termsix}{\termseven}}{\howe{\relone}}{\termthree}$,
    i.e. $\rel{\vecvarone}{\termone}{\howe{\relone}}{\termthree}$.
    Formally, 
    $$
    \infer[\Howethree]
    {\rel{\vecvarone}{\app{\termsix}{\termseven}}{\howe{\relone}}{\termthree}}
    { \rel{\vecvarone}{\termsix}{\howe{\relone}}{\termfour} &&
      \rel{\vecvarone}{\termseven}{\howe{\relone}}{\termfive} &&
      \infer{\rel{\vecvarone}{\app{\termfour}{\termfive}}{\relone}{\termthree}}{
        \rel{\vecvarone}{\app{\termfour}{\termfive}}{\relone}{\termtwo} &&
        \rel{\vecvarone}{\termtwo}{\relone}{\termthree}
      }
    }
    $$
  \item If $\termone$ is a probabilistic sum, say
    $\ps{\termsix}{\termseven}$, then 
    $\rel{\vecvarone}{\ps{\termsix}{\termseven}}{\howe{\relone}}{\termtwo}$
    holds by hypothesis. The last rule used has to be \Howefour. Thus, we
    get $\rel{\vecvarone}{\termsix}{\howe{\relone}}{\termfour}$,
    $\rel{\vecvarone}{\termseven}{\howe{\relone}}{\termfive}$ and
    $\rel{\vecvarone}{\ps{\termfour}{\termfive}}{\relone}{\termtwo}$ as
    additional hypothesis. By transitivity of $\relone$, from
    $\rel{\vecvarone}{\ps{\termfour}{\termfive}}{\relone}{\termtwo}$ and
    $\rel{\vecvarone}{\termtwo}{\relone}{\termthree}$ we deduce
    $\rel{\vecvarone}{\ps{\termfour}{\termfive}}{\relone}{\termthree}$.  We
    conclude by \Howefour\ on
    $\rel{\vecvarone}{\termsix}{\howe{\relone}}{\termfour}$,
    $\rel{\vecvarone}{\termseven}{\howe{\relone}}{\termfive}$ and the
    latter, obtaining
    $\rel{\vecvarone}{\ps{\termsix}{\termseven}}{\howe{\relone}}{\termthree}$,
    i.e. $\rel{\vecvarone}{\termone}{\howe{\relone}}{\termthree}$.
    Formally, 
    $$
    \infer[\Howefour]
    {\rel{\vecvarone}{\ps{\termsix}{\termseven}}{\howe{\relone}}{\termthree}}
    { \rel{\vecvarone}{\termsix}{\howe{\relone}}{\termfour} &&
      \rel{\vecvarone}{\termseven}{\howe{\relone}}{\termfive} &&
      \infer{\rel{\vecvarone}{\ps{\termfour}{\termfive}}{\relone}{\termthree}}{
        \rel{\vecvarone}{\ps{\termfour}{\termfive}}{\relone}{\termtwo}
        && \rel{\vecvarone}{\termtwo}{\relone}{\termthree} } }
    $$
  \end{varitemize}
  This concludes the proof.
\end{proof}
  \begin{lemma}\label{lemma:howeprop3}
    If $\relone$ is reflexive, then
    $\rel{\vecvarone}{\termone}{\relone}{\termtwo}$ implies
    $\rel{\vecvarone}{\termone}{\howe{\relone}}{\termtwo}$.
  \end{lemma}
\begin{proof}
  We will prove it by inspection on the structure of $\termone$.
  \begin{varitemize}
  \item 
    If $\termone$ is a variable, say $\varone \in \vecvarone$, then
    $\rel{\vecvarone}{\varone}{\relone}{\termtwo}$ holds by hypothesis. We
    conclude by \Howeone\ on the latter, obtaining
    $\rel{\vecvarone}{\varone}{\howe{\relone}}{\termtwo}$,
    i.e. $\rel{\vecvarone}{\termone}{\howe{\relone}}{\termtwo}$. Formally,
    $$
    \infer[\Howeone]
    {\rel{\vecvarone}{\varone}{\howe{\relone}}{\termtwo}}
    {\rel{\vecvarone}{\varone}{\relone}{\termtwo}}
    $$
  \item If $\termone$ is a $\lambda$-abstraction, say
    $\abstr{\varone}{\termfive}$, then 
    $\rel{\vecvarone}{\abstr{\varone}{\termfive}}{\relone}{\termtwo}$ holds
    by hypothesis. Moreover, since 
    $\relone$ reflexive implies $\howe{\relone}$ compatible,
    $\howe{\relone}$ is 
    reflexive too. Then, from
    $\rel{\vecvarone\cup\{\varone\}}{\termfive}{\howe{\relone}}{\termfive}$
    and $\rel{\vecvarone}{\abstr{\varone}{\termfive}}{\relone}{\termtwo}$
    we conclude, by \Howetwo,
    $\rel{\vecvarone}{\abstr{\varone}{\termfive}}{\howe{\relone}}{\termtwo}$,
    i.e. $\rel{\vecvarone}{\termone}{\howe{\relone}}{\termtwo}$. Formally,
    $$
    \infer[\Howetwo]
    {\rel{\vecvarone}{\abstr{\varone}{\termfive}}{\howe{\relone}}{\termtwo}}
    {
      \rel{\vecvarone\cup\{\varone\}}{\termfive}{\howe{\relone}}{\termfive}
      &&
      \rel{\vecvarone}{\abstr{\varone}{\termfive }}{\relone}{\termtwo}
      && \varone\notin\vecvarone }
    $$
  \item If $\termone$ is an application, say $\app{\termthree}{\termfour}$,
    then 
    $\rel{\vecvarone}{\app{\termthree}{\termfour}}{\relone}{\termtwo}$
    holds by hypothesis. By reflexivity of $\relone$, hence that of
    $\howe{\relone}$ too, we get
    $\rel{\vecvarone}{\termthree}{\howe{\relone}}{\termthree}$ and
    $\rel{\vecvarone}{\termfour}{\howe{\relone}}{\termfour}$. Then, from
    the latter and
    $\rel{\vecvarone}{\app{\termthree}{\termfour}}{\relone}{\termtwo}$ we
    conclude, by \Howethree,
    $\rel{\vecvarone}{\app{\termthree}{\termfour}}{\howe{\relone}}{\termtwo}$,
    i.e. $\rel{\vecvarone}{\termone}{\howe{\relone}}{\termtwo}$. Formally,
    $$
    \infer[\Howethree]
    {\rel{\vecvarone}{\app{\termthree}{\termfour}}{\howe{\relone}}{\termtwo}}
    { \rel{\vecvarone}{\termthree}{\howe{\relone}}{\termthree} &&
      \rel{\vecvarone}{\termfour}{\howe{\relone}}{\termfour} &&
      \rel{\vecvarone}{\app{\termthree}{\termfour}}{\relone}{\termtwo}
    }
    $$
  \item If $\termone$ is a probabilistic sum, say
    $\ps{\termthree}{\termfour}$, then 
    $\rel{\vecvarone}{\ps{\termthree}{\termfour}}{\relone}{\termtwo}$ holds
    by hypothesis. By reflexivity of $\relone$, hence that of
    $\howe{\relone}$ too, we get
    $\rel{\vecvarone}{\termthree}{\howe{\relone}}{\termthree}$ and
    $\rel{\vecvarone}{\termfour}{\howe{\relone}}{\termfour}$. Then, from
    the latter and
    $\rel{\vecvarone}{\ps{\termthree}{\termfour}}{\relone}{\termtwo}$ we
    conclude, by \Howefour,
    $\rel{\vecvarone}{\ps{\termthree}{\termfour}}{\howe{\relone}}{\termtwo}$,
    i.e. $\rel{\vecvarone}{\termone}{\howe{\relone}}{\termtwo}$. Formally,
    $$
    \infer[\Howefour]
    {\rel{\vecvarone}{\ps{\termthree}{\termfour}}{\howe{\relone}}{\termtwo}}
    { \rel{\vecvarone}{\termthree}{\howe{\relone}}{\termthree} &&
      \rel{\vecvarone}{\termfour}{\howe{\relone}}{\termfour} &&
      \rel{\vecvarone}{\ps{\termthree}{\termfour}}{\relone}{\termtwo} }
    $$
  \end{varitemize}
  This concludes the proof.
\end{proof}

Moreover, if $\relone$ is a preorder and closed under term-substitution,
then its lifted relation $\howe{\relone}$ is substitutive. Then,
reflexivity of $\relone$ implies compatibility of $\howe{\relone}$ by
Lemma~\ref{lemma:howeprop1}. It follows $\howe{\relone}$ reflexive too,
hence closed under term-substitution.

\begin{lemma}\label{lemma:closesubsCBN}
  If $\relone$ is reflexive, transitive and closed under term-substitution,
  then $\howe{\relone}$ is (term) substitutive and hence also closed under
  term-substitution.
\end{lemma}
\begin{proof}
  We show that, for all $\vecvarone\in\powfin{\setvar}$,
  $\varone\in\setvar-\vecvarone$,
  $\termone,\,\termtwo\in\LOPp{\vecvarone\cup\{\varone\}}$ and
  $\termthree,\,\termfour\in\LOPp{\vecvarone}$,
  \begin{align*}
    \rel{\vecvarone\cup\{\varone\}}{\termone}{\howe{\relone}}{\termtwo}\wedge
    \rel{\vecvarone}{\termthree}{\howe{\relone}}{\termfour}\Rightarrow
    \rel{\vecvarone}{\subst{\termone}{\varone}{\termthree}}{\howe{\relone}}{\subst{\termtwo}{\varone}{\termfour}}.
  \end{align*}
  We prove the latter by induction on the derivation of
  $\rel{\vecvarone\cup\{\varone\}}{\termone}{\howe{\relone}}{\termtwo}$,
  thus on the structure of $\termone$.
  \begin{varitemize}
  \item If $\termone$ is a variable, then either $\termone = \varone$ or
    $\termone\in\vecvarone$. In the latter case, suppose $\termone =
    \vartwo$. Then, 
    by hypothesis, 
    $\rel{\vecvarone\cup\{\varone\}}{\vartwo}{\howe{\relone}}{\termtwo}$
    holds and the only way to deduce it is by rule \Howeone\ from
    $\rel{\vecvarone\cup\{\varone\}}{\vartwo}{\relone}{\termtwo}$. Hence,
    by the fact $\relone$ is closed under term-substitution and
    $\termfour\in\LOPp{\vecvarone}$, we obtain
    $\rel{\vecvarone}{\subst{\vartwo}{\varone}{\termfour}}{\relone}{\subst{\termtwo}{\varone}{\termfour}}$
    which is equivalent to
    $\rel{\vecvarone}{\vartwo}{\relone}{\subst{\termtwo}{\varone}{\termfour}}$. Finally,
    by Lemma~\ref{lemma:howeprop3}, we conclude
    $\rel{\vecvarone}{\vartwo}{\howe{\relone}}{\subst{\termtwo}{\varone}{\termfour}}$
    which is equivalent to
    $\rel{\vecvarone}{\subst{\vartwo}{\varone}{\termthree}}{\howe{\relone}}{\subst{\termtwo}{\varone}{\termfour}}$,
    i.e.
    $\rel{\vecvarone}{\subst{\termone}{\varone}{\termthree}}{\howe{\relone}}{\subst{\termtwo}{\varone}{\termfour}}$
    holds. Otherwise, $\termone = \varone$ and
    $\rel{\vecvarone\cup\{\varone\}}{\varone}{\howe{\relone}}{\termtwo}$
    holds. The only way to deduce the latter is by the rule \Howeone\ from
    $\rel{\vecvarone\cup\{\varone\}}{\varone}{\relone}{\termtwo}$. Hence,
    by the fact $\relone$ is closed under term-substitution and
    $\termfour\in\LOPp{\vecvarone}$, we obtain
    $\rel{\vecvarone}{\subst{\varone}{\varone}{\termfour}}{\relone}{\subst{\termtwo}{\varone}{\termfour}}$
    which is equivalent to
    $\rel{\vecvarone}{\termfour}{\relone}{\subst{\termtwo}{\varone}{\termfour}}$. By
    Lemma~\ref{lemma:howeprop2}, we deduce the following:
    $$
    \infer[]
    {\rel{\vecvarone}{\termthree}{\howe{\relone}}{\subst{\termtwo}{\varone}{\termfour}}}
    { \rel{\vecvarone}{\termthree}{\howe{\relone}}{\termfour} &&
      \rel{\vecvarone}{\termfour}{\relone}{\subst{\termtwo}{\varone}{\termfour}}
    }
    $$
    which is equivalent to
    $\rel{\vecvarone}{\subst{\varone}{\varone}{\termthree}}{\howe{\relone}}{\subst{\termtwo}{\varone}{\termfour}}$. Thus,
    $\rel{\vecvarone}{\subst{\termone}{\varone}{\termthree}}{\howe{\relone}}{\subst{\termtwo}{\varone}{\termfour}}$
    holds.
  \item If $\termone$ is a $\lambda$-abstraction, say
    $\abstr{\vartwo}{\termfive}$, then 
    $\rel{\vecvarone\cup\{\varone\}}{\abstr{\vartwo}{\termfive}}{\howe{\relone}}{\termtwo}$
    holds by hypothesis. The only way to deduce the latter is by rule
    \Howetwo\ as follows:
    $$
    \infer[\Howetwo]
    {\rel{\vecvarone\cup\{\varone\}}{\abstr{\vartwo}{\termfive}}{\howe{\relone}}{\termtwo}}
    {
      \rel{\vecvarone\cup\{\varone,\vartwo\}}{\termfive}{\howe{\relone}}{\termsix}
      &&
      \rel{\vecvarone\cup\{\varone\}}{\abstr{\vartwo}{\termsix}}{\relone}{\termtwo}
      && \varone,\vartwo\notin\vecvarone }
    $$

    Let us denote $\vecvartwo = \vecvarone\cup\{\vartwo\}$. Then, by
    induction hypothesis on
    $\rel{\vecvartwo\cup\{\varone\}}{\termfive}{\howe{\relone}}{\termsix}$,
    we get
    $\rel{\vecvartwo}{\subst{\termfive}{\varone}{\termthree}}{\howe{\relone}}{\subst{\termsix}{\varone}{\termfour}}$. Moreover,
    by the fact $\relone$ is closed under term-substitution and
    $\termfour\in\LOPp{\vecvarone}$, we obtain that
    $\rel{\vecvarone}{\subst{(\abstr{\vartwo}{\termsix})}{\varone}{\termfour}}{\relone}{\subst{\termtwo}{\varone}{\termfour}}$
    holds, i.e.
    $\rel{\vecvarone}{\subst{\abstr{\vartwo}{\termsix}}{\varone}{\termfour}}{\relone}{\subst{\termtwo}{\varone}{\termfour}}$. 
    By \Howetwo, we deduce the following:
    $$
    \infer[\Howetwo]
    {\rel{\vecvarone}{\subst{\abstr{\vartwo}{\termfive}}{\varone}{\termthree}}{\howe{\relone}}{\subst{\termtwo}{\varone}{\termfour}}}
    {
      \rel{\vecvarone\cup\{\vartwo\}}{\subst{\termfive}{\varone}{\termthree}}{\howe{\relone}}{\subst{\termsix}{\varone}{\termfour}}
      &&
      \rel{\vecvarone}{\subst{\abstr{\vartwo}{\termsix}}{\varone}{\termfour}}{\relone}{\subst{\termtwo}{\varone}{\termfour}}
      && \vartwo\notin\vecvarone } 
    $$
    which is equivalent to
    $\rel{\vecvarone}{\subst{(\abstr{\vartwo}{\termfive})}{\varone}{\termthree}}{\howe{\relone}}{\subst{\termtwo}{\varone}{\termfour}}$. Thus,
    $\rel{\vecvarone}{\subst{\termone}{\varone}{\termthree}}{\howe{\relone}}{\subst{\termtwo}{\varone}{\termfour}}$
    holds.
  \item If $\termone$ is an application, say $\app{\termfive}{\termsix}$,
    then 
    $\rel{\vecvarone\cup\{\varone\}}{\app{\termfive}{\termsix}}{\howe{\relone}}{\termtwo}$
    holds by hypothesis. The only way to deduce the latter is by rule
    \Howethree\ as follows:
    $$
    \infer[\Howethree]
    {\rel{\vecvarone\cup\{\varone\}}{\app{\termfive}{\termsix}}{\howe{\relone}}{\termtwo}}
    {
      \rel{\vecvarone\cup\{\varone\}}{\termfive}{\howe{\relone}}{\termfive'}
      &&
      \rel{\vecvarone\cup\{\varone\}}{\termsix}{\howe{\relone}}{\termsix'}
      &&
      \rel{\vecvarone\cup\{\varone\}}{\app{\termfive'}{\termsix'}}{\relone}{\termtwo}
    }
    $$

    By induction hypothesis on
    $\rel{\vecvarone\cup\{\varone\}}{\termfive}{\howe{\relone}}{\termfive'}$
    and
    $\rel{\vecvarone\cup\{\varone\}}{\termsix}{\howe{\relone}}{\termsix'}$,
    we get
    $\rel{\vecvarone}{\subst{\termfive}{\varone}{\termthree}}{\howe{\relone}}{\subst{\termfive'}{\varone}{\termfour}}$
    and
    $\rel{\vecvarone}{\subst{\termsix}{\varone}{\termthree}}{\howe{\relone}}{\subst{\termsix'}{\varone}{\termfour}}$.
    Moreover, by the fact $\relone$ is closed under term-substitution and
    $\termfour\in\LOPp{\vecvarone}$, we obtain that
    $\rel{\vecvarone}{\subst{(\app{\termfive'}{\termsix'})}{\varone}{\termfour}}{\relone}{\subst{\termtwo}{\varone}{\termfour}}$
    holds, i.e.
    $\rel{\vecvarone}{\app{\subst{\termfive'}{\varone}{\termfour}}{\subst{\termsix'}{\varone}{\termfour}}}{\relone}{\subst{\termtwo}{\varone}{\termfour}}$.
    By \Howethree, we deduce the following:
    $$
    \infer[\Howethree]
    {\rel{\vecvarone}{\app{\subst{\termfive}{\varone}{\termthree}}{\subst{\termsix}{\varone}{\termthree}}}{\howe{\relone}}{\subst{\termtwo}{\varone}{\termfour}}}
    {
      \rel{\vecvarone}{\subst{\termfive}{\varone}{\termthree}}{\howe{\relone}}{\subst{\termfive'}{\varone}{\termfour}}
      &&
      \rel{\vecvarone}{\subst{\termsix}{\varone}{\termthree}}{\howe{\relone}}{\subst{\termsix'}{\varone}{\termfour}}
      &&
      \rel{\vecvarone}{\app{\subst{\termfive'}{\varone}{\termfour}}{\subst{\termsix'}{\varone}{\termfour}}}{\relone}{\subst{\termtwo}{\varone}{\termfour}}}
    $$
    which is equivalent to
    $\rel{\vecvarone}{\subst{(\app{\termfive}{\termsix})}{\varone}{\termthree}}{\howe{\relone}}{\subst{\termtwo}{\varone}{\termfour}}$. Thus,
    $\rel{\vecvarone}{\subst{\termone}{\varone}{\termthree}}{\howe{\relone}}{\subst{\termtwo}{\varone}{\termfour}}$
    holds.
  \item If $\termone$ is a probabilistic sum, say
    $\ps{\termfive}{\termsix}$, then 
    $\rel{\vecvarone\cup\{\varone\}}{\ps{\termfive}{\termsix}}{\howe{\relone}}{\termtwo}$
    holds by hypothesis. The only way to deduce the latter is by rule
    \Howefour\ as follows:
    $$
    \infer[\Howefour]
    {\rel{\vecvarone\cup\{\varone\}}{\ps{\termfive}{\termsix}}{\howe{\relone}}{\termtwo}}
    {
      \rel{\vecvarone\cup\{\varone\}}{\termfive}{\howe{\relone}}{\termfive'}
      &&
      \rel{\vecvarone\cup\{\varone\}}{\termsix}{\howe{\relone}}{\termsix'}
      &&
      \rel{\vecvarone\cup\{\varone\}}{\ps{\termfive'}{\termsix'}}{\relone}{\termtwo}
    }
    $$
    By induction hypothesis on
    $\rel{\vecvarone\cup\{\varone\}}{\termfive}{\howe{\relone}}{\termfive'}$
    and
    $\rel{\vecvarone\cup\{\varone\}}{\termsix}{\howe{\relone}}{\termsix'}$,
    we get
    $\rel{\vecvarone}{\subst{\termfive}{\varone}{\termthree}}{\howe{\relone}}{\subst{\termfive'}{\varone}{\termfour}}$
    and
    $\rel{\vecvarone}{\subst{\termsix}{\varone}{\termthree}}{\howe{\relone}}{\subst{\termsix'}{\varone}{\termfour}}$. Moreover,
    by the fact $\relone$ is closed under term-substitution and
    $\termfour\in\LOPp{\vecvarone}$, we obtain that
    $\rel{\vecvarone}{\subst{(\ps{\termfive'}{\termsix'})}{\varone}{\termfour}}{\relone}{\subst{\termtwo}{\varone}{\termfour}}$,
    i.e.
    $\rel{\vecvarone}{\ps{\subst{\termfive'}{\varone}{\termfour}}{\subst{\termsix'}{\varone}{\termfour}}}{\relone}{\subst{\termtwo}{\varone}{\termfour}}$. 
    By \Howefour, we conclude the following:
    $$
    \infer[\Howefour]
    {\rel{\vecvarone}{\ps{\subst{\termfive}{\varone}{\termthree}}{\subst{\termsix}{\varone}{\termthree}}}{\howe{\relone}}{\subst{\termtwo}{\varone}{\termfour}}}
    {
      \rel{\vecvarone}{\subst{\termfive}{\varone}{\termthree}}{\howe{\relone}}{\subst{\termfive'}{\varone}{\termfour}}
      &&
      \rel{\vecvarone}{\subst{\termsix}{\varone}{\termthree}}{\howe{\relone}}{\subst{\termsix'}{\varone}{\termfour}}
      &&
      \rel{\vecvarone}{\ps{\subst{\termfive'}{\varone}{\termfour}}{\subst{\termsix'}{\varone}{\termfour}}}{\relone}{\subst{\termtwo}{\varone}{\termfour}}
    }
    $$
    which is equivalent to
    $\rel{\vecvarone}{\subst{(\ps{\termfive}{\termsix})}{\varone}{\termthree}}{\howe{\relone}}{\subst{\termtwo}{\varone}{\termfour}}$. Thus,
    $\rel{\vecvarone}{\subst{\termone}{\varone}{\termthree}}{\howe{\relone}}{\subst{\termtwo}{\varone}{\termfour}}$
    holds.
  \end{varitemize}
  This concludes the proof.
\end{proof}
Something is missing, however, before we can conclude that
$\howe{\cbnpas}$ is a precongruence, namely transitivity.
We also follow Howe here building the transitive
closure of a $\LOP$-relation $\relone$ as the relation $\tcrel{\relone}$
defined by the rules in Figure~\ref{fig:transitiveclosure}. 
\begin{figure*}
$$
\infer[\TCone]
{\rel{\vecvarone}{\termone}{\tcrel{\relone}}{\termtwo}}
{\rel{\vecvarone}{\termone}{\relone}{\termtwo}}
$$
$$
\infer[\TCtwo]
{\rel{\vecvarone}{\termone}{\tcrel{\relone}}{\termthree}} {
  \rel{\vecvarone}{\termone}{\tcrel{\relone}}{\termtwo} &&
  \rel{\vecvarone}{\termtwo}{\tcrel{\relone}}{\termthree} }
$$
\caption{Transitive Closure for $\LOP$.}\label{fig:transitiveclosure}
\end{figure*}
Then, it is easy to prove $\tcrel{\relone}$ of being compatible and closed
under term-substitution if $\relone$ is.  
\begin{lemma}\label{lemma:tcrelCom}
  If $\relone$ is compatible, then so is $\tcrel{\relone}$.
\end{lemma}
\begin{proof}
  We need to prove that \Comone, \Comtwo, \Comthree, and \Comfour\ hold for
  $\tcrel{\relone}$:
  \begin{varitemize}
  \item Proving \Comone{} means to show:
    $$
    \forall\vecvarone\in\powfin{\setvar}, \varone\in\vecvarone\Rightarrow
    \rel{\vecvarone}{\varone}{\relone}{\varone}.
    $$
    Since $\relone$ is compatible, therefore reflexive,
    $\rel{\vecvarone}{\varone}{\relone}{\varone}$ holds. Hence
    $\rel{\vecvarone}{\varone}{\tcrel{\relone}}{\varone}$ follows by
    \TCone.
  \item Proving \Comtwo{} means to show:
    $\forall\vecvarone\in\powfin{\setvar}$,
    $\forall\varone\in\setvar-\vecvarone$, $\forall\termone,
    \termtwo\in\LOP(\vecvarone\cup\{\varone\})$,
    $$
    \rel{\vecvarone\cup\{\varone\}}{\termone}{\tcrel{\relone}}{\termtwo}\Rightarrow
    \rel{\vecvarone}{\abstr{\varone}{\termone}}{\tcrel{\relone}}{\abstr{\varone}{\termtwo}}.
    $$
    We prove it by induction on the derivation of
    $\rel{\vecvarone\cup\{\varone\}}{\termone}{\tcrel{\relone}}{\termtwo}$,
    looking at the last rule used. The base case has $\TCone$ as last rule:
    thus, $\rel{\vecvarone\cup\{\varone\}}{\termone}{\relone}{\termtwo}$
    holds. Then, since $\relone$ is compatible, it follows
    $\rel{\vecvarone}{\abstr{\varone}{\termone}}{\relone}{\abstr{\varone}{\termtwo}}$. We
    conclude applying $\TCone$ on the latter, obtaining
    $\rel{\vecvarone}{\abstr{\varone}{\termone}}{\tcrel{\relone}}{\abstr{\varone}{\termtwo}}$. Otherwise,
    if $\TCtwo$ is the last rule used, we get that, for some
    $\termthree\in\LOP(\vecvarone\cup\{\varone\})$,
    $\rel{\vecvarone\cup\{\varone\}}{\termone}{\tcrel{\relone}}{\termthree}$
    and
    $\rel{\vecvarone\cup\{\varone\}}{\termthree}{\tcrel{\relone}}{\termtwo}$
    hold. Then, by induction hypothesis on both of them, we have
    $\rel{\vecvarone}{\abstr{\varone}{\termone}}{\tcrel{\relone}}{\abstr{\varone}{\termthree}}$
    and
    $\rel{\vecvarone}{\abstr{\varone}{\termthree}}{\tcrel{\relone}}{\abstr{\varone}{\termtwo}}$. We
    conclude applying $\TCtwo$ on the latter two, obtaining
    $\rel{\vecvarone}{\abstr{\varone}{\termone}}{\tcrel{\relone}}{\abstr{\varone}{\termtwo}}$.
  \item Proving \Comthree{} means to show:
    $\forall\vecvarone\in\powfin{\setvar}$, $\forall\termone,\termtwo,\termthree,\termfour\in\LOP(\vecvarone)$,
    $$
    \rel{\vecvarone}{\termone}{\tcrel{\relone}}{\termtwo}\ \wedge\
    \rel{\vecvarone}{\termthree}{\tcrel{\relone}}{\termfour}\Rightarrow
    \rel{\vecvarone}{\termone\termthree}{\tcrel{\relone}}{\termtwo\termfour}.
    $$
    Firstly, we prove the following two characterizations:
    \begin{align}
      &\label{equ:com3lp}\forall
      \termone,\termtwo,\termthree,\termfour\in\LOP(\vecvarone).\
      \rel{\vecvarone}{\termone}{\tcrel{\relone}}{\termtwo}\ \wedge \
      \rel{\vecvarone}{\termthree}{\relone}{\termfour}
      \Rightarrow \rel{\vecvarone}{\app{\termone}{\termthree}}{\tcrel{\relone}}{\app{\termtwo}{\termfour}},\\
      &\label{equ:com3rp}\forall
      \termone,\termtwo,\termthree,\termfour\in\LOP(\vecvarone).\
      \rel{\vecvarone}{\termone}{\relone}{\termtwo}\ \wedge \
      \rel{\vecvarone}{\termthree}{\tcrel{\relone}}{\termfour} \Rightarrow
      \rel{\vecvarone}{\app{\termone}{\termthree}}{\tcrel{\relone}}{\app{\termtwo}{\termfour}}.
    \end{align}
    In particular, we only prove (\ref{equ:com3lp}) in details, since
    (\ref{equ:com3rp}) is similarly provable. We prove (\ref{equ:com3lp})
    by induction on the derivation
    $\rel{\vecvarone}{\termone}{\tcrel{\relone}}{\termtwo}$, looking at the
    last rule used.  The base case has $\TCone$ as last rule: we get that
    $\rel{\vecvarone}{\termone}{\relone}{\termtwo}$ holds. Then, using
    $\relone$ compatibility property and
    $\rel{\vecvarone}{\termthree}{\relone}{\termfour}$, it follows
    $\rel{\vecvarone}{\app{\termone}{\termthree}}{\relone}{\app{\termtwo}{\termfour}}$. We
    conclude applying $\TCone$ on the latter, obtaining
    $\rel{\vecvarone}{\app{\termone}{\termthree}}{\tcrel{\relone}}{\app{\termtwo}{\termfour}}$.
    Otherwise, if $\TCtwo$ is the last rule used, we get that, for some
    $\termfive\in\LOP$,
    $\rel{\vecvarone}{\termone}{\tcrel{\relone}}{\termfive}$ and
    $\rel{\vecvarone}{\termfive}{\tcrel{\relone}}{\termtwo}$ hold. Then, by
    induction hypothesis on
    $\rel{\vecvarone}{\termone}{\tcrel{\relone}}{\termfive}$ along with
    $\rel{\vecvarone}{\termthree}{\relone}{\termfour}$, we have
    $\rel{\vecvarone}{\app{\termone}{\termthree}}{\tcrel{\relone}}{\app{\termfive}{\termfour}}$. Then,
    since $\relone$ is compatible and so reflexive too,
    $\rel{\vecvarone}{\termfour}{\relone}{\termfour}$ holds. By induction
    hypothesis on $\rel{\vecvarone}{\termfive}{\tcrel{\relone}}{\termtwo}$
    along with the latter, we get
    $\rel{\vecvarone}{\app{\termfive}{\termfour}}{\tcrel{\relone}}{\app{\termtwo}{\termfour}}$. We
    conclude applying $\TCtwo$ on
    $\rel{\vecvarone}{\app{\termone}{\termthree}}{\tcrel{\relone}}{\app{\termfive}{\termfour}}$
    and
    $\rel{\vecvarone}{\app{\termfive}{\termfour}}{\tcrel{\relone}}{\app{\termtwo}{\termfour}}$,
    obtaining
    $\rel{\vecvarone}{\app{\termone}{\termthree}}{\tcrel{\relone}}{\app{\termtwo}{\termfour}}$.

    Let us focus on the original \Comthree{} statement. We prove it by
    induction on the two derivations
    $\rel{\vecvarone}{\termone}{\tcrel{\relone}}{\termtwo}$ and
    $\rel{\vecvarone}{\termthree}{\tcrel{\relone}}{\termfour}$, which we
    name here as $\pfone$ and $\pftwo$ respectively. Looking at the last
    rules used, there are four possible cases as four are the combinations
    that permit to conclude with $\pfone$ and $\pftwo$:
    \begin{varenumerate}
    \item $\TCone$ for both $\pfone$ and $\pftwo$;
    \item $\TCone$ for $\pfone$ and $\TCtwo$ for $\pftwo$;
    \item $\TCtwo$ for $\pfone$ and $\TCone$ for $\pftwo$;
    \item $\TCtwo$ for both $\pfone$ and $\pftwo$.
    \end{varenumerate}
    Observe now that the first three cases are addressed by
    (\ref{equ:com3lp}) and (\ref{equ:com3rp}). Hence, it remains to prove
    the last case, where both derivations are concluded applying $\TCtwo$
    rule. According to $\TCtwo$ rule definition, we get two additional
    hypothesis from each derivation. In particular, for $\pfone$, we get
    that, for some $\termfive\in\LOP(\vecvarone)$,
    $\rel{\vecvarone}{\termone}{\tcrel{\relone}}{\termfive}$ and
    $\rel{\vecvarone}{\termfive}{\tcrel{\relone}}{\termtwo}$
    hold. Similarly, for $\pftwo$, we get that, for some
    $\termsix\in\LOP(\vecvarone)$,
    $\rel{\vecvarone}{\termthree}{\tcrel{\relone}}{\termsix}$ and
    $\rel{\vecvarone}{\termsix}{\tcrel{\relone}}{\termfour}$ hold. Then, by
    a double induction hypothesis, firstly on
    $\rel{\vecvarone}{\termone}{\tcrel{\relone}}{\termfive}$,
    $\rel{\vecvarone}{\termthree}{\tcrel{\relone}}{\termsix}$ and secondly
    on $\rel{\vecvarone}{\termfive}{\tcrel{\relone}}{\termtwo}$,
    $\rel{\vecvarone}{\termsix}{\tcrel{\relone}}{\termfour}$, we get
    $\rel{\vecvarone}{\app{\termone}{\termthree}}{\tcrel{\relone}}{\app{\termfive}{\termsix}}$
    and
    $\rel{\vecvarone}{\app{\termfive}{\termsix}}{\tcrel{\relone}}{\app{\termtwo}{\termfour}}$
    respectively. We conclude applying $\TCtwo$ on these latter, obtaining
    $\rel{\vecvarone}{\app{\termone}{\termthree}}{\tcrel{\relone}}{\app{\termtwo}{\termfour}}$.
  \item Proving \Comfour{} means to show:
    $\forall\vecvarone\in\powfin{\setvar}$, $\forall\termone,\termtwo,\termthree,\termfour\in\LOP(\vecvarone)$,
    $$
    \rel{\vecvarone}{\termone}{\tcrel{\relone}}{\termtwo}\ \wedge\
    \rel{\vecvarone}{\termthree}{\tcrel{\relone}}{\termfour}\Rightarrow
    \rel{\vecvarone}{\ps{\termone}{\termthree}}{\tcrel{\relone}}{\ps{\termtwo}{\termfour}}.
    $$
    We do not detail the proof since it boils down to that of \Comthree,
    where partial sums play the role of applications.
  \end{varitemize}
  This concludes the proof.
\end{proof}

\begin{lemma}\label{lemma:tcrelCTS}
  If $\relone$ is closed under term-substitution, then so is
  $\tcrel{\relone}$.
\end{lemma}
\begin{proof}
  We need to prove $\tcrel{\relone}$ of being closed under
  term-substitution: for all $\vecvarone\in\powfin{\setvar}$,
  $\varone\in\setvar-\vecvarone$,
 $\termone,\termtwo\in\LOPp{\vecvarone\cup\{\varone\}}$ and
  $\termthree,\termfour\in\LOPp{\vecvarone}$,
  $$
  \rel{\vecvarone\cup\{\varone\}}{\termone}{\tcrel{\relone}}{\termtwo}\wedge\termthree\in\LOPp{\vecvarone}\Rightarrow
  \rel{\vecvarone}{\subst{\termone}{\varone}{\termthree}}{\tcrel{\relone}}{\subst{\termtwo}{\varone}{\termthree}}.
  $$
  We prove the latter by induction on the derivation of
  $\rel{\vecvarone\cup\{\varone\}}{\termone}{\tcrel{\relone}}{\termtwo}$,
  looking at the last rule used. The base case has $\TCone$ as last rule:
  we get that
  $\rel{\vecvarone\cup\{\varone\}}{\termone}{\relone}{\termtwo}$
  holds. Then, since $\relone$ is closed under term-substitution, it
  follows
  $\rel{\vecvarone}{\subst{\termone}{\varone}{\termthree}}{\relone}{\subst{\termtwo}{\varone}{\termthree}}$. We
  conclude applying $\TCone$ on the latter, obtaining
  $\rel{\vecvarone}{\subst{\termone}{\varone}{\termthree}}{\tcrel{\relone}}{\subst{\termtwo}{\varone}{\termthree}}$.
  Otherwise, if $\TCtwo$ is the last rule used, we get that, for some
  $\termfour\in\LOPp{\vecvarone\cup\{\varone\}}$,
  $\rel{\vecvarone\cup\{\varone\}}{\termone}{\tcrel{\relone}}{\termfour}$
  and
  $\rel{\vecvarone\cup\{\varone\}}{\termfour}{\tcrel{\relone}}{\termtwo}$
  hold. Then, by induction hypothesis on both of them, we have
  $\rel{\vecvarone}{\subst{\termone}{\varone}{\termthree}}{\tcrel{\relone}}{\subst{\termfour}{\varone}{\termthree}}$
  and
  $\rel{\vecvarone}{\subst{\termfour}{\varone}{\termthree}}{\tcrel{\relone}}{\subst{\termtwo}{\varone}{\termthree}}$. We
  conclude applying $\TCtwo$ on the latter two, obtaining
  $\rel{\vecvarone}{\subst{\termone}{\varone}{\termthree}}{\tcrel{\relone}}{\subst{\termtwo}{\varone}{\termthree}}$.
\end{proof}

It is important to note that the transitive closure of an already Howe's
lifted relation is a preorder if the starting relation is.

\begin{lemma}\label{lemma:tcrelPO}
  If a $\LOP$-relation $\relone$ is a preorder relation, then so is
  $\tcrel{(\howe{\relone})}$.
\end{lemma}
\begin{proof}
  We need to show 
  $\tcrel{(\howe{\relone})}$ of being 
  reflexive and transitive. Of course, being a transitive closure,
  $\tcrel{(\howe{\relone})}$ is a a transitive relation. Moreover, since
  $\relone$ is reflexive, by Lemma~\ref{lemma:howeprop1}, $\howe{\relone}$
  is reflexive too because compatible. Then, by Lemma~\ref{lemma:tcrelCom},
  so is $\tcrel{(\howe{\relone})}$.
\end{proof}
This is just the first half of the story: we also need
to prove that $\tcrel{(\howe{\cbnpas})}$ is  a
simulation. As we already know it is a preorder, the following lemma
gives us the missing bit:
\begin{lemma}[Key Lemma]\label{lemma:keylemma}
  If $\relu{\termone}{\howe{\cbnpas}}{\termtwo}$, then for every
  $\setone\subseteq\LOP(\varone)$ it holds that
  $\sem{\termone}(\abstr{\varone}{\setone})\leq\sem{\termtwo}(\abstr{\varone}{(\howe{\cbnpas}(\setone))})$.
\end{lemma} 
The proof of this lemma is delicate and is discussed in the next section.
From the lemma, using a standard argument we derive the needed
substitutivity results, and ultimately the most important result
of this section.
\begin{theorem}\label{thm:pasprecongrCBN}
  $\cbnpas$ is a precongruence relation for $\LOP$-terms.
\end{theorem}
\begin{proof}
  We prove the result by observing that $\tcrel{(\howe{\cbnpas})}$ is a
  precongruence and by showing that
  $\cbnpas=\tcrel{(\howe{\cbnpas})}$. First of all,
  Lemma~\ref{lemma:howeprop1} and Lemma~\ref{lemma:tcrelCom} ensure that
  $\tcrel{(\howe{\cbnpas})}$ is compatible and Lemma~\ref{lemma:tcrelPO}
  tells us that $\tcrel{(\howe{\cbnpas})}$ is a preorder. As a consequence,
  $\tcrel{(\howe{\cbnpas})}$ is a precongruence.  Consider now the
  inclusion $\cbnpas\subseteq\tcrel{(\howe{\cbnpas})}$.  By
  Lemma~\ref{lemma:howeprop3} and by definition of transitive closure
  operator $\tcrel{(\cdot)}$, it follows that
  $\cbnpas\,\subseteq\,(\howe{\cbnpas})\,\subseteq\,\tcrel{(\howe{\cbnpas})}$.
  We show the converse by proving that $\tcrel{(\howe{\cbnpas})}$ is
  included in a relation $\relone$ that is a call-by-name probabilistic
  applicative simulation, therefore contained in the largest one. In
  particular, since $\tcrel{(\howe{\cbnpas})}$ is closed under
  term-substitution (Lemma~\ref{lemma:closesubsCBN} and
  Lemma~\ref{lemma:tcrelCTS}), it suffices to show the latter only on the
  closed version of terms and cloned values. $\relone$ acts like
  $\tcrel{(\howe{\cbnpas})}$ on terms, while given two cloned values
  $\clabstr{\varone}{\termone}$ and $\clabstr{\varone}{\termtwo}$,
  $(\clabstr{\varone}{\termone})\relone(\clabstr{\varone}{\termtwo})$ iff
  $\termone\tcrel{(\howe{\cbnpas})}\termtwo$. Since we already know that
  $\tcrel{(\howe{\cbnpas})}$ is a preorder (and thus $\relone$ is itself a
  preorder), all that remain to be checked are the following two points:
  \begin{varitemize}
  \item If $\termone\tcrel{(\howe{\cbnpas})}\termtwo$, then for every
    $\setone\subseteq\LOPp{\varone}$ it holds that
    \begin{equation}\label{equ:probsimCBN}
      \translop(\termone,\evlabel,\clabstr{\varone}{\setone})\leq\translop(\termtwo,\evlabel,\relone(\clabstr{\varone}{\setone})).
    \end{equation}
    Let us proceed by induction on the structure of the proof of
    $\termone\tcrel{(\howe{\cbnpas})}\termtwo$:
    \begin{varitemize}
    \item The base case has $\TCone$ as last rule: we get that
      $\rel{\emptyset}{\termone}{\howe{\cbnpas}}{\termtwo}$ holds.  Then,
      in particular by Lemma~\ref{lemma:keylemma},
      \begin{align*}
        \translop(\termone,\evlabel,\clabstr{\varone}{\setone})&=\sem{\termone}(\abstr{\varone}{\setone})\\
        &\leq
        \sem{\termtwo}(\abstr{\varone}{\howe{\cbnpas}(\setone)})\\ &\leq\sem{\termtwo}(\abstr{\varone}{\tcrel{(\howe{\cbnpas})}(\setone)})\\
        &\leq\sem{\termtwo}(\relone(\clabstr{\varone}{\setone}))=\translop(\termtwo,\evlabel,\relone(\clabstr{\varone}{\setone})).
      \end{align*}
    \item If $\TCtwo$ is the last rule used, we obtain that, for some
      $\termfour\in\LOP(\emptyset)$,
      $\rel{\emptyset}{\termone}{\tcrel{(\howe{\cbnpas})}}{\termfour}$ and
      $\rel{\emptyset}{\termfour}{\tcrel{(\howe{\cbnpas})}}{\termtwo}$
      hold. Then, by induction hypothesis, we get
      \begin{align*}
        \translop(\termone,\evlabel,\setone)&\leq\translop(\termfour,\evlabel,\relone(\setone)),\\
        \translop(\termfour,\evlabel,\relone(\setone))&\leq\translop(\termtwo,\evlabel,\relone(\relone(\setone))).
      \end{align*}
      But of course $\relone(\relone(\setone))\subseteq\relone(\setone)$,
      and as a consequence:
      $$
      \translop(\termone,\evlabel,\setone)\leq\translop(\termtwo,\evlabel,\relone(\setone))
      $$
      and (\ref{equ:probsimCBN}) is satisfied.
    \end{varitemize}
  \item If $\termone\tcrel{(\howe{\cbnpas})}\termtwo$, then for every
    $\termthree\in\LOPp{\emptyset}$ and for every
    $\setone\subseteq\LOPp{\emptyset}$ it holds that
    $$
    \translop(\clabstr{\varone}{\termone},\termthree,\setone)\leq\translop(\clabstr{\varone}{\termtwo},\termthree,\relone(\setone)).
    $$
    But if $\termone\tcrel{(\howe{\cbnpas})}\termtwo$, then
    $\subst{\termone}{\varone}{\termthree}\tcrel{(\howe{\cbnpas})}\subst{\termtwo}{\varone}{\termthree}$.
    This is means that whenever
    $\subst{\termone}{\varone}{\termthree}\in\setone$,
    $\subst{\termtwo}{\varone}{\termthree}\in\howe{\cbnpas}(\setone)\subseteq\tcrel{(\howe{\cbnpas})}(\setone)$
    and ultimately
    \begin{align*}
      \translop(\clabstr{\varone}{\termone},\termthree,\setone)&=1\\ 
      &=\translop(\clabstr{\varone}{\termtwo},\termthree,\tcrel{(\howe{\cbnpas})}(\setone))\\
      &=\translop(\clabstr{\varone}{\termtwo},\termthree,\relone(\setone)).
    \end{align*}
    If $\subst{\termone}{\varone}{\termthree}\notin\setone$, on the other
    hand,
    $$
    \translop(\clabstr{\varone}{\termone},\termthree,\setone)=0\leq\translop(\clabstr{\varone}{\termtwo},\termthree,\relone(\setone)).
    $$
  \end{varitemize}
  This concludes the proof.
\end{proof}

\begin{corollary}\label{cor:pabcongrCBN}
  $\cbnpab$ is a congruence relation for $\LOP$-terms.
\end{corollary}
\begin{proof}
  $\cbnpab$ is an equivalence relation by definition, in particular a
  symmetric relation. Since $\cbnpab = \cbnpas \cap \cbnpas^{op}$ by
  Proposition~\ref{prop:pab=pascopas}, $\cbnpab$ is also compatible as a
  consequence of Theorem~\ref{thm:pasprecongrCBN}.
\end{proof}

\subsection{Proof of the Key Lemma}
As we have already said, Lemma~\ref{lemma:keylemma} is indeed a crucial step towards showing that probabilistic 
applicative simulation is a precongruence.
Proving the Key Lemma~\ref{lemma:keylemma} turns out to be much more difficult than for deterministic
or nondeterministic cases. In particular, the case when $\termone$ is an application relies on another technical lemma we
are now going to give, which itself can be proved by tools from linear programming.

The combinatorial problem we will face while proving the Key Lemma can actually be decontextualized and
understood independently. Suppose we have $\natone=3$ non-disjoint sets 
$\setone_1,\setone_2,\setone_3$ whose elements are labelled with real numbers.
As an example, we could be in a situation like the one in Figure~\ref{fig:vennent} (where for the
sake of simplicity only the labels are indicated).
\begin{figure*}
  \centering
  \subfigure[]
  {\includegraphics[scale=0.8]{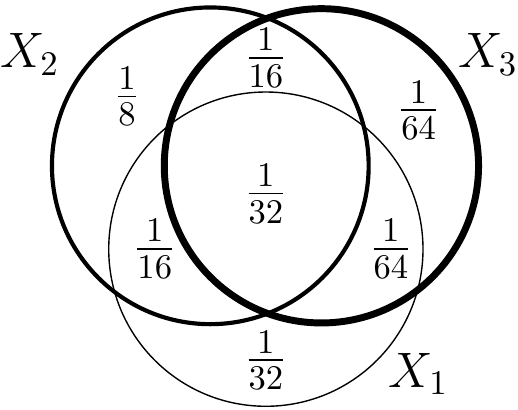}\label{fig:vennent}}
  \hspace{20mm}
  \subfigure[]
  {\includegraphics[scale=0.8]{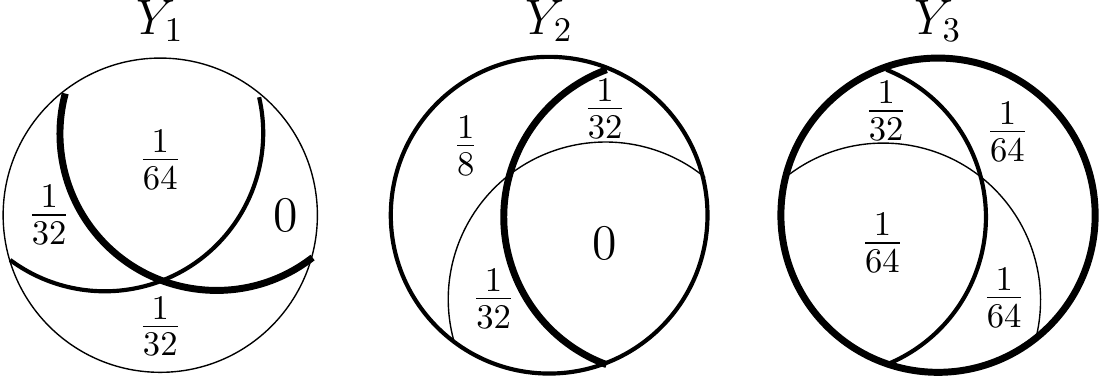}\label{fig:venndisent}}
  \caption{Disentangling Sets}\label{fig:venn}
\end{figure*}
\newcommand{\meas}[1]{||#1||} 
We fix three real numbers $\probone_1\defi\frac{5}{64}$, $\probone_2\defi\frac{3}{16}$,
$\probone_3\defi\frac{5}{64}$. It is routine to check that for every
$\indsetone\subseteq\{1,2,3\}$ it holds that
$$
\sum_{i\in\indsetone}\probone_i\leq\meas{\bigcup_{i\in\indsetone}\setone_i},
$$
where $\meas{\setone}$ is  the sum of the labels of the elements of $\setone$.
Let us observe that it is of course possible to turn the three sets $\setone_1,\setone_2,\setone_3$
into three disjoint sets $\settwo_1$, $\settwo_2$ and $\settwo_3$ where
each $\settwo_i$ contains (copies of) the elements of $\setone_i$ whose labels, 
however, are obtained by splitting the ones of the original elements. Examples
of those sets are in Figure~\ref{fig:venndisent}: if you superpose the
three sets, you obtain the Venn diagram we started from. Quite remarkably, however,
the examples from Figure~\ref{fig:venn} have an additional property, namely
that for every $i\in\{1,2,3\}$ it holds that $\probone_i\leq\meas{\settwo_i}$.
We now show that finding sets satisfying the properties above is always possible, 
even when $\natone$ is arbitrary.

Suppose $\probone_1,\ldots,\probone_n\in\RRN_{[0,1]}$, and suppose that for
each $\indsetone\subseteq\{1,\ldots,n\}$ a real number
$\realone_\indsetone\in\RRN_{[0,1]}$ is defined such that for every such
$\indsetone$ it holds that
$\sum_{i\in\indsetone}\probone_i\leq\sum_{\indsettwo\cap\indsetone\neq\emptyset}
\realone_\indsettwo\leq 1$. Then $(\{\probone_i\}_{1\leq i\leq
  n},\{\realone_\indsetone\}_{\indsetone\subseteq\{1,\ldots,n\}})$ is said
to be a \emph{probability assignment} for $\{1,\ldots,n\}$.
Is it always possible to ``disentangle'' probability assignments? The
answer is positive. 

The following is a formulation of Max-Flow-Min-Cut Theorem:
\begin{theorem}[Max-Flow-Min-Cut]\label{t:mf-mc}
  For any flow network, the value of the maximum flow is equal to the
  capacity of the minimum cut.
\end{theorem}
\begin{lemma}[Disentangling Probability Assignments]\label{lemma:disentangling}
  Let $\passone\defi(\{\probone_i\}_{1\leq i\leq n},\{\realone_\indsetone\}_{\indsetone\subseteq\{1,\ldots,n\}})$
  be a probability assignment.
  Then for every nonempty $\indsetone\subseteq\{1,\ldots,n\}$ and
  for every $k\in\indsetone$ there is $\realtwo_{k,\indsetone}\in\RRN_{[0,1]}$ such that
  the following conditions all hold:
  \begin{varenumerate}
  \item\label{point:first}
      for every $\indsetone$, it holds that $\sum_{k\in\indsetone}\realtwo_{k,\indsetone}\leq 1$;
    \item\label{point:second}
      for every $k\in\{1,\ldots,n\}$, it holds that $\probone_k\leq\sum_{k\in\indsetone}
      \realtwo_{k,\indsetone}\cdot\realone_\indsetone$.
  \end{varenumerate}
\end{lemma}
\begin{proof}
  For every probability assignment $\passone$, let us define the \emph{flow
    network of $\passone$} as the digraph $\fnet{\passone}$ where:
  \begin{varitemize}
  \item 
    $\vt{\passone}\defi(\pow{\{1,\dots,n\}}-\emptyset)\cup\{s,t\}$, where $s,t$ are a
    distinguished source and target, respectively;
  \item 
    $\ed{\passone}$ is composed by three kinds of edges:
    \begin{varitemize}
    \item $(s,\{i\})$ for every $i\in\{1,\dots,n\}$, with an assigned
      capacity of $\probone_i$;
    \item $(\indsetone,\indsetone\cup\{i\})$, for every nonempty
      $\indsetone\subseteq\{1,\dots,n\}$ and $i\not\in\indsetone$, with an
      assigned capacity of $1$;
    \item $(\indsetone,t)$, for every nonempty
      $\indsetone\subseteq\{1,\dots,n\}$, with an assigned capacity of
      $\realone_\indsetone$.
    \end{varitemize}
  \end{varitemize}
  We prove the following two lemmas on $\nt{\passone}$ which
  together entail the result.
  \begin{varitemize}
  \item
    \begin{lemma}
      If $\nt{\passone}$ admits a flow summing to
      $\sum_{i\in\{1,\dots,n\}}\probone_i$, then the $\realtwo_{k,\indsetone}$
      exist for which conditions \ref{point:first}. and
      \ref{point:second}. hold.
    \end{lemma}
    \begin{proof}
      Let us fix $\probone\defi\sum_{i\in\{1,\dots,n\}}\probone_i$. The idea
      then is to start with a flow of value $\probone$ in input to the source
      $s$, which by hypothesis is admitted by $\nt{\passone}$ and the maximum
      one can get, and split it into portions going to singleton vertices
      $\{i\}$, for every $i\in\indsetone$, each of value
      $\probone_i$. Afterwards, for every other vertex
      $\indsetone\subseteq\{1,\dots,n\}$, values of flows on the incoming
      edges are summed up and then distributed to the outgoing adges as one
      wishes, thanks to conservation property of the flow. Formally, a flow
      ${\flone}:\ed{\probone}\rightarrow\RRN_{[0,1]}$ is turned into a
      function $\overline{\flone}:\ed{\probone}\rightarrow(\RRN_{[0,1]})^n$
      defined as follows:
      \begin{varitemize}
      \item For every $i\in\{1,\dots,n\}$,
        $\overline{\flone}_{(s,\{i\})}\defi(0,\dots,\flone_{(s,\{i\})},\dots,0)$,
        where the only possibly nonnull component is exactly the $i$-th;
      \item For every nonempty $\indsetone\subseteq\{1,\dots,n\}$, as soon as
        $\overline{\flone}$ has been defined on all ingoing edges of
        $\indsetone$, we can define it on all its outgoing ones, by just
        splitting each component as we want. This is possible, of course,
        because $\flone$ is a flow and, as such, ingoing and outgoing values
        are the same. More formally, let us fix
        $\overline{\flone}_{(*,\indsetone)}\defi\sum_{\indsetthree\subseteq\{1,\dots,n\}}\overline{\flone}_{(\indsetthree,\indsetone)}$
        and indicate with $\overline{\flone}_{(*,\indsetone),k}$ its $k$-th
        component. Then, for every $i\not\in\indsetone$, we set
        $\overline{\flone}_{(\indsetone,\indsetone\cup\{i\})}\defi(\ratioone_{1,i}\cdot\overline{\flone}_{(*,\indsetone),1},\dots,\ratioone_{n,i}\cdot
        \overline{\flone}_{(*,\indsetone),n})$ where, for every
        $j\in\{1,\dots,n\}$, $\ratioone_{j,i}\in\RRN_{[0,1]}$ are such that
        $\sum_{i\not\in\indsetone}
        \ratioone_{j,i}\cdot\overline{\flone}_{(*,\indsetone),j}=\overline{\flone}_{(*,\indsetone),j}$
        and
        $\sum_{j=1}^n\ratioone_{j,i}\cdot\overline{\flone}_{(*,\indsetone),j}=
        \flone_{(\indsetone,\indsetone\cup\{i\})}$. Of course, a similar
        definition can be given to $\overline{\flone}_{(\indsetone,t)}$, for
        every nonempty $\indsetone\subseteq\{1,\dots,n\}$.
      \end{varitemize}
      Notice that, the way we have just defined $\overline{\flone}$
      guarantees that the sum of all components of $\overline{\flone}_\edone$
      is always equal to $\flone_\edone$, for every
      $\edone\in\ed{\passone}$. Now, for every nonempty
      $\indsetone\subseteq\{1,\dots,n\}$, fix $\realtwo_{k,\indsetone}$ to be
      the ratio $\ratioone_k$ of $\overline{\flone}_{(\indsetone,t)}$; i.e.,
      the $k$-th component of $\overline{\flone}_{(\indsetone,t)}$ (or $0$ if
      the first is itself $0$). On the one hand, for every nonempty
      $\indsetone\subseteq\{1,\dots,n\}$,
      $\sum_{k\in\indsetone}\realtwo_{k,\indsetone}$ is obviously less or
      equal to $1$, hence condition \ref{point:first}. holds. On the other,
      each component of $\overline{\flone}$ is itself a flow, since it
      satisfies the capacity and conservation constraints. Moreover,
      $\nt{\passone}$ is structured in such a way that the $k$-th component
      of $\overline{\flone}_{(\indsetone,t)}$ is $0$ whenever
      $k\not\in\indsetone$. As a consequence, since $\overline{\flone}$
      satisfies the capacity constraint, for every $k\in\{1,\dots,n\}$,
      $$
      \probone_k \leq \sum_{k\in\indsetone}
      \realtwo_{k,\indsetone}\cdot\overline{\flone}_{(\indsetone,t)} \leq
      \sum_{k\in\indsetone} \realtwo_{k,\indsetone}\cdot\realone_\indsetone
      $$
      and so condition \ref{point:second}.  holds too.
    \end{proof}
  \item
    \begin{lemma}
      $\nt{\passone}$ admits a flow summing to
      $\sum_{i\in\indsetone}\probone_i$.
    \end{lemma}
    \begin{proof}
      We prove the result by means of Theorem~\ref{t:mf-mc}. In
      particular, we just prove that the capacity of any cut must be at least
      $\probone\defi\sum_{i\in\{1,\dots,n\}}\probone_i$.
      A cut $(\cutsone,\cuttone)$ is said to be \emph{degenerate}
      if there are $\indsetone\subseteq\{1,\dots,n\}$ and $i\in\{1,\dots,n\}$
      such that $\indsetone\in\cutsone$ and
      $\indsetone\cup\{i\}\in\cuttone$. It is easy to verify that every
      degenerate cut has capacity greater or equal to $1$, thus greater or
      equal to $\probone$. As a consequence, we can just concentrate on
      non-degenerate cuts and prove that all of them have capacity at least
      $\probone$. Given two cuts $\cutone\defi(\cutsone,\cuttone)$ and
      $\cuttwo\defi(\cutstwo,\cutttwo)$, we say that $\cutone\leq\cuttwo$ iff
      $\cutsone\leq\cutstwo$. Then, given $\indsetone\subseteq\{1,\dots,n\}$,
      we call $\indsetone$-cut any cut $(\cutsone,\cuttone)$ such that
      $\bigcup_{\{i\}\in\cutsone}\{i\}=\indsetone$. The \emph{canonical}
      $\indsetone$-cut is the unique $\indsetone$-cut
      $\ccut{\indsetone}\defi(\cutsone,\cuttone)$ such that
      $\cutsone=\{s\}\cup\{\indsettwo\subseteq\{1,\dots,n\}\,|\,\indsettwo\cap\indsetone\neq\emptyset\}$. Please
      observe that, by definition, $\ccut{\indsetone}$ is non-degenerate
      and that the capacity $\ce{\ccut{\indsetone}}$ of $\ccut{\indsetone}$
      is at least $\probone$, because the forward edges in
      $\ccut{\indsetone}$ (those connecting elements of $\cutsone$ to those
      of $\cuttone$) are those going from $s$ to the singletons not in
      $\cutsone$, plus the edges going from any $\indsettwo\in\cutsone$ to
      $t$. The sum of the capacities of such edges are greater or equal to
      $\probone$ by hypothesis.
      We now need to prove the following two lemmas.
      \begin{varitemize}
      \item
        \begin{lemma}\label{l:helper1}
          For every non-degenerate $\indsetone$-cuts $\cutone,\cuttwo$ such
          that $\cutone>\cuttwo$, there is a non-degenerate $\indsetone$-cut
          $\cutthree$ such that $\cutone\geq\cutthree>\cuttwo$ and
          $\ce{\cutthree}\geq\ce{\cuttwo}$.
        \end{lemma}
        \begin{proof}
          Let $\cutone\defi(\cutsone,\cuttone)$ and
          $\cuttwo\defi(\cutstwo,\cutttwo)$. Moreover, let $\indsettwo$ be any
          element of $\cutsone\mathord{\setminus}\cutstwo$. Then, consider
          $\cutthree\defi(\cutstwo\cup\{\indsetthree\subseteq\{1,\dots,n\}\,|\,
          \indsettwo\subseteq\indsetthree\},\cutttwo\mathord{\setminus}\{\indsetthree\subseteq\{1,\dots,n\}\,|\,\indsettwo\subseteq\indsetthree\})$
          and verify that $\cutthree$ is the cut we are looking for. Indeed,
          $\cutthree$ is non-degenerate because it is obtained from $\cuttwo$,
          which is non-degenerate by hypothesis, by adding to it $\indsettwo$
          and all its supersets. Of course, $\cutthree>\cuttwo$. Moreover,
          $\cutone\geq\cutthree$ holds since $\indsettwo\in\cutsone$ and
          $\cutone$ is non-degenerate, which implies $\cutone$ contains all
          supersets of $\indsettwo$ as well. It is also easy to check that
          $\ce{\cutthree}\geq\ce{\cuttwo}$. In fact, in the process of
          constructing $\cutthree$ from $\cuttwo$ we do not lose any forward
          edges coming from $s$, since $\indsettwo$ cannot be a singleton with
          $\cutone$ and $\cuttwo$ both $\indsetone$-cuts, or any other edge
          coming from some element of $\cutstwo$, since $\cuttwo$ is
          non-degenerate.
        \end{proof}
      \item
        \begin{lemma}\label{l:helper2}
          For every non-degenerate $\indsetone$-cuts $\cutone,\cuttwo$ such
          that $\cutone\geq\cuttwo$, $\ce{\cutone}\geq\ce{\cuttwo}$.
        \end{lemma}
        \begin{proof}
          Let $\cutone\defi(\cutsone,\cuttone)$ and
          $\cuttwo\defi(\cutstwo,\cutttwo)$. We prove the result by induction
          on the $n\defi\card{\cutsone}-\card{\cutstwo}$. If $n=0$, then
          $\cutone=\cuttwo$ and the thesis follows. If $n>0$, then
          $\cutone>\cuttwo$ and, by Lemma~\ref{l:helper1}, there is a
          non-degenerate $\indsetone$-cut $\cutthree$ such that
          $\cutone\geq\cutthree>\cuttwo$ and
          $\ce{\cutthree}\geq\ce{\cuttwo}$. By induction hypothesis on
          $\cutone$ and $\cutthree$, it follows that
          $\ce{\cutone}\geq\ce{\cutthree}$. Thus,
          $\ce{\cutone}\geq\ce{\cuttwo}$.
        \end{proof}
      \end{varitemize}
      The two lemmas above permit to conclude. Indeed, for every
      non-degenerate cut $\cuttwo$, there is of course a $\indsetone$ such
      that $\cuttwo$ is a $\indsetone$-cut (possibly with $\indsetone$ as the
      empty set). Then, let us consider the canonical $\ccut{\indsetone}$. On
      the one hand, $\ce{\ccut{\indsetone}}\geq\probone$. On the other, since
      $\ccut{\indsetone}$ is non-degenerate,
      $\ce{\cuttwo}\geq\ce{\ccut{\indsetone}}$ by
      Lemma~\ref{l:helper2}. Hence, $\ce{\cuttwo}\geq\probone$.
    \end{proof}
  \end{varitemize}
  This concludes the main proof.
\end{proof}

In the coming proof of Lemma~\ref{lemma:keylemma} we will widely, and often
implicitly, use the following technical Lemmas. We denote with
$\clabstr{\varone}\cbnpas(\setone)$ the set of distinguished values
$\{\clabstr{\varone}{\termone}\,|\,\exists\termtwo\in\setone.\,\termtwo\cbnpas\termone\}$.
\begin{lemma}\label{lemma:pascommval}
  For every $\setone\subseteq\LOPp{\varone}$,
  $\cbnpas(\clabstr{\varone}{\setone})=\clabstr{\varone}\cbnpas(\setone)$.
\end{lemma}
\begin{proof}
  \begin{align*}
    \clabstr{\varone}{\termone}\in\cbnpas(\clabstr{\varone}{\setone})&\Leftrightarrow
    \exists\termtwo\in\setone.\,\clabstr{\varone}{\termtwo}\cbnpas\clabstr{\varone}{\termone}\\
    &\Leftrightarrow\exists\termtwo\in\setone.\,\termtwo\cbnpas\termone\\
    &\Leftrightarrow\clabstr{\varone}{\termone}\in\clabstr{\varone}{\cbnpas(\setone)}.
  \end{align*}
  This concludes the proof.
\end{proof}

\begin{lemma}\label{lemma:pascomm}
  If $\termone\cbnpas\termtwo$, then for every $\setone\in\LOPp{\varone}$,
  $\sem{\termone}(\abstr{\varone}{\setone})\leq\sem{\termtwo}(\abstr{\varone}{\cbnpas(\setone)})$.
\end{lemma}
\begin{proof}
  If $\termone\cbnpas\termtwo$, then by definition
  $\sem{\termone}(\clabstr{\varone}{\setone})\leq\sem{\termtwo}(\cbnpas(\clabstr{\varone}{\setone}))$.
  Therefore, by Lemma~\ref{lemma:pascommval},
  $\sem{\termtwo}(\cbnpas(\clabstr{\varone}{\setone}))\leq\sem{\termtwo}(\clabstr{\varone}\cbnpas(\setone))$.
\end{proof}
\begin{remark}\label{r:latetech}
  Throughout the following proof we will implicitly use a routine result
  stating that $\termone\cbnpas\termtwo$ implies
  $\sem{\termone}(\abstr{\varone}{\setone})\leq\sem{\termtwo}(\abstr{\varone}{\mathord{\cbnpas}(\setone)})$,
  for every $\setone\subseteq\LOPp{\varone}$. The property
  needed by the latter is precisely the reason why we have  formulated
  $\LOP$ as a multisorted labelled Markov chain: 
  $\mathord{\cbnpas}(\clabstr{\varone}{\setone})$ consists of distinguished
  values only, and is nothing but
  $\clabstr{\varone}{\mathord{\cbnpas}(\setone)}$.
\end{remark}

\begin{proof}[of Lemma~\ref{lemma:keylemma}]
  This is equivalent to proving that if $\relu{\termone}{\howe{\cbnpas}}{\termtwo}$, then 
  for every $\setone\subseteq\LOP(\varone)$ the following implication holds:
  if $\ibsemn{\termone}{\distone}$, then 
  $\distone(\abstr{\varone}{\setone})\leq\sem{\termtwo}(\abstr{\varone}{(\howe{\cbnpas}(\setone)}))$.
  This is an induction on the structure of the proof of $\ibsemn{\termone}{\distone}$.
  \begin{varitemize}
    \item
      If $\distone=\emdist$, then of course
      $\distone(\abstr{\varone}{\setone})=0\leq\sem{\termtwo}(\abstr{\varone}{\settwo})$ for every $\setone,\settwo\subseteq\LOP(\varone)$.
    \item
      If $\termone$ is a value $\abstr{\varone}{\termthree}$ and $\distone(\abstr{\varone}{\termthree})=1$, then the
      proof of $\relu{\termone}{\howe{\cbnpas}}{\termtwo}$ necessarily ends as follows:
      $$
      \infer
          {\rel{\emcon}{\abstr{\varone}{\termthree}}{\howe{\cbnpas}}{\termtwo}}
          {
          \rel{\{\varone\}}{\termthree}{\howe{\cbnpas}}{\termfour}
          &&
          \rel{\emcon}{\abstr{\varone}{\termfour}}{\cbnpas}{\termtwo}
          }
      $$
      Let $\setone$ be any subset of $\LOP(\varone)$. Now, if
      $\termthree\not\in\setone$, then
      $\distone(\abstr{\varone}{\setone})=0$ and the inequality trivially
      holds. If, on the contrary, $\termthree\in\setone$, then
      $\termfour\in\howe{\cbnpas}(\setone)$. Consider
      $\cbnpas(\termfour)$, the set of terms that are in relation
      with $\termfour$ via $\cbnpas$. 
      We have that for every $\termfive\in\;\cbnpas(\termfour)$, both
      $\rel{\{\varone\}}{\termthree}{\howe{\cbnpas}}{\termfour}$ and
      $\rel{\{\varone\}}{\termfour}{\cbnpas}{\termfive}$ hold, and as a
      consequence $\rel{\{\varone\}}{\termthree}{\howe{\cbnpas}}{\termfive}$ does
      (this is a consequence of a property of $\howe{(\cdot)}$, see~\cite{EV}).  
      In other words, $\cbnpas(\termfour)\subseteq\howe{\cbnpas}(\setone)$.
      But then, by Lemma~\ref{lemma:pascomm},
      $$
      \sem{\termtwo}(\abstr{\varone}{\howe{\cbnpas}(\setone)})\geq
      \sem{\termtwo}(\abstr{\varone}{\cbnpas(\termfour)})\geq
      \sem{\abstr{\varone}{\termfour}}(\abstr{\varone}{\termfour})=1.
      $$
    \item
      If $\termone$ is an application $\termthree\termfour$, then $\ibsemn{\termone}{\distone}$ is obtained
      as follows:
      $$
      \infer
          {\ibsemn{\termthree\termfour}{\sum_{\termfive}}\distthree(\abstr{\varone}{\termfive})\cdot\distfive_{\termfive,\termfour}}
          { \ibsemn{\termthree}{\distthree} &&
            \{\ibsemn{\subst{\termfive}{\varone}{\termfour}}{\distfive_{\termfive,\termfour}}\}_{\termfive,\termfour}
          }
          $$
      Moreover, the proof of
      $\rel{\emptyset}{\termone}{\howe{\cbnpas}}{\termtwo}$ must end as
      follows:
      $$
      \infer
          {\rel{\emcon}{\app{\termthree}{\termfour}}{\howe{\cbnpas}}{\termtwo}}
          { \rel{\emcon}{\termthree}{\howe{\cbnpas}}{\termsix} &&
            \rel{\emcon}{\termfour}{\howe{\cbnpas}}{\termseven} &&
          \rel{\emcon}{\app{\termsix}{\termseven}}{\cbnpas}{\termtwo} }
      $$
      Now, since $\ibsemn{\termthree}{\distthree}$ and
      $\rel{\emcon}{\termthree}{\howe{\cbnpas}}{\termsix}$, by induction
      hypothesis we get that for every $\settwo\subseteq\LOP(\varone)$ it
      holds that
      $\distthree(\abstr{\varone}{\settwo})\leq\sem{\termsix}(\abstr{\varone}{\howe{\cbnpas}(\settwo)})$.
      Let us now take a look at the distribution
      $$
      \distone=\sum_{\termfive}\distthree(\abstr{\varone}{\termfive})\cdot\distfive_{\termfive,\termfour}.
      $$
      Since $\distthree$ is a \emph{finite} distribution, the sum above
      is actually the sum of finitely many summands.  Let the support
      $\supp{\distthree}$ of $\distthree$ be
      $\{\abstr{\varone}{\termfive_1},\ldots,\abstr{\varone}{\termfive_n}\}$. 
      It is now time to put the above into a form that is
      amenable to treatment by Lemma~\ref{lemma:disentangling}. Let us
      consider the $n$ sets
      $\howe{\cbnpas}(\termfive_1),\ldots,\howe{\cbnpas}(\termfive_n)$;
      to each term $\termnine$ in them we can associate the probability
      $\sem{\termsix}(\abstr{\varone}{\termnine})$. We are then in the
      scope of Lemma~\ref{lemma:disentangling}, since by induction
      hypothesis we know that for every $\settwo\subseteq\LOP(\varone)$,
      $$
      \distthree(\abstr{\varone}{\setone})\leq\sem{\termsix}(\abstr{\varone}{\howe{\cbnpas}(\setone)}).
      $$
      We can then conclude that for every 
      \begin{align*}
      \termnine\in\howe{\cbnpas}(\{\termfive_1,\ldots,\termfive_n\})=\bigcup_{1\leq
        i\leq n}\howe{\cbnpas}(\termfive_i)
      \end{align*}
      there are $n$ real numbers $\realone_1^{\termnine,\termsix},\ldots,\realone_n^{\termnine,\termsix}$ such that:
      \begin{align*}
        \sem{\termsix}(\abstr{\varone}{\termnine})&\geq
        \sum_{1\leq i\leq n}\realone_i^{\termnine,\termsix}\qquad\forall\,\termnine\in\bigcup_{1\leq i\leq n}\howe{\cbnpas}(\termfive_i);\\
        \distthree(\abstr{\varone}{\termfive_i})&\leq\sum_{\termnine\in\howe{\cbnpas}(\termfive_i)}\realone_i^{\termnine,\termsix}\qquad\forall\,1\leq i\leq n.
      \end{align*}
      So, we can conclude that
      \begin{align*}
        \distone&\leq\sum_{1\leq i\leq
          n}\left(\sum_{\termnine\in\howe{\cbnpas}(\termfive_i)}\realone_i^{\termnine,\termsix}\right)\cdot
        \distfive_{\termfive_i,\termfour}\\
        &=\sum_{1\leq i\leq
          n}\sum_{\termnine\in\howe{\cbnpas}(\termfive_i)}\realone_i^{\termnine,\termsix}\cdot\distfive_{\termfive_i,\termfour}.
      \end{align*}
      Now, whenever $\termfive_i\howe{\cbnpas}\termnine$ and $\termfour\howe{\cbnpas}\termseven$,
      we know that, by Lemma~\ref{lemma:closesubsCBN}, 
      $\subst{\termfive_i}{\varone}{\termfour}\howe{\cbnpas}\subst{\termnine}{\varone}{\termseven}$. We
      can then apply the inductive hypothesis to the $n$ derivations of
      $\ibsemn{\subst{\termfive_i}{\varone}{\termfour}}{\distfive_{\termfive_i,\termfour}}$,
      obtaining that, for every $\setone\subseteq\LOP(\varone)$,
      {\footnotesize
        \begin{align*}
          &\distone(\abstr{\varone}{\setone})\leq\sum_{1\leq i\leq
            n}\sum_{\termnine\in\howe{\cbnpas}(\termfive_i)}
          \realone_i^{\termnine,\termsix}\cdot\sem{\subst{\termnine}{\varone}{\termseven}}(\abstr{\varone}{\howe{\cbnpas}{(\setone)}})\\
          &\leq\sum_{1\leq i\leq
            n}\sum_{\termnine\in\howe{\cbnpas}(\{\termfive_1,\ldots,\termfive_n\})}\realone_i^{\termnine,\termsix}\cdot
            \sem{\subst{\termnine}{\varone}{\termseven}}(\abstr{\varone}{\howe{\cbnpas}{(\setone)}})\\
          &=\sum_{\termnine\in\howe{\cbnpas}(\{\termfive_1,\ldots,\termfive_n\})}\sum_{1\leq
            i\leq
            n}\realone_i^{\termnine,\termsix}\cdot\sem{\subst{\termnine}{\varone}{\termseven}}(\abstr{\varone}{\howe{\cbnpas}{(\setone)}})\\
          &=\sum_{\termnine\in\howe{\cbnpas}(\{\termfive_1,\ldots,\termfive_n\})}\left(\sum_{1\leq
              i\leq n}\realone_i^{\termnine,\termsix}\right)\cdot\sem{\subst{\termnine}{\varone}{\termseven}}(\abstr{\varone}{\howe{\cbnpas}{(\setone)}})\\
          &\leq\sum_{\termnine\in\howe{\cbnpas}(\{\termfive_1,\ldots,\termfive_n\})}
          \sem{\termsix}(\abstr{\varone}{\termnine})\cdot
          \sem{\subst{\termnine}{\varone}{\termseven}}(\abstr{\varone}{\howe{\cbnpas}{(\setone)}})\\
          &\leq\sum_{\termnine\in\LOP(\varone)}
          \sem{\termsix}(\abstr{\varone}{\termnine})\cdot
          \sem{\subst{\termnine}{\varone}{\termseven}}(\abstr{\varone}{\howe{\cbnpas}{(\setone)}})\\
          &=\sem{\termsix\termseven}(\abstr{\varone}{\howe{\cbnpas}{(\setone)}})\leq
          \sem{\termtwo}(\abstr{\varone}{\cbnpas((\howe{\cbnpas})(\setone))})\\
          &\leq\sem{\termtwo}(\abstr{\varone}{\howe{\cbnpas}(\setone)}),
          \end{align*}
        }
        which is the thesis. 
      \item
        If $\termone$ is a probabilistic sum
        $\ps{\termthree}{\termfour}$, then $\ibsemn{\termone}{\distone}$ is
        obtained as follows:
        $$
        \infer
        {\ibsemn{\ps{\termthree}{\termfour}}{\frac{1}{2}\cdot\distthree
            + \frac{1}{2}\cdot\distfour}} {
          \ibsemn{\termthree}{\distthree} && \ibsemn{\termfour}{\distfour}
        }
        $$
        
        Moreover, the proof of
        $\rel{\emptyset}{\termone}{\howe{\cbnpas}}{\termtwo}$ must end as
        follows:
        $$
        \infer
        {\rel{\emcon}{\ps{\termthree}{\termfour}}{\howe{\cbnpas}}{\termtwo}}
        {
          \rel{\emcon}{\termthree}{\howe{\cbnpas}}{\termsix}
          &&
          \rel{\emcon}{\termfour}{\howe{\cbnpas}}{\termseven}
          &&
          \rel{\emcon}{\ps{\termsix}{\termseven}}{\cbnpas}{\termtwo}
        }
        $$
        Now:
        \begin{varitemize}
        \item
          Since $\ibsemn{\termthree}{\distthree}$ and $\rel{\emcon}{\termthree}{\howe{\cbnpas}}{\termsix}$, by induction hypothesis
          we get that for every $\settwo\subseteq\LOP(\varone)$ it holds that 
          $\distthree(\abstr{\varone}{\settwo})\leq\sem{\termsix}(\abstr{\varone}{\howe{\cbnpas}(\settwo)})$;
        \item
          Similarly, since $\ibsemn{\termfour}{\distfour}$ and $\rel{\emcon}{\termfour}{\howe{\cbnpas}}{\termseven}$, by induction hypothesis
          we get that for every $\settwo\subseteq\LOP(\varone)$ it holds that 
          $\distfour(\abstr{\varone}{\settwo})\leq\sem{\termseven}(\abstr{\varone}{\howe{\cbnpas}(\settwo)})$. 
        \end{varitemize}
        Let us now take a look at the distribution 
        $$
        \distone=\frac{1}{2}\cdot\distthree
        + \frac{1}{2}\cdot\distfour.
        $$
        The idea then is to prove that, for every
        $\setone\subseteq\LOP(\varone)$, it holds
        $\distone(\abstr{\varone}{\setone})\leq\sem{\ps{\termsix}{\termseven}}(\abstr{\varone}{\howe{\cbnpas}(\setone)})$.
        In fact, since 
        $\sem{\ps{\termsix}{\termseven}}(\abstr{\varone}{\howe{\cbnpas}(\setone)})\leq\sem{\termtwo}(\abstr{\varone}{\howe{\cbnpas}(\setone)})$,
        the latter would imply the thesis
        $\distone(\abstr{\varone}{\setone})\leq\sem{\termtwo}(\abstr{\varone}{\howe{\cbnpas}(\setone)})$.
        But by induction hypothesis and
        Lemma~\ref{lemma:semsumCBN}:
        \begin{align*}
          \distone(\abstr{\varone}{\setone})&=\frac{1}{2}\cdot\distthree(\abstr{\varone}{\setone})
          +
          \frac{1}{2}\cdot\distfour(\abstr{\varone}{\setone})\\
          &\leq
          \frac{1}{2}\cdot\sem{\termsix}(\abstr{\varone}{\howe{\cbnpas}(\setone)})
          +
          \frac{1}{2}\cdot\sem{\termseven}(\abstr{\varone}{\howe{\cbnpas}(\setone)})\\
          &=
          \sem{\ps{\termsix}{\termseven}}(\abstr{\varone}{\howe{\cbnpas}(\setone)}).
        \end{align*}
      \end{varitemize}
      This concludes the proof.
\end{proof}

\subsection{Context Equivalence}\label{sec:pabce}
We now formally introduce probabilistic context equivalence and prove it
to be coarser than probabilistic applicative bisimilarity.

\begin{definition}\label{def:context}
  A $\LOP$-term context 
  is a syntax tree with a unique ``hole'' $\ctxhole{\cdot}$, generated as follows:
  $$
  \ctxone,\ctxtwo\in\ctxset \, ::= \, \ctxhole{\cdot}\, |\,
  \abstr{\varone}{\ctxone}\, |\, \app{\ctxone}{\termone}\, |\,
  \app{\termone}{\ctxone}\, |\, \ps{\ctxone}{\termone}\, |\,
  \ps{\termone}{\ctxone}.
  $$
We denote with $\ctxone\ctxhole{\termtwo}$ the $\LOP$-term that results
from filling the hole with a $\LOP$-term $\termtwo$:
\begin{align*}
  \ctxhole{\cdot}\ctxhole{\termtwo} &\defi \termtwo;\\
  (\abstr{\varone}{\ctxone})\ctxhole{\termtwo} &\defi
  \abstr{\varone}{\ctxone\ctxhole{\termtwo}};\\
  (\app{\ctxone}{\termone})\ctxhole{\termtwo} &\defi
  \app{\ctxone\ctxhole{\termtwo}}{\termone};\\
  (\app{\termone}{\ctxone})\ctxhole{\termtwo} &\defi
  \app{\termone}{\ctxone\ctxhole{\termtwo}};\\
  (\ps{\ctxone}{\termone})\ctxhole{\termtwo} &\defi
  \ps{\ctxone\ctxhole{\termtwo}}{\termone};\\
  (\ps{\termone}{\ctxone})\ctxhole{\termtwo} &\defi
  \ps{\termone}{\ctxone\ctxhole{\termtwo}}.
\end{align*}
\end{definition}
We also write $\ctxone\ctxhole{\ctxtwo}$ for the context resulting from
replacing the occurrence of $\ctxhole{\cdot}$ in the syntax tree $\ctxone$
by the tree $\ctxtwo$.

We continue to keep track of free variables by sets $\vecvarone$ of
variables and we inductively define subsets
$\ctxsetp{\vecvarone}{\vecvartwo}$ of contexts by the following rules:
$$
\infer[\Ctxone] {\ctxhole{\cdot}\in\ctxsetp{\vecvarone}{\vecvarone}}{ }
$$

$$
\infer[\Ctxtwo]
{\abstr{\varone}{\ctxone}\in\ctxsetp{\vecvarone}{\vecvartwo}}
{\ctxone\in\ctxsetp{\vecvarone}{\vecvartwo\cup\{\varone\}} &&
  \varone\not\in\vecvartwo}
$$
  
$$
\infer[\Ctxthree]
{\app{\ctxone}{\termone}\in\ctxsetp{\vecvarone}{\vecvartwo}}
{\ctxone\in\ctxsetp{\vecvarone}{\vecvartwo} && \termone\in\LOP(\vecvartwo)}
$$

$$
\infer[\Ctxfour]
{\app{\termone}{\ctxone}\in\ctxsetp{\vecvarone}{\vecvartwo}}
{\termone\in\LOP(\vecvartwo) && \ctxone\in\ctxsetp{\vecvarone}{\vecvartwo}
}
$$

$$
\infer[\Ctxfive]
{\ps{\ctxone}{\termone}\in\ctxsetp{\vecvarone}{\vecvartwo}}
{\ctxone\in\ctxsetp{\vecvarone}{\vecvartwo} && \termone\in\LOP(\vecvartwo)}
$$

$$
\infer[\Ctxsix] {\ps{\termone}{\ctxone}\in\ctxsetp{\vecvarone}{\vecvartwo}}
{\termone\in\LOP(\vecvartwo) && \ctxone\in\ctxsetp{\vecvarone}{\vecvartwo}
}
$$
We use double indexing over $\vecvarone$ and $\vecvartwo$ to indicate the
sets of free variables before and after the filling of the hole by a
term. The two following properties explain this idea.

\begin{lemma}\label{lemma:fillholeterm}
  If $\termone\in\LOPp{\vecvarone}$ and
  $\ctxone\in\ctxsetp{\vecvarone}{\vecvartwo}$, then
  $\ctxone\ctxhole{\termone}\in\LOPp{\vecvartwo}$.
\end{lemma}
\begin{proof}
  By induction on the derivation of
  $\ctxone\in\ctxsetp{\vecvarone}{\vecvartwo}$ from the rules
  $\Ctxone$-$\Ctxsix$.
\end{proof}

\begin{lemma}\label{lemma:fillholectx}
  If $\ctxone\in\ctxsetp{\vecvarone}{\vecvartwo}$ and
  $\ctxtwo\in\ctxsetp{\vecvartwo}{\vecvartwo}$, then
  $\ctxtwo\ctxhole\ctxone\in\ctxsetp{\vecvarone}{\vecvartwo}$.
\end{lemma}
\begin{proof}
  By induction on the derivation of
  $\ctxtwo\in\ctxsetp{\vecvartwo}{\vecvartwo}$ from the rules
  $\Ctxone$-$\Ctxsix$.
\end{proof}
Let us recall here the definition of context preorder and equivalence. 

\begin{definition}
  \label{def:ctxeqCBN}
  The \emph{probabilistic context preorder} with respect to
  call-by-name evaluation is the $\LOP$-relation given by
  $\rel{\vecvarone}{\termone}{\cbnconleq}{\termtwo}$ iff 
  $\forall \, \ctxone\in\ctxsetp{\vecvarone}{\emptyset}$,
  $\ctxone\ctxhole{\termone}\evp{\probone}$ implies
  $\ctxone\ctxhole{\termtwo}\evp{\probtwo}$ with $\probone\leq\probtwo$.
  The $\LOP$-relation of \emph{probabilistic context equivalence},
  denoted $\rel{\vecvarone}{\termone}{\cbnconequiv}{\termtwo}$, holds iff
  $\rel{\vecvarone}{\termone}{\cbnconleq}{\termtwo}$ and
  $\rel{\vecvarone}{\termtwo}{\cbnconleq}{\termone}$ do.
\end{definition}

\begin{lemma}\label{lemma:ctxprecCBN}
  The context preorder $\cbnconleq$ is a precongruence relation.
\end{lemma}
\begin{proof}
  Proving $\cbnconleq$ being a precongruence relation means to prove it
  transitive and compatible.  We start by proving $\cbnconleq$ being
  transitive, that is, for every $\vecvarone\in\powfin{\setvar}$ and for
  every $\termone,\,\termtwo,\,\termthree\in\LOPp{\vecvarone}$,
  $\rel{\vecvarone}{\termone}{\cbnconleq}{\termtwo}$ and
  $\rel{\vecvarone}{\termtwo}{\cbnconleq}{\termthree}$ imply
  $\rel{\vecvarone}{\termone}{\cbnconleq}{\termthree}$. By
  Definition~\ref{def:ctxeqCBN}, the latter boils down to prove that,
  the following hypotheses
  \begin{varitemize}
  \item For every $\ctxone$, $\ctxone\ctxhole{\termone}\evp{\probone}$
    implies $\ctxone\ctxhole{\termtwo}\evp{\probtwo}$, with
    $\probone\leq\probtwo$;
  \item For every $\ctxone$, $\ctxone\ctxhole{\termtwo}\evp{\probone}$
    implies $\ctxone\ctxhole{\termthree}\evp{\probtwo}$, with
    $\probone\leq\probtwo$,
  \item $\ctxtwo\ctxhole{\termone}\evp{\probthree}$
  \end{varitemize}
  imply $\ctxtwo\ctxhole{\termthree}\evp{\probfour}$, with
  $\probthree\leq\probfour$. We can easily apply the first hypothesis when
  $\ctxone$ is just $\ctxtwo$, then the second hypothesis (again with
  $\ctxone$ equal to $\ctxtwo$), and get the thesis.
  We prove $\cbnconleq$ of being a compatible relation starting from
  $\Comtwo$ property because $\Comone$ is trivially valid. In particular,
  we must show that, for every $\vecvarone\in\powfin{\setvar}$, for every
  $\varone\in\setvar-\{\vecvarone\}$ and for every $\termone,\,
  \termtwo\in\LOP(\vecvarone\cup\{\varone\})$, if
  $\rel{\vecvarone\cup\{\varone\}}{\termone}{\cbnconleq}{\termtwo}$ then
  $\rel{\vecvarone}{\abstr{\varone}{\termone}}{\cbnconleq}{\abstr{\varone}{\termtwo}}$.
  By Definition~\ref{def:ctxeqCBN}, the latter boils down to prove that,
  the following hypotheses
  \begin{varitemize}
  \item For every $\ctxone$, $\ctxone\ctxhole{\termone}\evp{\probone}$
    implies $\ctxone\ctxhole{\termtwo}\evp{\probtwo}$, with
    $\probone\leq\probtwo$,
  \item $\ctxtwo\ctxhole{\abstr{\varone}{\termone}}\evp{\probthree}$
  \end{varitemize}
  imply $\ctxtwo\ctxhole{\abstr{\varone}{\termtwo}}\evp{\probfour}$, with
  $\probthree\leq\probfour$. Since
  $\ctxtwo\in\ctxsetp{\vecvarone}{\emptyset}$, let us consider the context
  $\abstr{\varone}{\ctxhole{\cdot}}\in\ctxsetp{\vecvarone\cup\{\varone\}}{\vecvarone}$. Then,
  by Lemma~\ref{lemma:fillholectx}, the context $\ctxthree$ of the form
  $\ctxtwo\ctxhole{\abstr{\varone}{\ctxhole{\cdot}}}$ is in
  $\ctxsetp{\vecvarone\cup\{\varone\}}{\emptyset}$. Please note that, by
  Definition~\ref{def:context}, $\ctxtwo\ctxhole{\abstr{\varone}{\termone}} =
  \ctxthree\ctxhole{\termone}$ and, therefore, the second hypothesis can be
  rewritten as $\ctxthree\ctxhole{\termone}\evp{\probthree}$. Thus,
  it follows that $\ctxthree\ctxhole{\termtwo}\evp{\probfour}$, with
  $\probthree\leq\probfour$.
  Moreover, observe that $\ctxthree\ctxhole{\termtwo}$ is nothing else than
  $\ctxtwo\ctxhole{\abstr{\varone}{\termtwo}}$.  Since we have just proved
  $\cbnconleq$ of being transitive, we prove $\Comthree$ property by
  showing that $\ComthreeL$ and $\ComthreeR$ hold. In fact, recall that by
  Lemma~\ref{lemma:com3LR}, the latter two, together, imply the former. In
  particular, to prove $\ComthreeL$ we must show that, for every
  $\vecvarone\in\powfin{\setvar}$ and for every $\termone,\,
  \termtwo,\,\termthree\in\LOP(\vecvarone)$, if
  $\rel{\vecvarone}{\termone}{\cbnconleq}{\termtwo}$ then
  $\rel{\vecvarone}{\app{\termone}{\termthree}}{\cbnconleq}{\app{\termtwo}{\termthree}}$.
  By Definition~\ref{def:ctxeqCBN}, the latter boils down to prove that,
  the following hypothesis
  \begin{varitemize}
  \item For every $\ctxone$, $\ctxone\ctxhole{\termone}\evp{\probone}$
    implies $\ctxone\ctxhole{\termtwo}\evp{\probtwo}$, with
    $\probone\leq\probtwo$,
  \item $\ctxtwo\ctxhole{\app{\termone}{\termthree}}\evp{\probthree}$
  \end{varitemize}
  imply $\ctxtwo\ctxhole{\app{\termtwo}{\termthree}}\evp{\probfour}$, with
  $\probthree\leq\probfour$. Since
  $\ctxtwo\in\ctxsetp{\vecvarone}{\emptyset}$, let us consider the context
  $\app{\ctxhole{\cdot}}{\termthree}\in\ctxsetp{\vecvarone}{\vecvarone}$. Then,
  by Lemma~\ref{lemma:fillholectx}, the context $\ctxthree$ of the form
  $\ctxtwo\ctxhole{\app{\ctxhole{\cdot}}{\termthree}}$ is in
  $\ctxsetp{\vecvarone}{\emptyset}$. Please note that, by
  Definition~\ref{def:context}, $\ctxtwo\ctxhole{\app{\termone}{\termthree}} =
  \ctxthree\ctxhole{\termone}$ and, therefore, the second hypothesis can be
  rewritten as $\ctxthree\ctxhole{\termone}\evp{\probthree}$.  Thus, 
  it follows that $\ctxthree\ctxhole{\termtwo}\evp{\probfour}$, with
  $\probthree\leq\probfour$. 
  Moreover, observe that $\ctxthree\ctxhole{\termtwo}$ is nothing else than
  $\ctxtwo\ctxhole{\abstr{\varone}{\termtwo}}$.
  We do not detail the proof for $\ComthreeR$ that follows the reasoning
  made for $\ComthreeL$, but considering $\ctxthree$ as the context
  $\ctxtwo\ctxhole{\app{\termthree}{\ctxhole{\cdot}}}$.
  Proving $\Comfour$ follows the same pattern resulted for $\Comthree$. In
  fact, by Lemma~\ref{lemma:com4LR}, $\ComfourL$ and $\ComfourR$ together
  imply $\Comfour$. We do not detail the proofs since they proceed the
  reasoning made for $\ComthreeL$, considering the appropriate context each
  time. This concludes the proof.
\end{proof}

\begin{corollary}
  The context equivalence $\cbnconequiv$ is a congruence relation.
\end{corollary}
\begin{proof}
  Straightforward consequence of the definition
  $\cbnconequiv=\cbnconleq\cap\cbnconleq^{\mathit{op}}$.
\end{proof}

\begin{lemma}\label{lemma:ctxcomprel}
  Let $\relone$ be a compatible $\LOP$-relation. If
  $\rel{\vecvarone}{\termone}{\relone}{\termtwo}$ and
  $\ctxone\in\ctxsetp{\vecvarone}{\vecvartwo}$, then
  $\rel{\vecvartwo}{\ctxone\ctxhole{\termone}}{\relone}{\ctxone\ctxhole{\termtwo}}$.
\end{lemma}
\begin{proof}
  By induction on the derivation of
  $\ctxone\in\ctxsetp{\vecvarone}{\vecvartwo}$:
  \begin{varitemize}
  \item If $\ctxone$ is due to $\Ctxone$ then $\ctxone =
    \ctxhole{\cdot}$. Thus, $\ctxone\ctxhole{\termone} = \termone$,
    $\ctxone\ctxhole{\termtwo} = \termtwo$ and the result trivially holds.
  \item If $\Ctxtwo$ is the last rule used, then $\ctxone =
    \abstr{\varone}{\ctxtwo}$, with
    $\ctxtwo\in\ctxsetp{\vecvarone}{\vecvartwo\cup\{\varone\}}$. By
    induction hypothesis, it holds that
    $\rel{\vecvartwo\cup\{\varone\}}{\ctxtwo\ctxhole{\termone}}{\relone}{\ctxtwo\ctxhole{\termtwo}}$. Since
    $\relone$ is a compatible relation, it follows
    $\rel{\vecvartwo}{\abstr{\varone}{\ctxtwo\ctxhole{\termone}}}{\relone}{\abstr{\varone}{\ctxtwo\ctxhole{\termtwo}}}$,
    that is
    $\rel{\vecvartwo}{\ctxone\ctxhole{\termone}}{\relone}{\ctxone\ctxhole{\termtwo}}$.
  \item If $\Ctxthree$ is the last rule used, then $\ctxone =
    \app{\ctxtwo}{\termthree}$, with
    $\ctxtwo\in\ctxsetp{\vecvarone}{\vecvartwo}$ and
    $\termthree\in\LOP(\vecvartwo)$. By induction hypothesis, it holds that
    $\rel{\vecvartwo}{\ctxtwo\ctxhole{\termone}}{\relone}{\ctxtwo\ctxhole{\termtwo}}$.
    Since $\relone$ is a compatible relation, it follows
    $\rel{\vecvartwo}{\app{\ctxtwo\ctxhole{\termone}}{\termthree}}{\relone}{\app{\ctxtwo\ctxhole{\termtwo}}{\termthree}}$,
    which by definition means
    $\rel{\vecvartwo}{(\app{\ctxtwo}{\termthree})\ctxhole{\termone}}{\relone}{(\app{\ctxtwo}{\termthree})\ctxhole{\termtwo}}$. Hence,
    the result
    $\rel{\vecvartwo}{\ctxone\ctxhole{\termone}}{\relone}{\ctxone\ctxhole{\termtwo}}$
    holds. The case of rule $\Ctxfour$ holds by a similar reasoning.
  \item If $\Ctxfive$ is the last rule used, then $\ctxone =
    \ps{\ctxtwo}{\termthree}$, with
    $\ctxtwo\in\ctxsetp{\vecvarone}{\vecvartwo}$ and
    $\termthree\in\LOP(\vecvartwo)$. By induction hypothesis, it holds that
    $\rel{\vecvartwo}{\ctxtwo\ctxhole{\termone}}{\relone}{\ctxtwo\ctxhole{\termtwo}}$.
    Since $\relone$ is a compatible relation, it follows
    $\rel{\vecvartwo}{\ps{\ctxtwo\ctxhole{\termone}}{\termthree}}{\relone}{\ps{\ctxtwo\ctxhole{\termtwo}}{\termthree}}$,
    which by definition means
    $\rel{\vecvartwo}{(\ps{\ctxtwo}{\termthree})\ctxhole{\termone}}{\relone}{(\ps{\ctxtwo}{\termthree})\ctxhole{\termtwo}}$. Hence,
    the result
    $\rel{\vecvartwo}{\ctxone\ctxhole{\termone}}{\relone}{\ctxone\ctxhole{\termtwo}}$
    holds. The case of rule $\Ctxsix$ holds by a similar reasoning.
  \end{varitemize}
  This concludes the proof.
\end{proof}

\begin{lemma}\label{lemma:ctxbisimCBN}
  If $\rel{\vecvarone}{\termone}{\cbnpab}{\termtwo}$ and
  $\ctxone\in\ctxsetp{\vecvarone}{\vecvartwo}$, then
  $\rel{\vecvartwo}{\ctxone\ctxhole{\termone}}{\cbnpab}{\ctxone\ctxhole{\termtwo}}$.
\end{lemma}
\begin{proof}
  Since $\cbnpab = \cbnpas \cap \cbnpas^{\mathit{op}}$ by
  Proposition~\ref{prop:pab=pascopas},
  $\rel{\vecvarone}{\termone}{\cbnpab}{\termtwo}$ implies
  $\rel{\vecvarone}{\termone}{\cbnpas}{\termtwo}$ and
  $\rel{\vecvarone}{\termtwo}{\cbnpas}{\termone}$. Since, by
  Theorem~\ref{thm:pasprecongrCBN}, $\cbnpas$ is a precongruence hence a
  compatible relation,
  $\rel{\vecvartwo}{\ctxone\ctxhole{\termone}}{\cbnpas}{\ctxone\ctxhole{\termtwo}}$
  and
  $\rel{\vecvartwo}{\ctxone\ctxhole{\termtwo}}{\cbnpas}{\ctxone\ctxhole{\termone}}$
  follow by Lemma~\ref{lemma:ctxcomprel}, i.e. 
  $\rel{\vecvartwo}{\ctxone\ctxhole{\termone}}{\cbnpab}{\ctxone\ctxhole{\termtwo}}$.
\end{proof}
\begin{theorem}\label{thm:pab_ce}
  For all $\vecvarone\in\powfin{\setvar}$ and every $\termone,\,
  \termtwo\in\LOP(\vecvarone)$,
  $\rel{\vecvarone}{\termone}{\cbnpab}{\termtwo}$ implies
  $\rel{\vecvarone}{\termone}{\cbnconequiv}{\termtwo}$.
\end{theorem}
\begin{proof}
  If $\rel{\vecvarone}{\termone}{\cbnpab}{\termtwo}$, then for every
  $\ctxone\in\ctxsetp{\vecvarone}{\emptyset}$,
  $\rel{\emptyset}{\ctxone\ctxhole{\termone}}{\cbnpab}{\ctxone\ctxhole{\termtwo}}$
  follows by Lemma~\ref{lemma:ctxbisimCBN}. By
  Lemma~\ref{lemma:sumsempabCBN}, the latter implies
  $\sumsem{\ctxone\ctxhole{\termone}}=\probone=\sumsem{\ctxone\ctxhole{\termtwo}}$. This
  means in particular that $\ctxone\ctxhole{\termone}\evp{\probone}$ iff
  $\ctxone\ctxhole{\termtwo}\evp{\probone}$, which is equivalent to
  $\rel{\vecvarone}{\termone}{\cbnconequiv}{\termtwo}$ by definition.
\end{proof}
The converse inclusion fails. A counterexample 
is described in the following.
\begin{example}\label{ex:count}
  For $\termone\defi\abstr{\varone}{\ps{\termthree}{\termfour}}$
  and $\termtwo\defi\ps{(\abstr{\varone}{\termthree})}{(\abstr{\varone}{\termfour})}$
  (where $\termthree$ is $\abstr{\vartwo}{\Omega}$ and $\termfour$ is $\abstr{\vartwo}{\abstr{\varthree}{\Omega}}$), 
  we have $\termone\not\cbnsimleq\termtwo$, hence
  $\termone\not\cbnpab\termtwo$, but
  $\termone\cbnconequiv\termtwo$.
\end{example}
We prove that the above two terms are context equivalent by means
of \emph{CIU-equivalence}. This is a relation that can be shown to
coincide with context equivalence by a Context Lemma, itself proved by the
Howe's technique.  See Section~\ref{sec:cfctxeq} and
Section~\ref{sec:ciu-eq} for supplementary details on the above
counterexample.  

\section{Context Free Context Equivalence}\label{sec:cfctxeq}

We present here a way of treating the problem of too concrete
representations of contexts: right now, we cannot basically work up-to
$\alpha$-equivalence classes of contexts. Let us dispense with them
entirely, and work instead with a coinductive characterization of the
context preorder, and equivalence, phrased in terms of $\LOP$-relations.

\begin{definition}\label{def:adequateCBN}
  A $\LOP$-relation $\relone$ is said to be \textit{adequate} if, for every
  $\termone,\,\termtwo\in\LOPp{\emptyset}$,
  $\rel{\emptyset}{\termone}{\relone}{\termtwo}$ implies
  $\termone\evp{\probone}$ and $\termtwo\evp{\probtwo}$, with
  $\probone\leq\probtwo$.
\end{definition}
Let us indicate with $\caset$ the collection of all compatible and adequate
$\LOP$-relations and let
\begin{equation}
  \label{eq:cfctxCBN}
  \cbncfleq \defi \bigcup\,\caset.
\end{equation}
It turns out that the context preorder $\cbnconleq$ is the largest
$\LOP$-relation that is both compatible and adequate, that is $\cbnconleq =
\cbncfleq$. Let us proceed towards a proof for the latter.

\begin{lemma}\label{lemma:cfctxcomp}
  For every $\relone,\reltwo\in\caset$, $\relone\circ\reltwo\in\caset$.
\end{lemma}
\begin{proof}
  We need to show that $\relone\circ\reltwo =
  \{(\termone,\termtwo)\,|\,\exists\,\termthree\in\LOPp{\vecvarone}.\,\rel{\vecvarone}{\termone}{\relone}{\termthree}\,
  \wedge\,\rel{\vecvarone}{\termthree}{\reltwo}{\termtwo}\}$ is a
  compatible and adequate $\LOP$-relation.  Obviously,
  $\relone\circ\reltwo$ is adequate: for every
  $(\termone,\termtwo)\in\relone\circ\reltwo$, there exists a term
  $\termthree$ such that
  $\termone\evp{\probone}\Rightarrow\termthree\evp{\probtwo}\Rightarrow\termtwo\evp{\probthree}$,
  with $\probone\leq\probtwo\leq\probthree$. Then,
  $\termone\evp{\probone}\Rightarrow\termtwo\evp{\probthree}$, with
  $\probone\leq\probthree$.  Note that the identity relation $\id \defi
  \{(\termone,\termone)\,|\,\termone\in\LOPp{\vecvarone}\}$ is in
  $\relone\circ\reltwo$. Then, $\relone\circ\reltwo$ is reflexive and, in
  particular, satisfies compatibility property $\Comone$.  Proving
  $\Comtwo$ means to show that, if
  $\rel{\vecvarone\cup\{\varone\}}{\termone}{(\relone\circ\reltwo)}{\termtwo}$,
  then
  $\rel{\vecvarone}{\abstr{\varone}{\termone}}{(\relone\circ\reltwo)}{\abstr{\varone}{\termtwo}}$. From
  the hypothesis, it follows that there exists a term $\termthree$ such
  that $\rel{\vecvarone\cup\{\varone\}}{\termone}{\relone}{\termthree}$ and
  $\rel{\vecvarone\cup\{\varone\}}{\termthree}{\reltwo}{\termtwo}$. Since
  both $\relone$ and $\reltwo$ are in $\caset$, hence compatible, it holds
  $\rel{\vecvarone}{\abstr{\varone}{\termone}}{\relone}{\abstr{\varone}{\termthree}}$
  and
  $\rel{\vecvarone}{\abstr{\varone}{\termthree}}{\reltwo}{\abstr{\varone}{\termtwo}}$. The
  latter together imply
  $\rel{\vecvarone}{\abstr{\varone}{\termone}}{(\relone\circ\reltwo)}{\abstr{\varone}{\termtwo}}$.
  Proving $\Comthree$ means to show that, if
  $\rel{\vecvarone}{\termone}{(\relone\circ\reltwo)}{\termtwo}$ and
  $\rel{\vecvarone}{\termfive}{(\relone\circ\reltwo)}{\termsix}$, then
  $\rel{\vecvarone}{\app{\termone}{\termfive}}{(\relone\circ\reltwo)}{\app{\termtwo}{\termsix}}$.
  From the hypothesis, it follows that there exist two terms
  $\termthree,\,\termfour$ such that, on the one hand,
  $\rel{\vecvarone}{\termone}{\relone}{\termthree}$ and
  $\rel{\vecvarone}{\termthree}{\reltwo}{\termtwo}$, and on the other hand,
  $\rel{\vecvarone}{\termfive}{\relone}{\termfour}$ and
  $\rel{\vecvarone}{\termfour}{\reltwo}{\termsix}$. Since both $\relone$
  and $\reltwo$ are in $\caset$, hence compatible, it holds:
  $$
  \rel{\vecvarone}{\termone}{\relone}{\termthree}\,\wedge\,\rel{\vecvarone}{\termfive}{\relone}{\termfour}\Rightarrow
  \rel{\vecvarone}{\app{\termone}{\termfive}}{\relone}{\app{\termthree}{\termfour}};
  $$
  $$
  \rel{\vecvarone}{\termthree}{\reltwo}{\termtwo}\,\wedge\,\rel{\vecvarone}{\termfour}{\reltwo}{\termsix}\Rightarrow
  \rel{\vecvarone}{\app{\termthree}{\termfour}}{\reltwo}{\app{\termtwo}{\termsix}}.
  $$
  The two together imply
  $\rel{\vecvarone}{\app{\termone}{\termfive}}{(\relone\circ\reltwo)}{\app{\termtwo}{\termsix}}$.

  Proceeding in the same fashion, one can easily prove property $\Comfour$.
\end{proof}

\begin{lemma}\label{lemma:cfprCBNadeq}
  $\LOP$-relation $\cbncfleq$ is adequate.
\end{lemma}
\begin{proof}
  It suffices to note that the property of being \textit{adequate} is
  closed under taking unions of relations. Indeed, if $\relone,\,\reltwo$
  are adequate relations, then it is easy to see that the union
  $\relone\cup\reltwo$ is: for every couple
  $(\termone,\termtwo)\in\relone\cup\reltwo$, either
  $\rel{\vecvarone}{\termone}{\relone}{\termtwo}$ or
  $\rel{\vecvarone}{\termone}{\reltwo}{\termtwo}$. Either way,
  $\termone\evp{\probone}\Rightarrow \termtwo\evp{\probtwo}$, with
  $\probone\leq\probtwo$, implying $\relone\cup\reltwo$ of being adequate.
\end{proof}

\begin{lemma}\label{lemma:cfprCBNprec}
  $\LOP$-relation $\cbncfleq$ is a precongruence.
\end{lemma}
\begin{proof}
  We need to show that $\cbncfleq$ is a transitive and compatible
  relation. By Lemma~\ref{lemma:cfctxcomp},
  $\cbncfleq\circ\cbncfleq\subseteq\cbncfleq$ which implies $\cbncfleq$ of
  being transitive. Let us now prove that $\cbncfleq$ is also compatible.
  Note that the identity relation $\id =
  \{(\termone,\termone)\,|\,\termone\in\LOPp{\vecvarone}\}$ is in $\caset$,
  which implies reflexivity of $\cbncfleq$ and hence, in particular, it
  satisfies property $\Comone$.  It is clear that property $\Comtwo$ is
  closed under taking unions of relations, so that $\cbncfleq$ satisfies
  $\Comtwo$ too. The same is not true for properties $\Comthree$ and
  $\Comfour$. By Lemma~\ref{lemma:com3LR} (respectively,
  Lemma~\ref{lemma:com4LR}), for $\Comthree$ (resp., $\Comfour$) it
  suffices to show that $\cbncfleq$ satisfies $\ComthreeL$ and $\ComthreeR$
  (resp., $\ComfourL$ and $\ComfourR$). This is obvious: contrary to
  $\Comthree$ (resp., $\Comfour$), these properties clearly are closed
  under taking unions of relations.

  This concludes the proof.
\end{proof}

\begin{corollary}
  $\cbncfleq$ is the largest compatible and adequate $\LOP$-relation.
\end{corollary}
\begin{proof}
  Straightforward consequence of Lemma~\ref{lemma:cfprCBNadeq} and
  Lemma~\ref{lemma:cfprCBNprec}.
\end{proof}

\begin{lemma}\label{lemma:ctxeq=cfctxpr}
  $\LOP$-relations $\cbnconleq$ and $\cbncfleq$ coincide.
\end{lemma}
\begin{proof}
  By Definition~\ref{def:ctxeqCBN}, it is immediate that $\cbnconleq$ is
  adequate. Moreover, by Lemma~\ref{lemma:ctxprecCBN}, $\cbnconleq$ is a
  precongruence. Therefore $\cbnconleq\in\caset$ implying
  $\cbnconleq\subseteq\cbncfleq$. Let us prove the converse.
  Since, by Lemma~\ref{lemma:cfprCBNprec}, $\cbncfleq$ is a precongruence
  hence a compatible relation, it holds that, for every
  $\termone,\,\termtwo\in\LOPp{\vecvarone}$ and for every
  $\ctxone\in\ctxsetp{\vecvarone}{\vecvartwo}$,
  $\rel{\vecvarone}{\termone}{\cbncfleq}{\termtwo}$ implies
  $\rel{\vecvartwo}{\ctxone\ctxhole{\termone}}{\cbncfleq}{\ctxone\ctxhole{\termtwo}}$.
  Therefore, for every $\termone,\,\termtwo\in\LOPp{\vecvarone}$ and for
  every $\ctxone\in\ctxsetp{\vecvarone}{\emptyset}$,
  \begin{align*}
    \rel{\vecvarone}{\termone}{\cbncfleq}{\termtwo}&\Rightarrow
    \rel{\emptyset}{\ctxone\ctxhole{\termone}}{\cbncfleq}{\ctxone\ctxhole{\termtwo}}
  \end{align*}
  which implies, by the fact that $\cbncfleq$ is adequate,
  \begin{align*}
    \ctxone\ctxhole{\termone}\evp{\probone}\Rightarrow\ctxone\ctxhole{\termtwo}\evp{\probtwo},\textnormal{
      with }\probone\leq\probtwo
  \end{align*}
  that is, by Definition~\ref{def:ctxeqCBN},
  \begin{align*}
    \rel{\vecvarone}{\termone}{\cbnconleq}{\termtwo}.
  \end{align*}
  This concludes the proof.
\end{proof}

\section{CIU-Equivalence}
\label{sec:ciu-eq}

\renewcommand{\ssemn}[2]{#1\Rightarrow_\mathsf{IN} #2}
\renewcommand{\bsemn}[2]{#1\Downarrow_\mathsf{IN} #2}

CIU-equivalence is a simpler characterization of that kind of program
equivalence we are interested in, i.e., context equivalence. In fact, we
will prove that the two notions coincide. While context equivalence
envisages a quantification over all contexts, CIU-equivalence relaxes such
constraint to a restricted class of contexts without affecting the
associated notion of program equivalence. Such a class of contexts is that
of \emph{evaluation} contexts. In particular, we use a different representation of
evaluation contexts, seeing them as a stack of evaluation frames.

\begin{definition}
  The set of \emph{frame stacks} is given by the following set of rules:
  $$
  \fsone,\fstwo::=\nil\midd\las{\termone}{\fsone}.
  $$
\end{definition}

The set of free variables of a frame stack $\fsone$ can be easily defined
as the union of the variables occurring free in the terms embedded into
it. Given a set of variables $\vecvarone$, define $\stk{\vecvarone}$ as the
set of frame stacks whose free variables are all from $\vecvarone$.
Given a frame stack $\fsone\in\stk{\vecvarone}$ and a term
$\termone\in\LOP(\vecvarone)$, we define the term
$\stktm{\fsone}{\termone}\in\LOP(\vecvarone)$ as follows:
\begin{align*}
  \stktm{\nil}{\termone}&\defi\termone;\\
  \stktm{\las{\termone}{\fsone}}{\termtwo}&\defi\stktm{\fsone}{\termtwo\termone}.
\end{align*}
We now define a binary relation $\cbnfsred$ between pairs of the form
$(\fsone,\termone)$ and \emph{sequences} of pairs in the same form:
\begin{align*}
  (\fsone,\termone\termtwo)&\cbnfsred(\las{\termtwo}{\fsone},\termone);\\
  (\fsone,\ps{\termone}{\termtwo})&\cbnfsred(\fsone,\termone),(\fsone,\termtwo);\\
  (\las{\termone}{\fsone},\abstr{\varone}{\termtwo})&\cbnfsred(\fsone,\subst{\termtwo}{\varone}{\termone}).
\end{align*}
Finally, we define a formal system whose judgments are in the form
$\cbnfscp{\fsone}{\termone}{\probone}$ and whose rules are as follows:
$$
\infer[\Howeredone] {\cbnfscp{\fsone}{\termone}{0}} {}
$$
$$
\infer[\Howeredtwo] {\cbnfscp{\nil}{\valone}{1}} {}
$$
$$
\infer[\Howeredthree]
{\cbnfscp{\fsone}{\termone}{\frac{1}{n}\sum_{i=1}^n\probone_i}}
{(\fsone,\termone)\cbnfsred(\fstwo_1,\termtwo_1),\ldots,(\fstwo_n,\termtwo_n)
  & \cbnfscp{\fstwo_i}{\termtwo_i}{\probone_i}}
$$
The expression $\cbnfssup{\fsone}{\termone}$ stands for the real number
$\sup_{\probone\in\RRN}\cbnfscp{\fsone}{\termone}{\probone}$.  

\begin{lemma}\label{lemma:ciuctxeq}
  For all closed frame stacks $\fsone\in\stk{\emptyset}$ and closed
  $\LOP$-terms $\termone\in\LOPp{\emptyset}$,
  $\cbnfssup{\fsone}{\termone}=\probone$ iff
  $\stktm{\fsone}{\termone}\evp{\probone}$. In particular,
  $\termone\evp{\probone}$ holds iff $\cbnfssup{\nil}{\termone}=\probone$.
\end{lemma}
\begin{proof}
  First of all, we recall here that the work of Dal Lago and
  Zorzi~\cite{DalLagoZorzi} provides various call-by-name inductive
  semantics, either big-steps or small-steps, which are all equivalent. 
  Then, the result can be deduced from the following properties:
  \begin{enumerate}
  \item For all $\fsone\in\stk{\emptyset}$, if
    $\cbnfscp{\fsone}{\termone}{\probone}$ then
    $\exists\distone.\;\ssemn{\stktm{\fsone}{\termone}}{\distone}$ with
    $\sum\distone=\probone$.
    \begin{proof}
      By induction on the derivation of
      $\cbnfscp{\fsone}{\termone}{\probone}$, looking at the last rule
      used.
      \begin{varitemize}
      \item $\Howeredone$ rule used: $\cbnfscp{\fsone}{\termone}{0}$. Then,
        consider the empty distribution $\distone\defi\emptyset$ and observe
        that $\ssemn{\stktm{\fsone}{\termone}}{\distone}$ by
        $\mathsf{se_{n}}$ rule.
      \item $\Howeredtwo$ rule used: $\cbnfscp{\fsone}{\termone}{1}$
        implies $\fsone=\nil$ and $\termone$ of being a value, say
        $\valone$. Then, consider the distribution $\distone\defi\{\valone^1\}$
        and observe that $\ssemn{\stktm{\nil}{\valone}=\valone}{\distone}$
        by $\mathsf{sv_{n}}$ rule. Of course, $\sum\distone=1=\probone$.
      \item $\Howeredthree$ rule used:
        $\cbnfscp{\fsone}{\termone}{\frac{1}{n}\sum_{i=1}^n\probone_i}$
        obtained from
        $(\fsone,\termone)\cbnfsred(\fstwo_1,\termtwo_1),\ldots,(\fstwo_n,\termtwo_n)$
        and, for every $i\in\{1,\dots,n\}$,
        $\cbnfscp{\fstwo_i}{\termtwo_i}{\probone_i}$. Then, by induction
        hypothesis, there exist $\disttwo_1,\ldots,\disttwo_n$ such that
        $\ssemn{\stktm{\fstwo_i}{\termtwo_i}}{\disttwo_i}$ with
        $\sum\disttwo_i=\probone_i$.
          
        Let us now proceed by cases according to the structure of
        $\termone$.
        \begin{varitemize}
        \item If $\termone=\abstr{\varone}{\termthree}$, then
          $\fsone=\las{\termfour}{\fstwo}$ implying $n=1$,
          $\fstwo_1=\fstwo$ and
          $\termtwo_1=\subst{\termthree}{\varone}{\termfour}$. Then,
          consider the distribution $\distone\defi\disttwo_1$ and observe that
          $\stktm{\fsone}{\termone}=\stktm{\las{\termfour}{\fstwo}}{\abstr{\varone}{\termthree}}=
          \stktm{\fstwo}{\app{(\abstr{\varone}{\termthree})}{\termfour}}\rn
          \stktm{\fstwo}{\subst{\termthree}{\varone}{\termfour}}=\stktm{\fstwo_1}{\termtwo_1}$. Hence,
          $\ssemn{\stktm{\fsone}{\termone}}{\distone}$ by $\mathsf{sm_{n}}$
          rule. Moreover,
          $\sum\distone=\sum\disttwo_1=\probone_1=\frac{1}{n}\sum_{i=1}^n\probone_i=\probone$.
        \item If $\termone=\ps{\termthree}{\termfour}$, then $n=2$,
          $\fstwo_1=\fstwo_2=\fsone$, $\termtwo_1=\termthree$ and
          $\termtwo_2=\termfour$. Then, consider the distribution
          $\distone\defi\sum_{i=1}^2\frac{1}{2}\disttwo_i$ and observe that
          $\stktm{\fsone}{\termone}=\stktm{\fsone}{\ps{\termthree}{\termfour}}\rn
          \stktm{\fsone}{\termthree},\stktm{\fsone}{\termfour}=\stktm{\fstwo_1}{\termtwo_1},\stktm{\fstwo_2}{\termtwo_2}$. Hence,
          $\ssemn{\stktm{\fsone}{\termone}}{\distone}$ by $\mathsf{sm_{n}}$
          rule. Moreover,
          $\sum\distone=\sum\sum_{i=1}^2\frac{1}{2}\disttwo_i=\frac{1}{2}\sum_{i=1}^2\sum\disttwo_i=\frac{1}{2}\sum_{i=1}^2\probone_i=p$.
        \item If $\termone=\app{\termthree}{\termfour}$, then $n=1$,
          $\fstwo_1=\las{\termfour}{\fsone}$ and
          $\termtwo_1=\termthree$. Then, consider the distribution
          $\distone\defi\disttwo_1$ and observe that
          $\ssemn{\stktm{\las{\termfour}{\fsone}}{\termthree}}{\disttwo_1}$
          implies $\ssemn{\stktm{\fsone}{\termone}}{\distone}$. Moreover,
          $\sum\distone=\sum\disttwo_1=\probone_1=
          \frac{1}{n}\sum_{i=1}^n\probone_i=p$.
        \end{varitemize}
      \end{varitemize}

      This concludes the proof.
    \end{proof}
  \item For all $\distone$, if $\ssemn{\termone}{\distone}$ then
    $\exists\fsone,\,\termtwo.\;\stktm{\fsone}{\termtwo}=\termone$ and
    $\cbnfscp{\fsone}{\termtwo}{\probone}$ with $\sum\distone=\probone$.
    \begin{proof}
      By induction on the derivation of $\ssemn{\termone}{\distone}$,
      looking at the last rule used. (We refer here to the inductive schema
      of inference rules gave in~\cite{DalLagoZorzi} for small-step
      call-by-name semantics of $\LOP$.)
      \begin{varitemize}
      \item $\mathsf{se_{n}}$ rule used:
        $\ssemn{\termone}{\emptyset}$. Then, for every $\fsone$ and every
        $\termtwo$ such that $\stktm{\fsone}{\termtwo}=\termone$,
        $\cbnfscp{\fsone}{\termtwo}{0}$ by $\Howeredone$ rule. Of course,
        $\sum\distone=0=\probone$.
      \item $\mathsf{sv_{n}}$ rule used: $\termone$ is a value, say
        $\valone$, and $\distone=\{\valone^1\}$ with
        $\ssemn{\valone}{\{\valone^1\}}$. Then, consider $\fsone\defi\nil$ and
        $\termtwo\defi\valone$: by definition
        $\stktm{\fsone}{\termtwo}=\stktm{\nil}{\valone}=\valone=\termone$. By
        $\Howeredtwo$ rule, $\cbnfscp{\nil}{\valone}{1}$ hence
        $\sum\distone=1=\probone$.
      \item $\mathsf{sv_{n}}$ rule used:
        $\ssemn{\termone}{\sum_{i=1}^n\frac{1}{n}\disttwo_i}$ from
        $\termone\rn\termfive_1,\ldots,\termfive_n$ with, for every
        $i\in\{1,\dots,n\}$, $\ssemn{\termfive_i}{\disttwo_i}$. By
        induction hypothesis, for every $i\in\{1,\dots,n\}$, there exist
        $\fstwo_i$ and $\termtwo_i$ such that
        $\stktm{\fstwo_i}{\termtwo_i}=\termfive_i$ and
        $\cbnfscp{\fstwo_i}{\termtwo_i}{\probone_i}$ with
        $\sum\disttwo_i=\probone_i$.

        Let us proceed by cases according to the structure of $\termone$.
        \begin{varitemize}
        \item If $\termone=\app{(\abstr{\varone}{\termthree})}{\termfour}$,
          then $n=1$ and
          $\termfive_1=\subst{\termthree}{\varone}{\termfour}$. Hence,
          consider $\fsone\defi\las{\termfour}{\nil}$ and
          $\termtwo\defi\abstr{\varone}{\termthree}$: by definition,
          $\stktm{\fsone}{\termtwo}=\stktm{\las{\termfour}{\nil}}{\abstr{\varone}{\termthree}}=
          \stktm{\nil}{\app{(\abstr{\varone}{\termthree})}{\termfour}}=\app{(\abstr{\varone}{\termthree})}{\termfour}=
          \termone$. By $\Howeredthree$ rule,
          $(\fsone,\termtwo)=(\las{\termfour}{\nil},\abstr{\varone}{\termthree})\cbnfsred(\nil,\subst{\termthree}{\varone}{\termfour})$
          with, by induction hypothesis result,
          $\cbnfscp{\nil}{\subst{\termthree}{\varone}{\termfour}}{\probone_1}$.
          The latter implies
          $\cbnfscp{\fsone}{\termtwo}{\probone_1}$. Moreover,
          $\sum\distone=\sum\sum_{i=1}^n\frac{1}{n}\disttwo_i=\sum\disttwo_1=\probone_1=p$.
        \item If $\termone=\ps{\termthree}{\termfour}$, then $n=2$,
          $\termfive_1=\termthree$ and $\termfive_2=\termfour$. Hence,
          consider $\fsone\defi\nil$ and $\termtwo\defi\ps{\termthree}{\termfour}$:
          by definition,
          $\stktm{\fsone}{\termtwo}=\stktm{\nil}{\ps{\termthree}{\termfour}}=
          \ps{\termthree}{\termfour}=\termone$. By $\Howeredthree$ rule,
          $(\fsone,\termtwo)=(\nil,\ps{\termthree}{\termfour})\cbnfsred(\nil,\termthree),(\nil,\termfour)$
          with, by induction hypothesis result,
          $\cbnfscp{\nil}{\termthree}{\probone_1}$ and
          $\cbnfscp{\nil}{\termfour}{\probone_2}$. The latter implies
          $\cbnfscp{\fsone}{\termtwo}{\frac{1}{2}\sum_{i=1}^2\probone_i}$. Moreover,
          $\sum\distone=\sum\sum_{i=1}^n\frac{1}{n}\disttwo_i=\sum\sum_{i=1}^n\frac{1}{2}\disttwo_i=
          \frac{1}{2}\sum_{i=1}^2\sum\disttwo_i=\frac{1}{2}\sum_{i=1}^2\probone_i=p$.
        \item If $\termone=\app{\termthree}{\termfour}$ and
          $\termthree\rn\termsix_1,\ldots,\termsix_n$, then
          $\termfive_i=\app{\termsix_i}{\termfour}$ for every
          $i\in\{1,\dots,n\}$. Hence, consider
          $\fsone\defi\las{\termfour}{\nil}$ and $\termtwo\defi\termthree$: by
          definition,
          $\stktm{\fsone}{\termtwo}=\stktm{\las{\termfour}{\nil}}{\termthree}=
          \stktm{\nil}{\app{\termthree}{\termfour}}=\app{\termthree}{\termfour}=
          \termone$. By $\Howeredthree$ rule,
          $(\fsone,\termtwo)=(\las{\termfour}{\nil},\termthree)\cbnfsred(\las{\termfour}{\nil},\termsix_1),\ldots,(\las{\termfour}{\nil},\termsix_n)$
          with, by induction hypothesis result,
          $\cbnfscp{\las{\termfour}{\nil}}{\termsix_i}{\probone_i}$ for
          every $i\in\{1,\dots,n\}$. The latter implies
          $\cbnfscp{\fsone}{\termtwo}{\frac{1}{n}\sum_{i=1}^n\probone_i}$. Moreover,
          $\sum\distone=\sum\sum_{i=1}^n\frac{1}{n}\disttwo_i=
          \frac{1}{n}\sum_{i=1}^n\sum\disttwo_i=\frac{1}{n}\sum_{i=1}^n\probone_i=p$.
        \end{varitemize}
      \end{varitemize}
      This concludes the proof.
    \end{proof}
  \end{enumerate}
  Generally speaking, the two properties above prove the following double
  implication:
  \begin{equation}\label{eq:stacksem}
    \cbnfscp{\fsone}{\termone}{\probone}\;\Longleftrightarrow\;
    \ibsemn{\stktm{\fsone}{\termone}}{\distone} \textnormal{ with }
    \sum\distone=\probone.
  \end{equation}
  Then,
  \begin{align*}
    p=\cbnfssup{\fsone}{\termone}&=\sup_{\probtwo\in\RRN}\cbnfscp{\fsone}{\termone}{\probtwo}
    =\sup_{\ibsemn{\stktm{\fsone}{\termone}}{\distone}}\sum\distone\\
    &=\sum\sup_{\ibsemn{\stktm{\fsone}{\termone}}{\distone}}\distone=
    \sum\sem{\stktm{\fsone}{\termone}}=\stktm{\fsone}{\termone}\evp{\probone},
  \end{align*}
  which concludes the proof.
\end{proof}

Given $\termone,\termtwo\in\LOP(\emptyset)$, we define
$\termone\cbnciuleq\termtwo$ iff for every $\fsone$,
$\cbnfssup{\fsone}{\termone}\leq\cbnfssup{\fsone}{\termtwo}$. This relation
can be extended to a relation on open terms in the usual way. Moreover, we
stipulate $\termone\cbnciuequiv\termtwo$ iff both
$\termone\cbnciuleq\termtwo$ and $\termtwo\cbnciuleq\termone$.

Since $\cbnciuleq$ is a preorder, proving it to be a precongruence boils
down to show the following implication:
$$
\termone\howe{(\cbnciuleq)}\termtwo\Rightarrow\termone\cbnciuleq\termtwo.
$$
Indeed, the converse implication is a consequence of
Lemma~\ref{lemma:howeprop3} and the obvious reflexivity of $\cbnciuleq$
relation. To do that, we extend Howe's construction to frame stacks in a
natural way:
$$
\infer[\Howestkone] {\nil\howe{\relone}\nil} {}
$$
$$
\infer[\Howestktwo]
{(\las{\termone}{\fsone})\howe{\relone}(\las{\termtwo}{\fstwo})} {
  \rel{\emptyset}{\termone}{\howe{\relone}}{\termtwo} &&
  \fsone\howe{\relone}\fstwo }
$$

\begin{lemma}\label{lemma:betaciueq}
  For every $\vecvarone\in\powfin{\setvar}$, it holds
  $\rel{\vecvarone}{\app{(\abstr{\varone}{\termone})}{\termtwo}}{\cbnciuequiv}{\subst{\termone}{\varone}{\termtwo}}$.
\end{lemma}
\begin{proof}
  We need to show that both
  $\rel{\vecvarone}{\app{(\abstr{\varone}{\termone})}{\termtwo}}{\cbnciuleq}{\subst{\termone}{\varone}{\termtwo}}$
  and
  $\rel{\vecvarone}{\subst{\termone}{\varone}{\termtwo}}{\cbnciuleq}{\app{(\abstr{\varone}{\termone}){\termtwo}}}$
  hold. Since $\cbnciuleq$ is defined on open terms by taking closing
  term-substitutions, 
  it suffices to show the result for close $\LOP$-terms only:
  $\app{(\abstr{\varone}{\termone})}{\termtwo}\cbnciuleq\subst{\termone}{\varone}{\termtwo}$
  and
  $\subst{\termone}{\varone}{\termtwo}\cbnciuleq\app{(\abstr{\varone}{\termone})}{\termtwo}$.

  Let us start with
  $\app{(\abstr{\varone}{\termone})}{\termtwo}\cbnciuleq\subst{\termone}{\varone}{\termtwo}$
  and prove that, for every close frame stack $\fsone$,
  $\cbnfssup{\fsone}{\app{(\abstr{\varone}{\termone})}{\termtwo}}\leq\cbnfssup{\fsone}{\subst{\termone}{\varone}{\termtwo}}$. The
  latter is an obvious consequence of the fact
  that 
  $(\fsone,\app{(\abstr{\varone}{\termone})}{\termtwo})$
  reduces to $(\fsone,\subst{\termone}{\varone}{\termtwo})$. Let us look
  into the details distinguishing two cases:
  \begin{varitemize}
  \item If $\fsone=\nil$, then
    $(\fsone,\app{(\abstr{\varone}{\termone})}{\termtwo})\cbnfsred(\las{\termtwo}{\fsone},\abstr{\varone}{\termone})
    \cbnfsred(\fsone,\subst{\termone}{\varone}{\termtwo})$ which implies
    that $\cbnfssup{\fsone}{\app{(\abstr{\varone}{\termone})}{\termtwo}} =
    \sup_{\probone\in\RRN}\cbnfscp{\fsone}{\app{(\abstr{\varone}{\termone})}{\termtwo}}{\probone}
    =
    \sup_{\probone\in\RRN}\cbnfscp{\fsone}{\subst{\termone}{\varone}{\termtwo}}{\probone}
    = \cbnfssup{\fsone}{\subst{\termone}{\varone}{\termtwo}}$.
  \item If $\fsone=\las{\termthree}{\fstwo}$, then we can proceed
    similarly.
  \end{varitemize}
    
  \noindent{}Similarly, to prove the converse,
  $\subst{\termone}{\varone}{\termtwo}\cbnciuleq\app{(\abstr{\varone}{\termone})}{\termtwo}$,
  let us fix $\probone$ as
  $\cbnfscp{\fsone}{\subst{\termone}{\varone}{\termtwo}}{\probone}$ and
  distinguish two cases:
  \begin{varitemize}
  \item If $\fsone=\nil$ 
    and $\probone = 0$, then
    $\cbnfscp{\fsone}{\app{(\abstr{\varone}{\termone})}{\termtwo}}{0}$
    holds too by $\Howeredone$ rule. Otherwise,
    $$
    \infer[\Howeredthree]
    {\cbnfscp{\fsone}{\app{(\abstr{\varone}{\termone})}{\termtwo}}{\probone}
    }
    {(\fsone,\termone)\cbnfsred(\las{\termtwo}{\fsone},\abstr{\varone}{\termone})
      & \infer[\Howeredthree]
      {\cbnfscp{\las{\termtwo}{\fsone}}{\abstr{\varone}{\termone}}{\probone}
      }
      {(\las{\termtwo}{\fsone},\abstr{\varone}{\termone})\cbnfsred(\fsone,\subst{\termone}{\varone}{\termtwo})
        & \cbnfscp{\fsone}{\subst{\termone}{\varone}{\termtwo}}{\probone}}
    }
    $$
    which implies $\cbnfssup{\fsone}{\subst{\termone}{\varone}{\termtwo}} =
    \sup_{\probone\in\RRN}\cbnfscp{\fsone}{\subst{\termone}{\varone}{\termtwo}}{\probone}
    =
    \sup_{\probone\in\RRN}\cbnfscp{\fsone}{\app{(\abstr{\varone}{\termone})}{\termtwo}}{\probone}
    = \cbnfssup{\fsone}{\app{(\abstr{\varone}{\termone})}{\termtwo}}$.
  \item If $\fsone=\las{\termthree}{\fstwo}$, then we can proceed
    similarly.
  \end{varitemize}
  This concludes the proof.
\end{proof}

\begin{lemma}\label{lemma:keyciu}
  For every $\fsone,\fstwo\in\stk{\emptyset}$ and
  $\termone,\termtwo\in\LOP(\emptyset)$, if
  $\fsone\howe{(\cbnciuleq)}\fstwo$ and
  $\termone\howe{(\cbnciuleq)}\termtwo$ and
  $\cbnfscp{\fsone}{\termone}{\probone}$, then
  $\cbnfssup{\fstwo}{\termtwo}\geq\probone$.
\end{lemma}
\begin{proof}
  We go by induction on the structure of the proof of
  $\cbnfscp{\fsone}{\termone}{\probone}$, looking at the last rule used.
  \begin{varitemize}
  \item
    If $\cbnfscp{\fsone}{\termone}{0}$, then trivially $\cbnfssup{\fstwo}{\termtwo}\geq 0$.
  \item If $\fsone=\nil$, $\termone=\abstr{\varone}{\termthree}$ and
    $\probone=1$, then $\fstwo=\nil$ since
    $\fsone\howe{(\cbnciuleq)}\fstwo$. From
    $\termone\howe{(\cbnciuleq)}\termtwo$, it follows that there is
    $\termfour$ with
    $\rel{\varone}{\termthree}{\howe{(\cbnciuleq)}}{\termfour}$ and
    $\rel{\emptyset}{\abstr{\varone}{\termfour}}{\cbnciuleq}{\termtwo}$.
    But the latter implies that $\cbnfssup{\nil}{\termtwo}\geq 1$, which is
    the thesis.
  \item Otherwise, $\Howeredthree$ rule is used and suppose we are in the
    following situation
    $$
    \infer[\Howeredthree] {\cbnfscp{\fsone}{\termone}{\frac{1}{n}\sum_{i=1}^n\probone_i}}
    {(\fsone,\termone)\cbnfsred(\fsthree_1,\termthree_1),\ldots,(\fsthree_n,\termthree_n)
      & \cbnfscp{\fsthree_i}{\termthree_i}{\probone_i}}
    $$
    Let us distinguish the following cases as in definition of $\cbnfsred$:
    \begin{varitemize}
    \item If $\termone=\termfour\termfive$, then $n=1$,
      $\fsthree_1=\las{\termfive}{\fsone}$ and $\termthree_1=\termfour$.
      From $\termone\howe{(\cbnciuleq)}\termtwo$ it follows that there are
      $\termsix,\termseven$ with
      $\rel{\emptyset}{\termfour}{\howe{(\cbnciuleq)}}{\termsix}$,
      $\rel{\emptyset}{\termfive}{\howe{(\cbnciuleq)}}{\termseven}$ and
      $\rel{\emptyset}{\app{\termsix}{\termseven}}{\cbnciuleq}{\termtwo}$.
      But then we can form the following:
      $$
      \infer[\Howestktwo]
      {\rel{\emptyset}{\fsthree_1}{\howe{(\cbnciuleq)}}{\las{\termseven}{\fstwo}}}
      {\rel{\emptyset}{\termfive}{\howe{(\cbnciuleq)}}{\termseven} &
        \rel{\emptyset}{\fsone}{\howe{(\cbnciuleq)}}{\fstwo}}
      $$
      and, by the induction hypothesis, conclude that
      $\cbnfssup{\las{\termseven}{\fstwo}}{\termsix}\geq\probone$. Now
      observe that
      $$
      (\fstwo,\termsix\termseven)\cbnfsred(\las{\termseven}{\fstwo},\termsix),
      $$
      and, as a consequence,
      $\cbnfssup{\fstwo}{\termsix\termseven}\geq\probone$, from which the
      thesis easily follows given that
      $\rel{\emptyset}{\app{\termsix}{\termseven}}{\cbnciuleq}{\termtwo}$.
    \item If $\termone=\ps{\termfour}{\termfive}$, then $n=2$,
      $\fsthree_1=\fsthree_2=\fsone$ and $\termthree_1=\termfour$,
      $\termthree_2=\termfive$. From $\fsone\howe{(\cbnciuleq)}\fstwo$, we
      get that $\fsthree_1\howe{(\cbnciuleq)}\fstwo$ and
      $\fsthree_2\howe{(\cbnciuleq)}\fstwo$. From
      $\termone\howe{(\cbnciuleq)}\termtwo$ it follows that there are
      $\termsix,\termseven$ with
      $\rel{\emptyset}{\termfour}{\howe{(\cbnciuleq)}}{\termsix}$,
      $\rel{\emptyset}{\termfive}{\howe{(\cbnciuleq)}}{\termseven}$ and
      $\rel{\emptyset}{\ps{\termsix}{\termseven}}{\cbnciuleq}{\termtwo}$. Then,
      by a double induction hypothesis, it follows
      $\cbnfssup{\fstwo}{\termsix}\geq\probone$ and
      $\cbnfssup{\fstwo}{\termseven}\geq\probone$. The latter together
      imply $\cbnfssup{\fstwo}{\ps{\termsix}{\termseven}}\geq\probone$,
      from which the thesis easily follows given that
      $\rel{\emptyset}{\ps{\termsix}{\termseven}}{\cbnciuleq}{\termtwo}$.
    \item If $\termone=\abstr{\varone}{\termfour}$, then
      $\fsone=\las{\termfive}{\fsthree}$ because the only case left. Hence
      $n=1$, $\fsthree_1=\fsthree$ and
      $\termthree_1=\subst{\termfour}{\varone}{\termfive}$. From
      $\fsone\howe{(\cbnciuleq)}\fstwo$, we get that
      $\fstwo=\las{\termsix}{\fsfour}$ where
      $\rel{\emptyset}{\termfive}{\howe{(\cbnciuleq)}}{\termsix}$ and
      $\fsthree\howe{(\cbnciuleq)}\fsfour$. From
      $\termone\howe{(\cbnciuleq)}\termtwo$, it follows that for some
      $\termseven$, it holds that
      $\rel{\varone}{\termfour}{\howe{(\cbnciuleq)}}{\termseven}$ and
      $\rel{\emptyset}{\abstr{\varone}{\termseven}}{\cbnciuleq}{\termtwo}$. Now:
      \begin{equation}\label{eq:fsred}
        (\fstwo,\abstr{\varone}{\termseven})=(\las{\termsix}{\fsfour},\abstr{\varone}{\termseven})
        \cbnfsred(\fsfour,\subst{\termseven}{\varone}{\termsix}).
      \end{equation}
      From $\rel{\varone}{\termfour}{\howe{(\cbnciuleq)}}{\termseven}$ and
      $\rel{\emptyset}{\termfive}{\howe{(\cbnciuleq)}}{\termsix}$, by
      substitutivity of $\cbnciuleq$, follow that
      $\rel{\emptyset}{\subst{\termfour}{\varone}{\termfive}}{\howe{(\cbnciuleq)}}{\subst{\termseven}{\varone}{\termsix}}$
      holds. By induction hypothesis, it follows that
      $\cbnfssup{\fsfour}{\subst{\termseven}{\varone}{\termsix}}\geq\probone$. Then,
      from (\ref{eq:fsred}) and
      $\rel{\emptyset}{\abstr{\varone}{\termseven}}{\cbnciuleq}{\termtwo}$,
      the thesis easily follows:
      $$
      \cbnfssup{\fstwo}{\termtwo}\geq\cbnfssup{\fstwo}{\abstr{\varone}{\termseven}}=
      \cbnfssup{\fsfour}{\subst{\termseven}{\varone}{\termsix}}\geq\probone.
      $$
    \end{varitemize}
  \end{varitemize}
  This concludes the proof.
\end{proof}

\begin{theorem}\label{thm:ciupr=ctxpr}
  For all $\vecvarone\in\powfin{\setvar}$ and for all
  $\termone,\,\termtwo\in\LOPp{\vecvarone}$,
  $\rel{\vecvarone}{\termone}{\cbnciuleq}{\termtwo}$ iff
  $\rel{\vecvarone}{\termone}{\cbnconleq}{\termtwo}$.
\end{theorem}
\begin{proof}
  ($\Rightarrow$) Since $\cbnciuleq$ is defined on open terms by taking
  closing term-substitutions, by Lemma~\ref{lemma:closesubsCBN} both it and
  $\howe{(\cbnciuleq)}$ are closed under term-substitution. Then, it
  suffices to show the result for closed $\LOP$-terms: for all
  $\termone,\,\termtwo\in\LOPp{\emptyset}$, if
  $\rel{\emptyset}{\termone}{\cbnciuleq}{\termtwo}$, then
  $\rel{\emptyset}{\termone}{\cbnconleq}{\termtwo}$.  Since $\cbnciuleq$ is
  reflexive, by Lemma~\ref{lemma:howeprop1} follows that
  $\howe{(\cbnciuleq)}$ is compatible, hence reflexive too. Taking $\fstwo
  = \fsone$ in Lemma~\ref{lemma:keyciu}, we conclude that
  $\rel{\emptyset}{\termone}{\howe{(\cbnciuleq)}}{\termtwo}$ implies
  $\rel{\emptyset}{\termone}{\cbnciuleq}{\termtwo}$. As we have remarked
  before the lemma, the latter entails that $\howe{(\cbnciuleq)} =
  \cbnciuleq$ which implies $\cbnciuleq$ of being compatible. Moreover,
  from Lemma~\ref{lemma:ciuctxeq} immediately follows that $\cbnciuleq$ is
  also adequate. Thus, $\cbnciuleq$ is contained in the largest compatible
  adequate $\LOP$-relation, $\cbncfleq$. From
  Lemma~\ref{lemma:ctxeq=cfctxpr} follows that $\cbnciuleq$ is actually
  contained in $\cbnconleq$. In particular, the latter means
  $\rel{\emptyset}{\termone}{\cbnciuleq}{\termtwo}$ implies
  $\rel{\emptyset}{\termone}{\cbnconleq}{\termtwo}$.

  ($\Leftarrow$) First of all, please observe that, since context
  preorder is compatible, if
  $\rel{\emptyset}{\termone}{\cbnconleq}{\termtwo}$ then, for all
  $\fsone\in\stk{\emptyset}$,
  $\rel{\emptyset}{\stktm{\fsone}{\termone}}{\cbnconleq}{\stktm{\fsone}{\termtwo}}$
  by Lemma~\ref{lemma:ctxcomprel}. Then, by adequacy property of
  $\cbnconleq$ and Lemma~\ref{lemma:ciuctxeq}, the latter implies
  $\rel{\emptyset}{\termone}{\cbnciuleq}{\termtwo}$. Ultimately, it holds
  that $\rel{\emptyset}{\termone}{\cbnconleq}{\termtwo}$ implies
  $\rel{\emptyset}{\termone}{\cbnciuleq}{\termtwo}$. Let us take into
  account the general case of open terms.  If
  $\rel{\vecvarone}{\termone}{\cbnconleq}{\termtwo}$, then by compatibility
  property of $\cbnconleq$ it follows
  $\rel{\emptyset}{\abstr{\vecvarone}{\termone}}{\cbnconleq}{\abstr{\vecvarone}{\termtwo}}$
  and hence
  $\rel{\emptyset}{\abstr{\vecvarone}{\termone}}{\cbnciuleq}{\abstr{\vecvarone}{\termtwo}}$. Then,
  from the fact that $\cbnciuleq$ is compatible (as established in
  $(\Rightarrow)$ part of this proof) and Lemma~\ref{lemma:betaciueq}, for
  every suitable $\vectermthree\subseteq\LOPp{\emptyset}$, it holds
  $\rel{\emptyset}{\subst{\termone}{\vecvarone}{\vectermthree}}{\cbnciuleq}{\subst{\termtwo}{\vecvarone}{\vectermthree}}$,
  i.e. $\rel{\vecvarone}{\termone}{\cbnciuleq}{\termtwo}$.
\end{proof}

\begin{corollary}\label{cor:ciueq=ctxeq}
  $\cbnciuequiv$ coincides with $\cbnconequiv$.
\end{corollary}
\begin{proof}
  Straightforward consequence of Theorem~\ref{thm:ciupr=ctxpr}.
\end{proof}

\begin{proposition}
  $\cbnconleq$ and $\cbnsimleq$ do not coincide.
\end{proposition}
\begin{proof}
  We will prove that $\termone\cbnciuleq\termtwo$ but
  $\termone\not\cbnsimleq\termtwo$, where
  \begin{align*}
    \termone&\defi\abstr{\varone}{\abstr{\vartwo}{\ps{\varone}{\vartwo}}};\\
    \termtwo&\defi\ps{(\abstr{\varone}{\abstr{\vartwo}{\varone}})}{(\abstr{\varone}{\abstr{\vartwo}{\vartwo}})}.
  \end{align*}
  $\termone\not\cbnsimleq\termtwo$ can be easily verified, so let us
  concentrate on $\termone\cbnciuleq\termtwo$, and prove that for every
  $\fsone$,
  $\cbnfssup{\fsone}{\termone}\leq\cbnfssup{\fsone}{\termtwo}$. Let us
  distinguish three cases:
  \begin{varitemize}
  \item If $\fsone=\nil$, then $(\fsone,\termone)$ cannot be further
    reduced and
    $(\fsone,\termtwo)\cbnfsred(\fsone,\abstr{\varone}{\abstr{\vartwo}{\varone}}),(\fsone,\abstr{\varone}{\abstr{\vartwo}{\vartwo}})$,
    where the last two pairs cannot be reduced. As a consequence,
    $\cbnfssup{\fsone}{\termone}=0=\cbnfssup{\fsone}{\termtwo}$.
  \item If $\fsone=\las{\termthree}{\fstwo}$, then we can proceed
    similarly.
  \item If $\fsone=\las{\termthree}{\las{\termfour}{\fstwo}}$, then observe
    that
    \begin{align*}
      (\fsone,\termone)&\cbnfsred(\las{\termfour}{\fstwo},\abstr{\vartwo}{\ps{\termthree}{\vartwo}})\cbnfsred(\fstwo,\ps{\termthree}{\termfour})\\
      &\cbnfsred(\fstwo,\termthree),(\fstwo,\termfour);\\
      (\fsone,\termtwo)&\cbnfsred(\fsone,\abstr{\varone}{\abstr{\vartwo}{\varone}}),(\fsone,\abstr{\varone}{\abstr{\vartwo}{\vartwo}});\\
      (\fsone,\abstr{\varone}{\abstr{\vartwo}{\varone}})&\cbnfsred(\las{\termfour}{\fstwo},\abstr{\vartwo}{\termthree})\cbnfsred(\fstwo,\termthree);\\
      (\fsone,\abstr{\varone}{\abstr{\vartwo}{\vartwo}})&\cbnfsred(\las{\termfour}{\fstwo},\abstr{\vartwo}{\vartwo})\cbnfsred(\fstwo,\termfour).\\
    \end{align*}
    As a consequence,
    $$
    \cbnfssup{\fsone}{\termone}=\frac{1}{2}\cbnfssup{\fstwo}{\termthree}+\frac{1}{2}\cbnfssup{\fstwo}{\termfour}=\cbnfssup{\fsone}{\termtwo}.
    $$
  \end{varitemize}
  This concludes the proof.
\end{proof}

\begin{example}
  We consider again the programs from Example~\ref{ex:exp}.  Terms
  \hsk{expone} and \hsk{exptwo} only differ because the former
  performs all probabilistic choices on {natural numbers} obtained by
  applying a function to its argument, while in the latter choices are
  done at the functional level, and the argument to those functions is
  provided only at a later stage.  As a consequence, the two terms are
  not applicative bisimilar, and the reason is akin to that for
  the   inequality of the terms in 
  Example~\ref{ex:count}. 
  In contrast, the bisimilarity between  \hsk{expone} and \hsk{expthree
  k}, where \hsk{k} is any natural number, intuitively holds because
  both 
  \hsk{expone} and \hsk{expthree
    k} evaluate to a single term when fed with a
  function, while they start evolving in a genuinely probabilistic way
  only after the second argument is provided. At
  that point, the two functions evolve in very different ways, but
  their semantics (in the sense of Section~\ref{sect:p}) is 
  the same (cf.,  Lemma~\ref{lemma:samesem}).
  As a bisimulation one can use the equivalence generated by the relation
  \begin{align*}
    &\Bigl(\bigcup_{\hsk{k}}\{( \hsk{expone}, \hsk{expthree k})\}\Bigr) \cup
    \{(M,N)\; | \:  {\sem M} = {\sem N }\}\\
    &\cup \Bigl(\bigcup_L \{ (\lambda \hsk{n}. \hsk{B} \sub L {\tt{f}},
    \lambda \hsk{n}. \hsk{C} \sub L {\tt{f}} )\} \Bigr)
  \end{align*}
  using $\tt{B}$ and $\tt{C}$ for the body of $\hsk{expone}$ and
  $\hsk{expthree}$ respectively.
\end{example}

\section{The Discriminating Power of Probabilistic Contexts}\label{sect:dppc}
\setlength{\unitlength}{15pt}
We show here that applicative bisimilarity and context equivalence
collapse if the tested terms are pure, \emph{deterministic}, $\lambda$-terms. In other words, if the
probabilistic choices are brought into the terms only through the inputs supplied to the tested functions, 
applicative bisimilarity and context equivalence yield exactly the same discriminating power.
To show this, we prove that, on pure $\lambda$-terms, both relations coincide 
with the \emph{Levy-Longo tree equality}, which equates terms with the same Levy-Longo tree (briefly
LLT) \cite{DezG01}.  

LLT's are the lazy variant of B\"{o}hm Trees (briefly BT), the most popular
tree structure in the $\lambda$-calculus.  BT's only correctly express the
computational content of $\lambda$-terms in a \emph{strong} regime, while
they fail to do so in the lazy one.  For instance, the term $\lambda x.
\Omega $ and $\Omega$, as both {\em unsolvable}~\cite{Barendregt84}, have
identical BT's, but in a lazy regime we would always distinguish between
them; hence they have different LLT's.  
LLT's were introduced by 
Longo~\cite{Longo}, 
developing an original idea by Levy~\cite{Levy75}.
The \emph{Levy-Longo tree} of $M$, $LT(M)$, is coinductively constructed as follows:
$LT(M) \defi \lambda x_1. \ldots x_n.  \bot$ if $M$ is an unsolvable of order
$n$; $LT(M) \defi \top$ if $M$ is an unsolvable of order $\infty$; finally if
$M$ has principal head normal form $\lambda x_1. \ldots x_n . y M_1 \ldots
M_m$, then $LT(M)$ is a tree with root $\lambda x_1. \ldots x_n . y$ and
with $LT(M_1), \ldots, LT(M_m)$ as subtrees. Being defined coinductively,
LLT's can of course be infinite. We write $\termone\LLT\termtwo$ iff
$LT(\termone)=LT(\termtwo)$.

\begin{example}
\label{e:mn}
Let $\Xi$ be an unsolvable of order $\infty$ such as $\Xi \defi (\lambda x
. \lambda y . ( x x) )(\lambda x . \lambda y . ( x x) )$, and consider the
terms
$$
M \defi \lambda  x. ( x (\lambda  y. (x \Xi \Omega y))\Xi);
\qquad 
N \defi \lambda x. ( x (x \Xi \Omega)\Xi).
$$ 
These terms have been used to prove non-full-abstraction results in a
canonical model for the lazy $\lambda$-calculus by Abramsky and 
Ong~\cite{AbramskyOng93}. For this, they show that in the model the convergence test is
definable (this operator, when it receives an argument, would return the
identity function if the supplied argument is convergent, and would diverge
otherwise). The convergence test, $\nabla$, can distinguish between the two
terms, as $M \nabla$ reduces to an abstraction, whereas $N \nabla$
diverges.  However, no pure $\lambda$-term can make the same distinction.  The
two terms also have different LL trees:
\begin{center}
  \begin{tabular}{lll}
    $
    \begin{picture}(8,4)(0,0)      
      \put(0,3.3){ $LT(M) =$}
      \put(4.2,3.5){\makebox(0,0){$\lambda x . x$}}
      \put(4.2,3){\line(-1,-1){.7}}
      \put(4.2,3){\line(1,-1){.7}}
      \put(3.2,1.9){\makebox(0,0){$ \lambda y. x$}}
      \put(5.2,1.9){\makebox(0,0){$\top$}}

      \put(3.2,1.5){\line(0,-1){.7}}
      \put(3.2,1.5){\line(1,-1){.7}}
      \put(3.2,1.5){\line(-1,-1){.7}}
      \put(1.9,0.1){$\top$}
      \put(2.9,0.1){$\bot$}
      \put(3.9,0.1){$y$}
    \end{picture}$ 
    &
    \hspace{3mm}
    &
    $
    \begin{picture}(8,4)(0,0)
      \put(-1,3.3){$LT(N) =$}
      \put(3.2,3.5){\makebox(0,0){$\lambda x . x$}}
      \put(3.2,3){\line(-1,-1){.7}}
      \put(3.2,3){\line(1,-1){.7}}
      \put(2.2,1.9){\makebox(0,0){$x$}}
      \put(4.2,1.9){\makebox(0,0){$\top$}}
      \put(2.2,1.5){\line(1,-1){.7}}
      \put(2.2,1.5){\line(-1,-1){.7}}
      \put(.9,.1){$\top$}
      \put(2.9,.1){$\bot$}
    \end{picture}$
  \end{tabular}
\end{center}
\noindent{}Although in $\LaPRO$, as in $\Lambda$, the convergence test operator is not
definable, $M$ and $N$ can be separated using probabilities by running
them in a context $\qct$ that would feed $\Omega \oplus \lambda z. \lambda
u. z$ as argument; then $\ctxone\ctxhole{\termone}\evp{\frac{1}{2}}$
whereas $\ctxone\ctxhole{\termtwo}\evp{\frac{1}{4}}$.
\end{example}

\begin{example}
\label{e:mn2}
Abramsky's canonical model is itself coarser than LLT equality. For instance,
the terms $M \defi \lambda x .x x$ and $N \defi \lambda x.  (x \lambda y.  (x
y))$, have different LLT's but are equal in Abramsky's model (and hence
equal for context equivalence in $\Lambda$).  They are separated by context
equivalence in $\LaPRO$, for instance using the context $\ctxone \defi
\app{\ctxhole{\cdot}}{(\ps{I}{\Omega})}$,
since $\ctxone\ctxhole{\termone}\evp{\frac{1}{4}}$ whereas
$\ctxone\ctxhole{\termtwo}\evp{\frac{1}{2}}$.
\end{example}
 
\noindent{}We already know that on full $\LaPRO$, applicative bisimilarity
($\cbnpab$) implies context equivalence ($\cbnconequiv$). Hence, to prove
that \emph{on pure $\lambda$-terms} the two equivalences collapse to LLT
equality ($\LLT$), it suffices to prove that, for those pure terms,
$\cbnconequiv$ implies $\LLT$, and that $\LLT$ implies $\cbnpab$.

The first implication is obtained by a variation on the B\"{o}hm-out
technique, a powerful methodology for separation results in the
$\lambda$-calculus, often employed in proofs about local structure
characterisation theorems of $\lambda$-models.  For this we exploit an
inductive characterisation of LLT equality via stratification
approximants (Definition~\ref{d:approx}).  
The key Lemma \ref{l:ndet} shows that any difference on the trees
of two $\lambda$-terms within level $n$ can be observed by a suitable
context of the probabilistic $\lambda$-calculus.

We write $\uplus M$ as an abbreviation for the term $\ps{\Omega}{M}$.  We
denote by $Q_n $, $n > 0$, the term $\lambda x_1.\ldots\lambda x_n.x_n x_1 x_2 \cdots
x_{n-1}$. This is usually called the {\em B\"{o}hm permutator} of degree
$n$.  B\"{o}hm permutators play a key role 
in the B\"{o}hm-out technique.  A variant of them, the
   $\uplus$-permutators, play a pivotal role 
in Lemma \ref{l:ndet} below. A term $M \in \LaPRO$ is a 
\emph{$\uplus$-permutator of degree $n$ } if either 
$\termone=Q_n$ or there exists $ 0 \leq r < n $ such that 
$$
M = \lambda  x_1.\ldots\lambda x_r. \uplus \lambda  x_{r+1} \cdots \lambda  x_n . x_n x_1\cdots x_{n-1} \, .
$$
Finally, a function $f$ from the  positive integers to
$\lambda$-terms is a 
\emph{$\uplus$-permutator  function} if, for all $n$, $f(n)$ is a  $\uplus$-permutator of degree $n$.
Before giving the main technical lemma, it is useful some auxiliary concepts.  
The definitions below rely on two notions of reduction: $\termone\longrightarrowP{\probone}\termtwo$
means that $\termone$ call-by-name reduces to $\termtwo$ in one step with probability $\probone$. (As a matter
of fact, $\probone$ can be either $1$ or $\frac{1}{2}$.) Then $\Longrightarrow$ is obtained by composing
$\longrightarrow$ zero or more times (and multiplying the corresponding real numbers). If $p=1$ (because, e.g.,
we are dealing with pure $\lambda$-terms) $\Longrightarrow_p$ can be abbreviated just as $\Longrightarrow$.
With a slight abuse of notation, we also denote with $\Longrightarrow$ the multi-step lazy reduction relation
of pure, \emph{open} terms. The specialised form of probabilistic choice $\uplus\termone$ can be thought of as a new syntactic 
construct. Thus $\Lambda^\uplus$ is the set of pure $\lambda$-terms extended with the $\uplus$ operator. 
As $\uplus$ is a derived operator, its operational rules are the expected ones:
$$
\infer[\uplus L]
  {\uplus\termone\longrightarrowP {\frac{1}{2}}\Om}
  {}
\qquad
\infer[\uplus R]
  {\uplus\termone\longrightarrowP {\frac{1}{2}}\termone}
  {}
$$
The restriction on $\Longrightarrow$ in which $\uplus R$, but not $\uplus L$, can be applied, is called
$\Rrightarrow$. In the following, we need the following lemma:
\begin{lemma}\label{l:uplus}
Let $\termone,\termtwo,\termthree,\termfour$ be closed $\Lambda^\uplus$ terms.  Suppose 
$\sum\sem{\termone}=\sum\sem{\termtwo}$, that $\termone\Rrightarrow_p\termthree$ and 
$\termtwo\Rrightarrow_p\termfour$. Then also
$\sum\sem{\termthree}=\sum\sem{\termfour}$.
\end{lemma} 
\begin{proof}
Of course, $p=\frac{1}{2^n}$ for some integer $n\in\NN$. Then
$\sum\sem{\termthree}=2^n\sum\sem{\termone}=2^n\sum\sem{\termtwo}=\sum\sem{\termfour}$.
\end{proof}
The proof of the key Lemma~\ref{l:ndet} below makes essential use of a characterization
of $\LLT$ by a bisimulation-like form of relation:
\begin{definition}[Open Bisimulation]
A relation $\relone$ on pure $\lambda$-terms is an \emph{open bisimulation}
if $\termone \RR \termtwo$ implies:
\begin{varenumerate}
\item 
   if $\termone\Longrightarrow \abstr{\varone}{\termthree}$, then 
   $\termtwo\Longrightarrow\abstr{\varone}{\termfour}$ and $\termthree \RR \termfour$;
\item  
   if  $\termone\Longrightarrow\varone\termthree_1\cdots\termthree_m$, then 
   $\termfour_1,\ldots,\termfour_m$ exist such that $\termtwo\Longrightarrow\varone\termfour_1\cdots\termfour_m$ 
   and $\termthree_i \RR \termfour_i$ for every $1 \leq i \leq m$;
\end{varenumerate}  
and conversely on reductions from $\termtwo$. \emph{Open bisimilarity}, written $\simO$, is the union of 
all open bisimulations. 
\end{definition}
Open bisimulation has the advantage of very easily providing a notion of approximation:
\begin{definition}[Approximants of $\simO$]\label{d:approx}
We set:
\begin{varitemize}
\item 
  ${\simOn 0}  \defi  \Lambda \times \Lambda$;
\item
  $\termone\simOn {n+1}\termtwo$ when       
  \begin{varenumerate}
  \item   
    if $\termone\Longrightarrow\abstr{\varone}{\termthree}$, then $\termfour$ exists such that
    $\termtwo\Longrightarrow\abstr{\varone}{\termfour}$ and $\termthree \simOn n  
    \termfour$;
  \item  
    if  $\termone\Longrightarrow\varone\termthree_1\cdots\termthree_m$, then $\termfour_1,\ldots,\termfour_m$ 
    exist such that $\termtwo\Longrightarrow\varone\termfour_1\cdots \termfour_m$ and   
    $\termthree_i\simOn n \termfour_i$,  for each $1 \leq i \leq m$;  
  \end{varenumerate}
and conversely on the reductions from $\termtwo$.
\end{varitemize}
\end{definition}
Please observe that:
\begin{lemma}
On pure $\lambda$-terms, the relations $\LLT$, $\simO$ and $(\bigcap_{n\in\NN}\simOn n)$ all
coincide.
\end{lemma}
We are now ready to state and prove the key technical lemma:
\begin{lemma}
\label{l:ndet}
Suppose $ M \not \simOn n N$ for some $n$, and let $\{x_1,\ldots , x_r\}$ be
the free variables in $M,N$.  Then there are integers $m_{x_1},\ldots ,
m_{x_r}$ and $k$, and permutator functions $f_{x_1},\ldots , f_{x_r}$ such
that, for all $m>k$, there are closed terms $\til {R_m}$ such that the
following holds: 
if $\substn{M}{x_1}{f_{x_1}(m+m_{x_1})}{x_r}{f_{x_r}(m+m_{x_r})}\til
{R_m}\evp{\probthree}$ and
$\substn{N}{x_1}{f_{x_1}(m+m_{x_1})}{x_r}{f_{x_r}(m+m_{x_r})}\til
{R_m}\evp{\probfour}$ , then $\probthree\neq\probfour$.
\end{lemma}
\begin{proof}
  The proof proceeds by induction on the least $n$ such that $ M \not
  \simOn n N$.  For any term $M$, $M^{\til f}$ will stand for
  $M\sub{f_{x_1}(m+m_{x_1})}{x_1} \ldots \sub{f_{x_r}(m+m_{x_r})}{x_r}$
  where $x_1\ldots x_r$ are the free variables in $M$. We also write
  $\Omega^m$ for a sequence of $m$ occurrences of $\Omega$: so, e.g., $M
  \Om^3 $ is $M \Om \Om \Om$. Finally, for any term $M$, we write $M \Up$
  to denote the fact that $M$ does not converge.
  \begin{varitemize}
  \item
    \textbf{Basic case.} 
    $M\not \simOn {1} N$. There are  a few  cases  to  consider  (their symmetric ones  are  analogous).
    \begin{varitemize}
    \item 
      The case where only one of the two terms diverges is easy.
    \item 
      $M \Longrightarrow x M_1 \cdots M_t$ and $N \Longrightarrow x N_1
      \cdots N_{s}$ with $t<s$. Take $ m_x = s$ and $f_x(n) = Q_n $ (the
      B\"{o}hm permutator of degree $n$).  The values of the other integers
      ($k,m_y$ for $y\neq x$) and of the other permutation functions are
      irrelevant. Set $\til {R_m} \defi \Om^m$. We have
      $$ 
      M^{\til f} \Om^m \Longrightarrow Q_{m+s} M_1^{\til f} \ldots
      M_t^{\til f} \Om^m \Dwa_1
      $$
      since $t+m < s+m$. We also have
      $$ 
      N^{\til f} \Om^m \Longrightarrow Q_{m+s} N_1^{\til f} \ldots
      N_{s}^{\til f} \Om^m \Up
      $$
      since $m > 0$ and therefore an $\Omega$ term will be end up at the
      head of the term.
    \item 
      $M \Longrightarrow x M_1 \cdots M_t$ and $N \Longrightarrow y
      N_1\cdots N_{s}$ with $x\neq y$.  Assume $t\leq s$ without loss of
      generality. Take $ m_x = s+1$, $m_y = s$, and $f_x(n) = f_y(n) = Q_n
      $.  The values of the other integers and permutation functions are
      irrelevant. Set $\til {R_m} \defi \Om^{m}$.  We have
      $$ 
      M^{\til f} \Om^m \Longrightarrow Q_{m+s+1} M_1^{\til f} \ldots
      M_t^{\til f} \Om^{m} \Dwa_1
      $$
      since $m+s+1>t+m$. We also have
      $$ 
      N^{\til f} \Om^m \Longrightarrow Q_{m+s} N_1^{\til f} \ldots
      N_{s}^{\til f} \Om^{m} \Up
      $$
      since $m > 0$ and therefore an $\Omega$ term will be end up at the
      head of the term.
    \item 
      $M\Longrightarrow \lambda x . M'$ and $N \Longrightarrow y \til
      N$, for some $y$ and $\til N$.  The values of the integers and
      permutator functions are irrelevant. Set $\til {R_m} \defi \emptyset
      $ (the empty sequence), and $f_y(n) \defi \uplus Q_n$. We have $ M^{\til
        f} \Longrightarrow \lambda x. {M'}^{\til f} \Dwa_1$, whereas \[
      N^{\til f}\Longrightarrow \uplus Q_{m+m_y} \til {N^{\til f}}
      \DwaP{<1} \]
    \end{varitemize}
  \item  
    \textbf{Inductive case:} $ M {\not \simOn {n+1}} N$. There are two cases  to look at. 
    \begin{varitemize}
    \item 
      $M \Longrightarrow x M_1 \cdots M_s$, $N \Longrightarrow x N_1
      \cdots N_s$ and for some $i$, $ M_i {\not\simOn {n}} N_i$. By
      induction, (for all variables $y$) there are integers $m_y, k$ and
      permutator functions $f_y$, such that for all $m> k$ there are $\til
      {S_m}$ and we have $\sum\sem{M_i^{\til f} \til {S_m}}\not= \sum\sem{N_i^{\til f}
      \til {S_m}}$.  Redefine $k$ if necessary so to make sure that
      $k>s$. Set $\til {R_m} \defi \Om^{m+m_x-s-1} (\lambda x_1 \ldots
      x_{m+m_x}.x_i )\til {S_m}$. We have:
      $$ 
      M^{\til f }\til {R_m} \Longrightarrow f_x(m+m_x) M_1^{\til f }\ldots
      M_s ^{\til f } \Om^{m+m_x-s-1} (\lambda x_1 \ldots x_{m+m_x}.x_i
      )\til {S_m} \LongrightarrowP{p} M_i^{\til f} \til {S_m}
      $$
      whereas 
      $$
      N^{\til f }\til {R_m} \Longrightarrow f_x(m+m_x) N_1^{\til f }\ldots
      N_s ^{\til f } \Om^{m+m_x-s-1} (\lambda x_1 \ldots x_{m+m_x}.x_i
      )\til {S_m} \LongrightarrowP{p} N_i^{\til f} \til {S_m}
      $$
      where $p $ is $ \myfrac 1 2$ or $1$ depending on whether $f_x$
      contains $\uplus$ or not. In any case, in both derivations, rule
      $\uplus L$ has not been used. By Lemma~\ref{l:uplus} and the
      inductive assumption $\sum\sem{M_i^{\til f} \til {S_m}}\not=
      \sum\sem{N_i^{\til f} \til {S_m}}$ we derive that $\sum\sem{M^{\til f
        }\til {R_m}}\not= \sum\sem{N^{\til f }\til {R_m}}$ too.
    \item 
      $M \Longrightarrow \lambda x. M'$, $N \Longrightarrow \lambda x
      .N'$ and $ M' {\not\simOn {n}} N'$. By induction, (for all variables
      $y$) there are integers $m_y, k$ and permutator functions $f_y$, such
      that for all $m> k$ there are $\til {S_m}$ and we have
      $\sum\sem{M'^{\til f} \til {S_m}}\not= \sum\sem{N'^{\til f} \til
        {S_m}}$.  Set $\til {R_m} \defi f_x(m+m_x)\til {S_m}$. Below for a
      term $L$, $L^{\til {f-x}}$ is defined as $L^{\til {f}}$ except that
      variable $x$ is left uninstantiated. We have:
      $$       
      M^{\til f }\til {R_m} \Longrightarrow (\lambda x . {M'}^{\til {f
          -x}}) f_x(m+m_x) \til {S_m} \longrightarrow ({M'}^{\til {f -x}}
      \sub {f_x(m+m_x)}{x}) \til {S_m} = {M'}^{\til {f}} \til {S_m}
      $$
      whereas 
      $$
      N^{\til f }\til {R_m} \Longrightarrow (\lambda x . {N'}^{\til {f
          -x}}) f_x(m+m_x) \til {S_m} \longrightarrow ({N'}^{\til {f -x}}
      \sub {f_x(m+m_x)}{x}) \til {S_m} = {N'}^{\til {f}} \til {S_m}
      $$
      Again, by Lemma~\ref{l:uplus} and the inductive hypothesis, we derive  
      $\sum\sem{M^{\til f }\til {R_m}}\not= \sum\sem{N^{\til f }\til {R_m}}$.
    \end{varitemize} 
  \end{varitemize} 
  This concludes the proof.
\end{proof}

The fact the B\"ohm-out technique actually works implies that the
discriminating power of probabilistic contexts is at least as strong as the
one of LLT's.
\begin{corollary}\label{c:CEvsLLT}
For  $\termone,\termtwo\in\Lambda$, $\termone\cbnconequiv\termtwo$ 
implies $\termone\LLT\termtwo$. 
\end{corollary} 
To show that LLT equality is included in probabilistic applicative bisimilarity,
we proceed as follows. First we define a refinement of the latter, 
essentially one in which we \emph{observe} all probabilistic choices. 
As a consequence, the underlying bisimulation game may ignore probabilities.  
The obtained notion of equivalence is strictly finer than probabilistic 
applicative bisimilarity. The advantage of the refinement is that \emph{both} the 
inclusion of LLT equality in the refinement, and the inclusion of the latter 
in probabilistic applicative bisimilarity turn out to be relatively easy
to prove. A \emph{direct} proof of the inclusion of LLT equality in 
probabilistic applicative bisimilarity would have
been harder, as it would have required extending the notion of a
Levy-Longo tree to $\LaPRO$, then reasoning on substitution closures of such 
trees. 
\begin{definition}
A relation $\relone \subseteq \LOPp{\emptyset}\times\LOPp{\emptyset}$ is a
\emph{strict applicative bisimulation} whenever $\termone \RR \termtwo$ implies
\begin{varenumerate}
\item if $\termone \longrightarrowP 1 \termfour$, then $\termtwo
  \LongrightarrowP 1 \termfive$ and $\termfour \RR \termfive$;
\item if $\termone \longrightarrowP {\frac{1}{2}} \termfour$, then
  $\termtwo \LongrightarrowP {\frac{1}{2}} \termfive$ and $\termfour \RR
  \termfive$;
\item if $\termone = \abstr{\varone}{\termfour}$, then $\termtwo
  \LongrightarrowP 1 \abstr{\varone}{\termfive}$ and $\subst{\termfour}{\varone}{\termthree} \RR
  \subst{\termfive}{\varone}{\termthree}$  for all $\termthree
  \in \LOPp{\emptyset}$;
\item the converse of 1., 2. and 3..
\end{varenumerate}
\emph{Strict applicative bisimilarity} is the union of all strict
applicative bisimulations.
\end{definition}
If two terms have the same LLT, then passing them the same
argument $\termone\in\LOP$ produces \emph{exactly} the same choice
structure: intuitively, whenever the first term finds (a copy of)
$\termone$ in head position, also the second will find $\termone$.  
\begin{lemma}
\label{l:LLrefine}
If $M \LLT N$ then $M \RR N$, for some strict applicative bisimulation
$\R$.
\end{lemma} 
Terms which are strict applicative bisimilar cannot be distinguished by
applicative bisimilarity proper, since the requirements induced by the latter are
less strict than the ones the former imposes:
\begin{lemma}
\label{l:Rrefine}
Strict applicative bisimilarity is included in applicative bisimilarity.
\end{lemma} 
Since we now know that for \emph{pure, deterministic} $\lambda$-terms,
$\LLT$ is included in $\cbnpab$ (by Lemma~\ref{l:LLrefine} and
Lemma~\ref{l:Rrefine}), that $\cbnpab$ is included in $\cbnconequiv$ (by
Theorem~\ref{thm:pab_ce}) and that the latter is
included in $\LLT$ (Corollary \ref{c:CEvsLLT}), we can conclude:
\begin{corollary}
\label{c:LLrefine}
The  relations $\LLT$, $\cbnpab$, and $\cbnconequiv$ coincide in $\Lambda$.
\end{corollary} 


\section{Coupled Logical Bisimulation}
In this section we derive a coinductive characterisation of
probabilistic context equivalence on the whole language $\LOP$
(as opposed to the subset of sum-free $\lambda$-terms as in
Section~\ref{sect:dppc}). For this, we need to manipulate
formal weighted sums. Thus we work with an extension of $\LaPRO$ in
which such weighted sums may appear in redex position. 
An advantage of having formal sums is that the transition
system  on the extended
language can be small-step and deterministic~--- any closed 
term that is not a value will have exactly one possible internal transition.
 
This will make it possible to pursue the \emph{logical bisimulation} method,
in  which the congruence of bisimilarity is proved using a standard
induction argument over all contexts. The refinement of the method 
handling probabilities, called \emph{coupled} logical bisimulation, 
uses  \emph{pairs} of relations, as we need to distinguish between
ordinary terms and terms possibly containing formal sums. 
Technically, in the proof of congruence we first
prove a correspondence between the transition system on
extended terms and the original one for $\LOP$; we then derive a 
few up-to techniques for coupled logical bisimulations that are 
needed in the following proofs; finally, we show that coupled logical 
bisimulations are preserved by the closure of the first relation with 
any context, and the closure of the second relation with any \emph{evaluation} 
context.

We preferred to follow logical bisimulations rather then  environmental
bisimulations because the former admit a simpler definition (in the
latter, each pair of terms is enriched with an environment, that is, an
extra set of pairs of terms). Moreover it is unclear what environments 
should be when one also considers formal sums. We leave this 
for future work.

Formal sums are a tool for representing the behaviour of running $\LaPRO$ terms.  
Thus, on terms with formal sums, only the results for closed terms interest us. 
However, the characterization of contextual equivalence in $\LaPRO$ as coupled 
logical bisimulation also holds on open terms. 
\subsection{Notation and Terminology}
We write $\LaDO$ for the extension of $\LaPRO$ in which
formal sums may appear in redex position. Terms of $\LaDO$ are
defined as follows ($M,N$ being $ \LaPRO$-terms):
$$ 
E,F::= E M  \midd \Prod {i\in I} \pair {M_i}{p_i}  \midd M\myoplus N
\midd \lambda x. M .
$$ 
In a formal sum $\Prod {i\in I} \pair{M_i}{p_i} $, $I$ 
is a countable (possibly empty) set of indices such that $\Sum_{i\in I} p_i \leq 1$. 
We use $+$ for binary formal sums.
Formal sums are ranged over by metavariables like $H,K$. When each $M_i$ is a value 
(i.e., an abstraction) then $\Prod {i\in I} \pair {M_i}{p_i} $ is a 
(\emph{formally summed}) \emph{value}; such values  
are ranged over by  $Z,Y,X$. If 
$H =\Prod {i\in I} \pair {M_i}{p_i}$ and $K = \Prod {j\in J} \pair {M_j}{p_j}$
where $I$ and $J$ are disjoint, then  $H \myoplus K$ abbreviates  
$\Prod {r\in I\cup J} \pair {M_r}{\myfrac{p_r}2}$. Similarly, if for every $j\in J$
$H_j$ is  $\Prod{i\in I}\pair{M_{i,j}}{p_{i,j}}$, then 
$\Prod {j} \pair{H_j}{p_j}$ stands for $\Prod {(i,j)} \pair {M_{i,j}}{p_{i,j}\cdot p_j}$.
For $H =\Prod {i} \pair{ M_i }{p_i}$ we write  $\prob H $ for the real number
$\Sum_i p_i$. If $Z =\Prod {i} \pair{ \lambda x. M_i }{p_i}$, then $Z \bullet\termtwo$ stands for 
$\Prod {i} \pair{  M_i \sub{\termtwo}{x}}{p_i}$.  The set of closed
terms is $\LaDO (\emptyset )$.

Any partial value distribution $\distone$ (in the sense of Section~\ref{sect:p})
can be seen as the formal sum $\Prod {\valone\in\val} \pair{\valone}{\distone(\valone)}$.
Similarly, any formal sum $H=\Prod {i\in I} \pair{M_i}{p_i}$ can be mapped to the distribution
$\sum_{i\in I}p_i\cdot\sem{M_i}$, that we indicate with $\sem{H}$.

Reduction  between $\LaDO$ terms, written $E \DSlongrightarrow F$, is defined by the rules in
Figure~\ref{fig:redlado}; these rules are given on top of the operational semantics
for $\LOP$ as defined in Section~\ref{sect:p}, which is invoked in the premise of rule
$\trans{spc}$ (if there is a $i$ with $M_i$ not a value).
\begin{figure}
$$
\infer[\trans{ss}]
   {\ps{M}{N}\DSlongrightarrow\pair{M}{\myfrac{1}{2}}+\pair{N}{\myfrac{1}{2}}}
   {}
$$
$$
\infer[\trans{sl}]
   {\lambda x. M\DSlongrightarrow\pair{\lambda x. M}{1}} 
   {}
\qquad
\infer[\trans{spc}]
   {\Prod{i}\pair{M_i}{p_i}\DSlongrightarrow \Prod {i}\pair{\distone_i}{p_i}}
   {
     \sem{M_i}=\distone_i}
$$
$$
\infer[\trans{sp}]
   {Z M \DSlongrightarrow Z \bullet M}
   {}
\qquad
\infer[\trans{sa}]
   {EM \DSlongrightarrow FM}
   {E\DSlongrightarrow F}
$$
\caption{Reduction Rules for $\LaDO$}\label{fig:redlado}
\end{figure}
The reduction relation $\DSlongrightarrow$ is deterministic
and strongly normalizing. We use $\DSLongrightarrow $ for its reflexive and transitive closure.
Lemma~\ref{l:new} shows the agreement between the new reduction relation
and the original one.
\begin{lemma}\label{l:new}
  For all $M \in \LaPRO (\emptyset )$ there is  a value $Z$ such that 
  $M \DSLongrightarrow Z$ and $\sem{M} = \sem{Z}$.
\end{lemma} 
\begin{proof}
  One first show that for all $E$  there is $n$ such that
  $E \DSlongrightarrow^n Z$. Then one reasons with a double induction: an induction on
  $n$, and a transition induction, exploiting the determinism of $\DSlongrightarrow$.
\end{proof}
\subsection{Context Equivalence and Bisimulation}
In $\LaDO$ certain terms (i.e., formal sums) may only appear in redex
position; ordinary terms (i.e., terms in $\LaPRO$),
by contrast, may appear in arbitrary
position. When extending context equivalence to  
$\LaDO$ we therefore have to distinguish these two cases. 
Moreover, as our main objective is the characterisation of context
equivalence in $\LaPRO$, we set a somewhat constrained context
equivalence in $\LaDO$ in which
contexts may not contain formal sums
(thus the  $\LaDO$ contexts are the same as 
the $\LaPRO$ contexts).  
We call these \emph{simple $\LaDO$ contexts}, whereas 
we call
\emph{general $\LaDO$ context} an unconstrained context, i.e.,  
a $\LaDO$ term in which  the hole $\contexthole$ may appear in any places
where a term from $\LaPRO$ was expected --- including within a formal sum.   
(Later we will see that allowing general  contexts does not affect the
resulting context equivalence.)
Terms possibly containing formal sums are tested in 
evaluation contexts, i.e., contexts  
 of the form $\contexthole\til M$. 
We write $E \DwaDO{p}$ if $E \DSLongrightarrow Z$ and $\prob Z = p $
(recall that $Z$ is unique, for a given $E$).
\begin{definition}[Context Equivalence in $\LaDO$]
\label{d:ciu}
Two $\LaPRO$-terms 
$M$ and $N $ are \emph{context equivalent in $\LaDO$}, written 
$M \conteqDO N$, if for all (closing) simple $\LaDO$ contexts  $\qct$, we have 
$\ct M \DwaDO{p}$ iff $\ct N\DwaDO{p}$. Two $\LaDO$-terms 
$E$ and $F $ are \emph{evaluation-context equivalent}, written 
$E \conteqDOeval F$, if for all (closing)  $\LaDO$ evaluation
contexts  $\qct$, we have $\ct E \DwaDO{p}$ iff $\ct F \DwaDO{p}$.
\end{definition}
In virtue of Lemma~\ref{l:new},  context
equivalence in $\LaPRO$ coincides with context
equivalence in $\LaDO$. 

We now introduce a bisimulation that yields a coinductive
characterisation of  context equivalence (and also of
evaluation-context equivalence).
A \emph{\coupledrel} is a pair $( \V, \E)$ where:
$\V \subseteq \LaPRO(\emptyset ) \times \LaPRO (\emptyset )$,
$\E \subseteq \LaDO(\emptyset ) \times \LaDO(\emptyset )$,
and  $\V \subseteq \E$.
Intuitively, we place in $\V$  the pairs of terms that should be preserved
by all contexts, and in $\E$ those that should be preserved by evaluation 
contexts. For a \coupledrel\  $\R = ( \V ,\E)$ we write $\R_1 $ for
$\V$ and $\R_2$ for $\E$. The union of \coupledrel s is defined 
componentwise: e.g., if $\R$ and $\SS$ are \coupledrel s, then 
the \coupledrel\ $\R\;\cup\SS$ has $(\R\;\cup\SS)_1  \defi \R_1\;\cup\SS_1$ and 
$(\R\;\cup\SS)_2  \defi  \R_2\;\cup\SS_2$.
If $\V $ is a relation on $\LaPRO$, then $\starred\V $ is the
context closure of $ \V$ in $\LaPRO$, i.e., the set of all (closed) terms of the form 
$(\ct {\til M} , \ct {\til N}$) where $\qct$ is a multi-hole
$\LaPRO$ context and $\til M \VV \til N$.
 
\begin{definition}
\label{d:clb}
A \coupledrel\ $\R$ is a \emph{coupled logical bisimulation} if whenever $E \RR_2 F$ we
have:
\begin{varenumerate}
\item
  if $E \DSlongrightarrow D$, then $F \DSLongrightarrow G$, where 
  $D\RR_2 G$;
\item 
  if $E$ is a formally summed value, then $F\DSLongrightarrow Y$  
  with $\prob E = \prob Y$, and for all $M \starred{\R_1} N$ we have 
  $(E \bullet M)\RR_2(Y\bullet N)$;
\item 
  the converse of 1. and 2..
\end{varenumerate} 
\emph{Coupled logical bisimilarity},  $\logbis$, is the union of all coupled logical bisimulations
(hence $\logbis_1$ is the union of the first component of all coupled
logical bisimulations, and similarly for $\logbis_2$).
\end{definition}
In a coupled bisimulation $(\R_1,\R_2)$, 
the bisimulation game is only played on the pairs in $\R_2$. 
However, the first relation
$\R_1$ is relevant, as  inputs for tested functions are built using
$\R_1$ (Clause 2.\ of Definition \ref{d:clb}). Actually, also the pairs in $\R_1$ are tested, because in any
coupled relations it must be $\R_1 \subseteq \R_2$.
The values produced by the bisimulation game for coupled bisimulation
on $\R_2$ are formal sums (not plain $\lambda$-terms),  and this is why we do not
require them to be in $\R_1$: formal sums should only  appear in redex
position, but terms in $\R_1$ can be used as arguments to bisimilar functions and can
therefore end up in arbitrary positions.

We will see below another aspect of the relevance of $\R_1$: 
the proof technique of logical bisimulation only allows us to prove  substitutivity of the
bisimilarity in arbitrary contexts \emph{for the pairs of terms} in $\R_1$. For pairs 
in $\R_2$ but not in $\R_1 $ the proof technique only allows us to
derive preservation in evaluation contexts.

In the proof of congruence of coupled logical  bisimilarity we will
push ``as many terms as possible'' into the first relation, i.e.,
the first relation will be as large as possible. However, in   proofs of bisimilarity
for concrete terms, the  first relation  may be very small, possibly a singleton or even empty. 
Then the bisimulation clauses become similar to those of applicative
bisimulation (as inputs of tested function are ``almost'' identical).
Summing up, in coupled logical  bisimulation the use of two relations
gives us more flexibility than in  ordinary logical bisimulation: depending on
the needs, we can tune the size of the first relation. 
It is possible that some of the above aspects of coupled logical
bisimilarity be specific to call-by-name, and that the call-by-value
version would require non-trivial modifications.

\begin{remark}
\label{r:envLog} 
In a coupled logical bisimulation, the first relation is used to 
construct the inputs for the tested functions (the formally summed
values produced in the bisimulation game for the second relation). 
Therefore,  such first relation may be thought of as a  
``global'' environment--- global because it is the same for each pair
of terms on which the bisimulation game is played.
As a consequence, coupled logical bisimulation remains quite different
from environmental bisimulation~\cite{SangiorgiKS11}, where 
 the ``environment'' for  constructing inputs  is local to
each pair of tested terms. Coupled logical bisimulation follows ordinary logical
bisimulation~\cite{SangiorgiKS07fsen}, in which there is only one   global environment; in  
ordinary logical bisimulation, however, the global environment coincides with 
the set of tested terms. The similarity with logical bisimulation is also revealed by
non-monotonicity of the associated functional (in contrast, the
functional associated to environmental bisimulation is monotone); see
Remark~\ref{r:nonmono}.
\end{remark} 
As an example of use of coupled logical bisimulation, we revisit the
counterexample~\ref{ex:count} to the completeness of applicative
bisimilarity with respect to contextual equivalence.

\begin{example}
\label{ex:countRevisited}
We consider the terms of Example~\ref{ex:count} and show that they are
in  $\logbis_1$, hence also in $\cbnconequiv$ (contextual equivalence
of $\LaPRO$), by
Corollary~\ref{c:BisContDO} and   $\conteqDO = \cbnconequiv$.
Recall that the terms are
$M \defi  \lambda x. (L\oplus P)$ and 
$N \defi (\lambda x. L )\oplus(\lambda x. P)$
for $L \defi \lambda z.\Omega$ and $P \defi \lambda y.\lambda z.\Omega$. 
We set $\R_1 $ to contain only $(M,N)$ (this is the pair
that interests us), and $R_2 $ to contain the pairs $(M,N), 
(\pair M 1,\pair{\lambda x . L}{\myfrac   12}+  \pair{\lambda x . P}{\myfrac 12}), 
(\pair{L+P} 1, \pair{ L}{\myfrac   12}+  \pair{ P}{\myfrac
12})$, and a set of  pairs with  identical components, namely 
$ 
(
\pair{ L}{\myfrac   12}+  \pair{ P}{\myfrac   12},
\pair{ L}{\myfrac   12}+  \pair{ P}{\myfrac   12})$,
$ 
(
\pair{\Omega}{\myfrac   12}+  \pair{ \lambda u . \Omega}{\myfrac   12},
\pair{\Omega}{\myfrac   12}+  \pair{ \lambda u . \Omega}{\myfrac   12})$,
$ 
(
 \pair{ \lambda u . \Omega}{\myfrac   12},$ $
   \pair{ \lambda u . \Omega}{\myfrac   12})$,
$ 
(
 \pair{  \Omega}{\myfrac   12},
 \pair{  \Omega}{\myfrac   12})$,
$ 
(\emptyset, \emptyset) $, 
where $\emptyset $ is the empty formal sum. 
Thus $(\R_1,\R_2)$ is a coupled logical bisimulation. 
\end{example} 
The main challenge towards the goal of relating coupled logical bisimilarity 
and context equivalence is the substitutivity of bisimulation. We establish 
the latter exploiting some \emph{up-to techniques} for bisimulation. We only 
give the definitions of the techniques, omitting the statements about their 
soundness. The first up-to technique allows us to drop the bisimulation game on
silent actions:
\begin{definition}[Big-Step Bisimulation]
A \coupledrel\  $\R$ is a \emph{big-step 
coupled logical bisimulation} if whenever $E \RR_2 F$, the following holds:
if $E \DSLongrightarrow Z$ then $F \DSLongrightarrow Y$ with $\prob Z =
\prob Y$, and for all $M \starred{\R_1} N$ we have $(Z \bullet M)\RR_2(Y
\bullet N)$.
\end{definition}
\begin{lemma}
\label{l:big}
If $\R$ is a big-step coupled logical bisimulation, then $\R \subseteq \SS$
for some coupled logical bisimulation $\SS$.
\end{lemma} 
In the reduction $\DSlongrightarrow$, computation is performed at the level of formal sums;
and this is reflected, in coupled bisimulation, by the application of values to formal sums only. 
The following up-to technique allows  computation, and application of  input values, also with ordinary  terms.
In the definition, we extract a formal sum from a term $E$ in $\LaDO$ using the function $\distQ(\cdot)$ inductively as
follows: 
\begin{align*}
\dist{EM}&\defi  \Prod {i} \pair{M_i M}{p_i}  \mbox{ whenever $\dist E=\Prod {i} \pair{M_i }{p_i}$;}\\
\dist{M}&\defi \pair{M}{1}; \hskip 1.6cm  
\dist{H} \;\,\defi\;\,\overl H.
\end{align*} 
\begin{definition}
\label{d:uptoFS}
A \coupledrel\ $\R$ is a \emph{bisimulation up-to formal sums}  if,
whenever $E \RR_2 F$, then either (one of the  bisimulation clauses of
Definition~\ref{d:clb} applies), or ($E,F \in \LaPRO$ and  one of the
following clauses applies):
\begin{varenumerate}
\item 
   $E\DSlongrightarrow D$ with 
   $\dist{D}=\pair{M}{\myfrac 1 2} + \pair{N}{\myfrac 1 2}$,  
   and $F\DSlongrightarrow G$
   with $\dist{G} = \pair{\termthree}{\myfrac 1 2} + \pair{\termfour}{\myfrac 1 2}$,  
   $\termone\RR_2\termthree$, and $\termtwo\RR_2\termfour$;
\item 
  $E = \lambda x. M$ and 
  $F = \lambda x. N$, and 
  for all $P \starred{\R_1} Q$ we have 
  $M\sub Px \RR_2  N\sub Qx$;
\item 
$E = (\lambda x. M) P \til M $ and 
$F = (\lambda x. N) Q \til N$, and  
$
M\sub {P}x \til M \RR_2  
N\sub {Q}x \til N 
$.
\end{varenumerate}
 \end{definition}
According to Definition~\ref{d:uptoFS}, in the bisimulation game for a
coupled relation, given a pair $(E,F) \in \R_2$, we can either
choose to follow the  bisimulation game in the original
Definition~\ref{d:clb}; or, if $  E$ and $F$ do not contain formal
sums, we can try one of the new clauses above. The advantage of the
first new clause is that it allows us to make a split on the derivatives
of the original terms. The advantage of the other two new clauses is that
they allow us to directly handle the given $\lambda$-terms,
without using the operational rules  of Figure~\ref{fig:redlado} and
therefore without introducing formal sums.  To understand the first
clause, suppose $E \defi (\termone\oplus\termtwo)\termthree$ and 
$F\defi\termfour\oplus\termfive$. We have $E\DSlongrightarrow (\pair{\termone}{\myfrac 1 2} +
\pair{\termtwo}{\myfrac 1 2})\termthree \defi G$ 
with $\dist{G} = \pair{\termone\termthree}{\myfrac 1 2} + \pair{\termtwo\termthree}{\myfrac 1 2}$,
and $F\DSlongrightarrow \pair{\termfour}{\myfrac 1 2}+\pair{\termfive}{\myfrac 1 2} \defi H$,
with $\dist{H}=H$, and it is sufficient now to ensure $(\termone\termthree)\RR_2\termfour$, and 
$(\termtwo\termthree)\RR_2\termfive$.

\begin{lemma}
\label{l:uptoMAIN}
If $\R$ is a bisimulation up-to formal sums, then $\R \subseteq \SS$ for
some coupled logical bisimulation $\SS$.
\end{lemma} 
\begin{proof}
We show that the \coupledrel\ $\SS$, with $\SS_1 = \R_1$ and  
$$\SS_2 \defi \R_2 \cup \{ (\Prod i \pair{H_i}{p_i}, \Prod i \pair{K_i}{p_i}) \st
\mbox{for each $i$, 
$\begin{array}[t]{l}
\mbox{ either $H_i \RR_2 K_i$} \\
\mbox{ or $H_i = \pair{M_i}1,K_i = \pair{N_i}1$ and  $M_i \RR_2 N_i$
   } \}, 
\end{array}$}
 $$
 is a big-step bisimulation and then apply
Lemma~\ref{l:big}. 
The key point for this is to show that whenever $M \RR_2 N$, if 
 $M \DSLongrightarrow Z$  and $N \DSLongrightarrow Y$, then $Z \SS_2 Y$.

For this, roughly,  
we reason on the tree whose nodes are the pairs of terms 
 produced by the
up-to bisimulation game for $\R_2$ and with root a pair $(M ,N)$ in
$\R_2$ (and with the proviso that 
a node $(E,F)$, if not a pair of values, and   not a pair of
$\LaPRO$-terms, has 
one only child, namely
 $(Z,Y)$ 
for $Z,Y$ s.t.\ 
$ E \DSLongrightarrow Z $ and $F \DSLongrightarrow Y$).

Certain paths in the tree may be
divergent; those that reach a leaf give the formal sums that
$M$ and $N$  produce.
Thus, 
if $M \ULdwa \overl Z$ and $N \ULdwa \overl Y$, then
we can write $Z = \Prod i \pair{Z_i}{p_i}$ and 
 $Y = \Prod i \pair{Y_i}{p_i}$, for $Z_i,Y_i,p_i$ s.t.\ 
$\{(Z_i,Y_i,p_i)\} $ represent exactly the multiset of the leaves in
the tree together with the probability of the path  reaching  each
 leaf. 
\end{proof} 
Using the above proof technique, we can prove the necessary
substitutivity property for bisimulation.  The use of up-to techniques, and
the way bisimulation is defined (in particular the presence of a clause for
$\tau$-steps and the possibility of using the pairs in the bisimulation
itself to construct inputs for functions), make it possible to use a
standard argument by induction over contexts.

\begin{lemma}
\label{l:CTclo}
If $\R$ is a bisimulation then the context closure $
\SS$ with
\[ \setlength\arraycolsep{3pt}
\begin{array}{rcl}
  \SS_1 & \defi & \starred{\R_1}; \\
  \SS_2 & \defi & \R_2\,\cup\,\starred{\R_1}\,\cup\,\{ (E \til M, F \til N) \st E \RR_2 F \mbox{ and } M_i
  \starred{\R_1} N_i \};
\end{array} \]
 is a  bisimulation up-to formal sums.
\end{lemma}  
\begin{corollary}
\label{c:BisCong}
\begin{enumerate}
\item
$M \approx_1 N  $ implies $\ct M \approx_1 \ct N$, for all $\qct$

\item 
$E \approx_2 F$ implies 
$\ct E \approx_2 \ct F$, for all evaluation contexts $\qct$.
\end{enumerate}
 \end{corollary}  
Using Lemma~\ref{l:CTclo} we can prove the inclusion in context equivalence.
\begin{corollary}\label{c:BisContDO}
  If $M \approx_1 N  $ then $M \conteqDO N$.
  Moreover, if $E \approx_2 F  $ then $E \conteqDOeval F$.
\end{corollary}  
The  converse of Corollary~\ref{c:BisContDO} is proved exploiting 
a few simple properties of $\conteqDOeval$ (e.g.,  its transitivity,  the
inclusion  ${\DSlongrightarrow} \subseteq {\conteqDOeval}$).

\begin{lemma}
\label{l:redEQUIV}
$E \DSlongrightarrow E'$ implies $ E \conteqDOeval E'$.
\end{lemma} 

\begin{proof}
If $E \DSlongrightarrow E'$ then $E \approx_2 E'$ hence
 $E \conteqDOeval E'$.
\end{proof} 

\begin{lemma}
\label{l:appEQUIV}
$Z \conteqDOeval Y$ implies
$Z \bullet M \conteqDOeval Y \bullet M$ for all $M$.
\end{lemma} 

\begin{proof}
Follows from definition of $\conteqDOeval$, transitivity of $\conteqDOeval$,
and Lemma~\ref{l:redEQUIV}.  
\end{proof} 

\begin{lemma}
\label{l:MiNiEQUIV}
If $M_i \conteqDO N_i$ for each $i$, then 
$\Prod i \pair{M_i}{p_i} \conteqDOeval
\Prod i \pair{N_i}{p_i}$
\end{lemma}  
\begin{proof}
  Suppose $\Prod i \pair{M_i}{p_i} \til M \DSLongrightarrow Z$ and 
  $\Prod i \pair{N_i}{p_i}\til M \DSLongrightarrow Y$.
  We have to show $\prob Z = \prob Y$.
  We have $Z = \Prod i \pair{Z_i}{p_i}$ for 
  $Z_i$ with  $M_i \DSLongrightarrow Z_i$. 
  Similarly
  $Y = \Prod i \pair{Y_i}{p_i}$ for 
  $Y_i$ with  $N_i \DSLongrightarrow N_i$.
  Then the result follows from $\prob{Z_i}=\prob{Y_i}$.
\end{proof} 

\begin{theorem}\label{t:ciuBis}
  We have ${\conteqDO} \subseteq {\approx_1}$, and   
  ${\conteqDOeval} \subseteq {\approx_2}$. 
\end{theorem}
\begin{proof}
  We take the  \coupledrel\ $\R$ with 
  \[
  \begin{array}{rcl}
  \R_1 & \defi & \{ (M,N) \st M \conteqDO N\}\\
  \R_2 & \defi & \{ (E,F) \st E \conteqDOeval F\}
   \} 
  \end{array}  
  \]  
  and show that $\R$ is a bisimulation.

For clause (1), one uses Lemma~\ref{l:redEQUIV} and transitivity of
$\conteqDOeval$.
For clause (2), 
consider a term
$Z$ with $Z\conteqDOeval F$.
By definition of $\conteqDOeval$,
 $F \DSLongrightarrow Y$ with $\prob Z =
\prob Y$. Take now arguments $M \conteqDO N$ (which is sufficient, since
$\starred\conteqDO \subseteq \conteqDO$). 
By Lemma~\ref{l:appEQUIV},
$Z \bullet M \conteqDOeval Y\bullet N$.
By Lemma~\ref{l:MiNiEQUIV}, 
$W \bullet M \conteqDOeval Y\bullet N$. 
Hence also 
$Z \bullet M \conteqDOeval Y\bullet N$, and we have
 $Z \bullet M \RR_2 Y\bullet N$.
\end{proof}
\noindent{}It also holds that coupled logical bisimilarity is preserved by
the formal sum construct; i.e., $M_i \approx_1 N_i$ for each $i\in I $
implies $ \Prod{ i\in I} \pair{M_i}{p_i} \approx_2 \Prod{ i\in I}
\pair{N_i}{p_i}$. As a consequence, context equivalence defined on general
$\LaDO$ contexts is the same as that set on simple contexts
(Definition~\ref{d:ciu}).

\begin{remark}
\label{r:nonmono} 
The functional induced by coupled logical bisimulation is \emph{not} monotone.
For instance, if $\V \subseteq \W$, then  a pair of terms may satisfy
the bisimulation clauses on $(\V,\E)$, for some $\E$, but not on
$(\W,\E)$, because the input for functions may be taken from the
larger relation $\W$.
(Recall that coupled relations are \emph{pairs} of relations. Hence operations on coupled
relations, such as union and inclusion, are defined component-wise.)
 However, Corollary~\ref{c:BisContDO} and 
Theorem~\ref{t:ciuBis} tell us that there is indeed a largest bisimulation, 
namely the pair $(\conteqDO,\conteqDOeval)$. 
\end{remark}
With logical (as well as environmental) bisimulations, up-to
techniques are particularly important to relieve the burden of proving
concrete equalities. A powerful up-to technique in higher-order
languages is \emph{up-to contexts}. We present  a form of up-to
contexts combined with the big-step version of logical bisimilarity. 
Below, for a relation $\R$ on $\LaPRO$, we write $\starredD\R$ for
the closure of the relation under general (closing) $\LaDO$ contexts.

\begin{definition}
\label{d:bigUPTOcon}
A \coupledrel\  $\R$ is a \emph{big-step 
coupled logical bisimulation up-to contexts} if whenever $E \RR_2 F$, the
following holds:
if $E \DSLongrightarrow Z$ then $F \DSLongrightarrow Y$ with $\prob Z =
\prob Y$, and for all $M \starred{\R_1} N$, we have $(Z \bullet
M)\starredD{\R_1}(Y \bullet N)$.
\end{definition}
For the soundness proof, we first derive the soundness of a small-step
up-to context technique, whose proof, in turn, is similar to that 
of Lemma~\ref{l:CTclo} (the up-to-formal-sums technique of
Definition~\ref{d:uptoFS} already
allows  some context manipulation; we need this technique for the proof 
of the up-to-contexts technique).

\begin{example}
We have seen that the terms \hsk{expone} and \hsk{exptwo}
of Example~\ref{ex:exp} are not applicative
bisimilar. We can show that they  are context equivalent, by proving
that they are coupled bisimilar.  We sketch a proof of this, in which
we employ the up-to technique from Definition~\ref{d:bigUPTOcon}.  We
use the coupled relation $\R$ in which $\R_1 \defi \{ (\hsk{expone},
\hsk{exptwo})\}$, and $\R_2 \defi \R_1 \cup \{ (A_M,B_N) \; | \; 
M
\starred{\R_1} N \}$ where 
$A_M \defi 
  \lambda n . ( (M n) \oplus (\hsk{expone}\, M \, (n+1)))$, and 
$B_N \defi 
(\lambda x . N x) \oplus (\hsk{exptwo}\, (\lambda x .N (x+1)))
$. This is a big-step coupled logical bisimulation up-to contexts. 
The interesting part is the matching argument for the terms $A_M,
B_N$; upon receiving an argument $m$ they yield the summed values 
$\Prod {i} \pair{M (m+1)}{p_i} $ and 
$\Prod {i} \pair{N (m+1)}{p_i} $ (for some $p_i$'s), and these  are in
$\starredD{\R_1}$.
\end{example}

\section{Beyond Call-by-Name Reduction}\label{sect:beyond}
So far, we have studied the problem of giving sound (and sometime complete)
coinductive methods for program equivalence in a probabilistic $\lambda$-calculus endowed
with \emph{call-by-name} reduction. One may  wonder whether what
we have obtained can be adapted to other notions of reduction, and in particular
to \emph{call-by-value} reduction (e.g., the  call-by-value
operational semantics of $\LOP$ from~\cite{DalLagoZorzi}).

Since our construction of a labelled Markov chain for $\LOP$ is somehow independent on the underlying
operational semantics, {defining} a call-by-value probabilistic applicative
bisimulation is effortless. The  proofs of congruence of the bisimilarity
and its soundness      in this paper can also be transplanted to
call-by-value. In defining $\LOP$ as a multisorted labelled Markov
chain for the strict regime, one should recall that functions are
applied to values only.
\begin{definition}
\label{d:multisortCBV}
$\LOP$ can be seen as a multisorted labelled Markov chain
$(\LOPp{\emptyset}\uplus\val,\val\uplus\{\evlabel\},\translop)$ that we
denote with $\cbv{\LOP}$. Please observe that, contrary to how we gave
Definition~\ref{d:multisort} for call-by-name semantics, labels here are
either values, which model parameter passing, or $\evlabel$, that models
evaluation. We define the transition probability matrix $\translop$ as
follows:
\begin{varitemize}
\item 
  For every term $\termone$ and for every 
  distinguished value $\clabstr{\varone}{\termtwo}$,
  $$
  \translop(\termone, \evlabel,\clabstr{\varone}{\termtwo})\defi\sem{\termone}(\clabstr{\varone}{\termtwo});
  $$
\item 
  For every value $\valone$ and for every 
  distinguished value $\clabstr{\varone}{\termtwo}$,
  $$
  \translop(\clabstr{\varone}{\termtwo}, \valone,
  \subst{\termtwo}{\varone}{\valone})\defi 1;
  $$
\item
  In all other cases, $\translop$ returns $0$.
\end{varitemize}
\end{definition}
Then, similarly to the call-by-name case, one can define both probabilistic
applicative simulation and bisimulation notions as probabilistic simulation
and bisimulation on $\cbv{\LOP}$. This way one can define
\emph{probabilistic applicative bisimilarity}, which is denoted $\cbvpab$,
and \emph{probabilistic applicative similarity}, denoted $\cbvpas$.

Proving that $\cbvpas$ is a precongruence, follows the reasoning we have
outlined for the lazy regime. Of course, one must prove a Key Lemma first.
\begin{lemma}\label{lemma:keylemmaCBN}
  If $\relu{\termone}{\howe{\cbvpas}}{\termtwo}$, then for every
  $\setone\subseteq\LOP(\varone)$ it holds that
  $\sem{\termone}(\abstr{\varone}{\setone})\leq\sem{\termtwo}(\abstr{\varone}{(\howe{\cbvpas}(\setone))})$.
\end{lemma}
As the statement, the proof is not particularly different from the one we
have provided for Lemma~\ref{lemma:keylemma}. The only delicate case is
obviously that of application. This is due to its operational semantics
that, now, takes into account also the distribution of values the
parameter reduces to. Anyway, one can prove $\cbvpab$ of implying context
equivalence. 

When we restrict our attention to pure $\lambda$-terms, as we do in
Section~\ref{sect:dppc}, we are strongly relying on call-by-name
evaluation: LLT's only reflect term equivalence in a call-by-name lazy
regime. We leave the task of generalizing the results to eager evaluation
to future work, but we conjecture that, in that setting, probabilistic
choice \emph{alone} does not give contexts the same discriminating power as
probabilistic bisimulation.  Similarly we have not investigated the
call-by-value version of coupled logical bisimilarity, as our current
proofs rely on the appearance of formal sums only in redex position, a
constraint that would probably have to be lifted for call-by-value.

\section{A Comparison with Nondeterminism}\label{sect:comparison}
Syntactically, $\LOP$ is identical to an eponymous language introduced by
de'Liguoro and Piperno~\cite{deLiguoroPiperno95}.  The semantics we present
here, however, is quantitative, and this has of course a great impact on
context equivalence.  While in a nondeterministic setting what one observes
is the \emph{possibility} of converging (or of diverging, or both), terms
with different convergence probabilities are considered different in an
essential way here. Actually, nondeterministic context equivalence and
probabilistic context equivalence are incomparable.  As an example of terms
that are context equivalent in the \emph{must} sense but not
probabilistically, we can take $\ps{I}{(\ps{I}{\Omega})}$ and
$\ps{I}{\Omega}$. Conversely, $I$ is probabilistically equivalent to any
term $\termone$ that reduces to $\ps{I}{\termone}$ (which can be defined
using fixed-point combinators), while $I$ and $\termone$ are not equivalent
in the \emph{must} sense, since the latter can diverge (the divergence is
irrelevant probabilistically because it has probability zero). \emph{May}
context equivalence, in contrast, is coarser than probabilistic
context equivalence.

Despite the differences, the two semantics have similarities.  Analogously
to what happens in nondeterministic $\lambda$-calculi, applicative
bisimulation and context equivalence do \emph{not} coincide in the
probabilistic setting, at least if call-by-name is considered. 
The counterexamples to full abstraction are much more complicated
in call-by-value $\lambda$-calculi ~\cite{Las98a}, and cannot
be easily adapted to the probabilistic setting.

\section{Conclusions}
This is the first paper in which bisimulation  techniques for program
equivalence 
are shown to be applicable to 
probabilistic $\lambda$-calculi.

On the one hand, Abramsky's idea of seeing interaction as application is
shown to be amenable to a probabilistic treatment, giving rise to a
congruence relation that is sound for context equivalence. Completeness,
however, fails: the way probabilistic applicative bisimulation is defined
allows one to distinguish terms that are context equivalent, but which behave
differently as for \emph{when} choices and interactions are performed.  On
the other, a notion of coupled logical bisimulation is introduced and
proved to precisely characterise context equivalence for $\LOP$. Along the
way,  applicative bisimilarity is proved to coincide with context
equivalence on pure $\lambda$-terms, yielding the Levy-Longo tree equality.

The crucial difference between the two main bisimulations studied
in the paper is not the style (applicative \emph{vis-\`a-vis} logical), but rather the fact
that while applicative bisimulation insists on relating only individual
terms, coupled logical bisimulation is more flexible and allows us to relate
formal sums (which we may think as distributions). This also explains
why we need distinct reduction rules for the two bisimulations.
See examples~\ref{ex:count} and \ref{ex:countRevisited}. 
While not complete, applicative bisimulation, as it stands,  is
simpler to use than coupled logical bisimulation. Moreover it is a natural
form of bisimulation, and it should be interesting trying to transport 
the techniques for handling it onto variants or extensions of the language.

Topics for future work abound --- some have already been hinted at in
earlier sections. Among the most interesting ones, one can mention  
 the transport of applicative bisimulation  onto the language $\LaDO$.
We conjecture that the resulting relation would coincide with coupled
logical bisimilarity and 
context equivalence, but going through  Howe's technique seems more difficult than for $\LOP$, given
the infinitary nature of formal sums and their confinement to redex positions.

Also interesting would be a more \emph{effective} notion of equivalence: even if the two introduced notions of bisimulation
avoid universal quantifications over all possible contexts, they refer to an essentially infinitary
operational semantics in which the meaning of a term is obtained as the least upper bound of all
its finite approximations. Would it be possible to define bisimulation in terms of approximations
without getting too fine grained? 

Bisimulations in the style of logical bisimulation (or environmental
bisimulation) are known to require up-to techniques in order to avoid
tedious equality proofs on concrete terms. In the paper we have
introduced some up-to techniques for coupled logical bisimilarity, but
additional techniques  would be useful. 
Up-to techniques could also be
developed for  applicative bisimilarity.

More in the long-run, we would like to develop sound operational techniques for so-called
\emph{computational indistinguishability}, a key notion in modern cryptography.
Computational indistinguishability is defined similarly to context
equivalence; the context is however required to work within appropriate resource bounds, while
the two terms can have different observable behaviors (although with negligible probability). We see this work
as a very first step in this direction: complexity bounds are not yet there, but probabilistic
behaviour, an essential ingredient, is correctly taken into account.



\bibliographystyle{plain}
\bibliography{biblio}

\begin{thebibliography}{10}

\bibitem{Abramsky-90}
S.~Abramsky.
\newblock {The Lazy {$\lambda$}-Calculus}.
\newblock In D.~Turner, editor, {\em Research Topics in Functional
  Programming}, pages 65--117. Addison Wesley, 1990.

\bibitem{AbramskyOng93}
Samson Abramsky and C.-H.~Luke Ong.
\newblock Full abstraction in the lazy lambda calculus.
\newblock {\em Inf. Comput.}, 105(2):159--267, 1993.

\bibitem{AstesianoCosta84}
Egidio Astesiano and Gerardo Costa.
\newblock Distributive semantics for nondeterministic typed lambda-calculi.
\newblock {\em Theor. Comput. Sci.}, 32:121--156, 1984.

\bibitem{Barendregt84}
Hendrik~Pieter Barendregt.
\newblock {\em {The Lambda Calculus -- Its Syntax and Semantics}}, volume 103
  of {\em Studies in Logic and the Foundations of Mathematics}.
\newblock North-Holland, 1984.

\bibitem{BernardoNL13}
Marco Bernardo, Rocco {De Nicola}, and Michele Loreti.
\newblock A uniform framework for modeling nondeterministic, probabilistic,
  stochastic, or mixed processes and their behavioral equivalences.
\newblock {\em Inf. Comput.}, 225:29--82, 2013.

\bibitem{BoLa94}
G.~Boudol and C.~Laneve.
\newblock The discriminating power of the $\lambda$-calculus with
  multiplicities.
\newblock {\em Inf. Comput.}, 126(1):83--102, 1996.

\bibitem{Boudol94}
G{\'e}rard Boudol.
\newblock Lambda-calculi for (strict) parallel functions.
\newblock {\em Inf. Comput.}, 108(1):51--127, 1994.

\bibitem{comaniciu2003kernel}
Dorin Comaniciu, Visvanathan Ramesh, and Peter Meer.
\newblock Kernel-based object tracking.
\newblock {\em IEEE Trans. on Pattern Analysis and Machine Intelligence,},
  25(5):564--577, 2003.

\bibitem{EV}
Ugo Dal~Lago, Davide Sangiorgi, and Michele Alberti.
\newblock On coinductive equivalences for probabilistic higher-order functional
  programs (long version).
\newblock Available at \url{http://arxiv...}, 2013.

\bibitem{DalLagoZorzi}
Ugo Dal~Lago and Margherita Zorzi.
\newblock Probabilistic operational semantics for the lambda calculus.
\newblock {\em RAIRO - Theor. Inf. and Applic.}, 46(3):413--450, 2012.

\bibitem{DanosHarmer}
Vincent Danos and Russell Harmer.
\newblock Probabilistic game semantics.
\newblock {\em ACM Trans. Comput. Log.}, 3(3):359--382, 2002.

\bibitem{NicolaHennessy84}
Rocco {De Nicola} and Matthew Hennessy.
\newblock Testing equivalences for processes.
\newblock {\em Theor. Comput. Sci.}, 34:83--133, 1984.

\bibitem{deLiguoroPiperno95}
Ugo de'Liguoro and Adolfo Piperno.
\newblock Non deterministic extensions of untyped lambda-calculus.
\newblock {\em Inf. Comput.}, 122(2):149--177, 1995.

\bibitem{DezG01}
M.\ Dezani-Ciancaglini and E.\ Giovannetti.
\newblock From bohm's theorem to observational equivalences: an informal
  account.
\newblock {\em Electr. Notes Theor. Comput. Sci.}, 50(2):83--116, 2001.

\bibitem{DezTU99}
M.\ Dezani-Ciancaglini, J.\ Tiuryn, and P.\ Urzyczyn.
\newblock Discrimination by parallel observers: The algorithm.
\newblock {\em Inf. Comput.}, 150(2):153--186, 1999.

\bibitem{EhrhardPaganiTasson}
Thomas Ehrhard, Michele Pagani, and Christine Tasson.
\newblock The computational meaning of probabilistic coherence spaces.
\newblock In {\em LICS}, pages 87--96, 2011.

\bibitem{GoldwasserMicali}
Shafi Goldwasser and Silvio Micali.
\newblock Probabilistic encryption.
\newblock {\em J. Comput. Syst. Sci.}, 28(2):270--299, 1984.

\bibitem{Goodman}
Noah~D. Goodman.
\newblock The principles and practice of probabilistic programming.
\newblock In {\em POPL}, pages 399--402, 2013.

\bibitem{Gordon-92}
Andrew~D. Gordon.
\newblock Bisimilarity as a theory of functional programming.
\newblock {\em Electr. Notes Theor. Comput. Sci.}, 1:232--252, 1995.

\bibitem{GordonABCGNRR13}
Andrew~D. Gordon, Mihhail Aizatulin, Johannes Borgstr{\"o}m, Guillaume Claret,
  Thore Graepel, Aditya~V. Nori, Sriram~K. Rajamani, and Claudio~V. Russo.
\newblock A model-learner pattern for bayesian reasoning.
\newblock In {\em POPL}, pages 403--416, 2013.

\bibitem{Hennessy12}
Matthew Hennessy.
\newblock Exploring probabilistic bisimulations, part {I}.
\newblock {\em Formal Asp. Comput.}, 24(4-6):749--768, 2012.

\bibitem{Howe-96}
Douglas~J. Howe.
\newblock Proving congruence of bisimulation in functional programming
  languages.
\newblock {\em Inf. Comput.}, 124(2):103--112, 1996.

\bibitem{JagadeesanPanangaden90}
Radha Jagadeesan and Prakash Panangaden.
\newblock A domain-theoretic model for a higher-order process calculus.
\newblock In {\em ICALP}, pages 181--194, 1990.

\bibitem{Plotkin}
C.~Jones and Gordon~D. Plotkin.
\newblock A probabilistic powerdomain of evaluations.
\newblock In {\em LICS}, pages 186--195, 1989.

\bibitem{KoutavasLS11}
V.\ Koutavas, P.\~B.\ Levy, and E.\ Sumii.
\newblock From applicative to environmental bisimulation.
\newblock {\em Electr. Notes Theor. Comput. Sci.}, 276:215--235, 2011.

\bibitem{LarsenSkou}
Kim~Guldstrand Larsen and Arne Skou.
\newblock Bisimulation through probabilistic testing.
\newblock {\em Inf. Comput.}, 94(1):1--28, 1991.

\bibitem{Las98a}
S.~B. Lassen.
\newblock {\em Relational Reasoning about Functions and Nondeterminism}.
\newblock PhD thesis, University of Aarhus, 1998.

\bibitem{LengletSS09}
Sergue\"{\i} Lenglet, Alan Schmitt, and Jean-Bernard Stefani.
\newblock Howe's method for calculi with passivation.
\newblock In {\em CONCUR}, pages 448--462, 2009.

\bibitem{Levy75}
Jean-Jacques L{\'e}vy.
\newblock An algebraic interpretation of equality in some models of the lambda
  calculus.
\newblock In C.~B\"ohm, editor, {\em Lambda Calculus and Computer Science
  Theory}, volume~37 of {\em LNCS}, pages 147--165. Springer-Verlag, 1975.

\bibitem{Longo}
Giuseppe Longo.
\newblock Set-theoretical models of lambda calculus: Theories, expansions and
  isomorphisms.
\newblock {\em Ann. Pure Appl. Logic}, 24:153--188, 1983.

\bibitem{manning1999foundations}
Christopher~D Manning and Hinrich Sch{\"u}tze.
\newblock {\em Foundations of statistical natural language processing}, volume
  999.
\newblock MIT Press, 1999.

\bibitem{Morris-68}
J.~Morris.
\newblock {\em Lambda Calculus Models of Programming Languages}.
\newblock PhD thesis, MIT, 1969.

\bibitem{Ong93}
C.-H.~Luke Ong.
\newblock Non-determinism in a functional setting.
\newblock In {\em LICS}, pages 275--286, 1993.

\bibitem{Panangaden09}
Prakash Panangaden.
\newblock {\em Labelled Markov Processes}.
\newblock Imperial College Press, 2009.

\bibitem{ParkPfenningThrun}
Sungwoo Park, Frank Pfenning, and Sebastian Thrun.
\newblock A probabilistic language based on sampling functions.
\newblock {\em ACM Trans. Program. Lang. Syst.}, 31(1), 2008.

\bibitem{pearl1988probabilistic}
Judea Pearl.
\newblock {\em Probabilistic reasoning in intelligent systems: networks of
  plausible inference}.
\newblock Morgan Kaufmann, 1988.

\bibitem{Pfeffer01}
Avi Pfeffer.
\newblock {IBAL}: A probabilistic rational programming language.
\newblock In {\em IJCAI}, pages 733--740. Morgan Kaufmann, 2001.

\bibitem{PittsSurvey}
A.~M. Pitts.
\newblock Howe's method for higher-order languages.
\newblock In D.~Sangiorgi and J.~Rutten, editors, {\em Advanced Topics in
  Bisimulation and Coinduction}, pages 197--232. Cambridge University Press,
  2011.

\bibitem{Pit97}
Andrew~M. Pitts.
\newblock Operationally-based theories of program equivalence.
\newblock In {\em Semantics and Logics of Computation}, pages 241--298.
  Cambridge University Press, 1997.

\bibitem{RamseyPfeffer}
Norman Ramsey and Avi Pfeffer.
\newblock Stochastic lambda calculus and monads of probability distributions.
\newblock In {\em POPL}, pages 154--165, 2002.

\bibitem{Djahromi78}
N.~Saheb-Djahromi.
\newblock Probabilistic {LCF}.
\newblock In {\em MFCS}, volume~64 of {\em LNCS}, pages 442--451, 1978.

\bibitem{Sands98}
David Sands.
\newblock From {SOS} rules to proof principles: An operational metatheory for
  functional languages.
\newblock In {\em POPL}, pages 428--441, 1997.

\bibitem{San94sce}
D.~Sangiorgi.
\newblock The lazy lambda calculus in a concurrency scenario.
\newblock {\em Inf.\ and Comp.}, 111(1):120--153, 1994.

\bibitem{SangiorgiKS07fsen}
Davide Sangiorgi, Naoki Kobayashi, and Eijiro Sumii.
\newblock Logical bisimulations and functional languages.
\newblock In {\em FSEN}, volume 4767 of {\em LNCS}, pages 364--379, 2007.

\bibitem{SangiorgiKS11}
Davide Sangiorgi, Naoki Kobayashi, and Eijiro Sumii.
\newblock Environmental bisimulations for higher-order languages.
\newblock {\em ACM Trans. Program. Lang. Syst.}, 33(1):5, 2011.

\bibitem{SaWabook}
Davide Sangiorgi and David Walker.
\newblock {\em The pi-Calculus -- a theory of mobile processes}.
\newblock Cambridge University Press, 2001.

\bibitem{Sieber93}
Kurt Sieber.
\newblock Call-by-value and nondeterminism.
\newblock In {\em TLCA}, volume 664 of {\em LNCS}, pages 376--390, 1993.

\bibitem{thrun2002robotic}
Sebastian Thrun.
\newblock Robotic mapping: A survey.
\newblock {\em Exploring artificial intelligence in the new millennium}, pages
  1--35, 2002.

\end{thebibliography}

\end{document}